\RequirePackage[l2tabu,orthodox]{nag}
\documentclass
[11pt,letterpaper]
{article} 

\usepackage[notes=true,later=false,camera=true]{dtrt}
\usepackage[utf8]{inputenc}
\usepackage{etex}
\usepackage{ stmaryrd }
\usepackage{xspace,enumerate}
\usepackage[T1]{fontenc}
\usepackage[full]{textcomp}
\usepackage[american]{babel}
\usepackage{mathtools}
\usepackage{setspace}

\usepackage{amsthm}
\usepackage{thmtools}
\usepackage{thm-restate}
\usepackage{ wasysym }
\usepackage{empheq}

   \usepackage{hyperref}
   \hypersetup{hyperindex=true,pdfpagemode=UseOutlines,bookmarksnumbered=true,bookmarksopen=true,bookmarksopenlevel=2,pdfstartview=FitH,pdfborder={0 0 1},linkbordercolor=dt@linkcolor,citebordercolor=dt@linkcolor,urlbordercolor=dt@linkcolor}
\usepackage[capitalise,nameinlink]{cleveref}
\crefname{lemma}{Lemma}{Lemmas}
\crefname{fact}{Fact}{Facts}
\crefname{figure}{Figure}{Figures}
\newcommand{\colorconstraints}{\text{Color Constraints}}
\crefname{colorconstraints}{(color constraints)}{Color Constraints}
\crefformat{colorconstraints}{#2\colorconstraints#3}
\crefname{indsetconstraints}{(indset constraints)}{IndSet Constraints}
\crefformat{indsetconstraints}{#2$\mathsf{IndSet\ Axioms}$#3}
\crefname{theorem}{Theorem}{Theorems}
\crefname{mtheorem}{Theorem}{Theorems}

\crefname{corollary}{Corollary}{Corollaries}
\crefname{claim}{Claim}{Claims}
\crefname{example}{Example}{Examples}
\crefname{algorithm}{Algorithm}{Algorithms}
\crefname{problem}{Problem}{Problems}
\crefname{definition}{Definition}{Definitions}
\usepackage{paralist}
\usepackage{turnstile}
\usepackage{mdframed}
\usepackage{tikz}
\usepackage{caption}
\DeclareCaptionType{Algorithm}
\usepackage{newfloat}
\newtheorem{theorem}{Theorem}[section]
\newtheorem{mtheorem}{Theorem}
\newtheorem*{theorem*}{Theorem}

\newtheorem{proposition}[theorem]{Proposition}
\newtheorem*{proposition*}{Proposition}
\newtheorem{lemma}[theorem]{Lemma}
\newtheorem*{lemma*}{Lemma}
\newtheorem{corollary}[theorem]{Corollary}
\newtheorem*{conjecture*}{Conjecture}
\newtheorem{fact}[theorem]{Fact}
\newtheorem*{fact*}{Fact}

\newtheorem*{hypothesis*}{Hypothesis}
\newtheorem{conjecture}[theorem]{Conjecture}
\theoremstyle{definition}
\newtheorem{definition}[theorem]{Definition}
\newtheorem*{definition*}{Definition}

\newtheorem{example}[theorem]{Example}

\newtheorem{algorithm}[theorem]{Algorithm}

\theoremstyle{remark}
\newtheorem{claim}[theorem]{Claim}
\newtheorem*{claim*}{Claim}
\newtheorem{remark}[theorem]{Remark}
\newtheorem*{remark*}{Remark}

\newtheorem*{observation*}{Observation}

\usepackage[margin=1in]{geometry}
\usepackage{newpxtext} 
\usepackage{textcomp} 
\usepackage[varg,bigdelims]{newpxmath}
\usepackage[scr=rsfso]{mathalfa}
\usepackage{bm} 
\let\mathbb\varmathbb
\usepackage{microtype}
\usepackage[capitalise,nameinlink]{cleveref}
\crefname{lemma}{Lemma}{Lemmas}
\crefname{fact}{Fact}{Facts}
\crefname{theorem}{Theorem}{Theorems}
\crefname{corollary}{Corollary}{Corollaries}
\crefname{claim}{Claim}{Claims}
\crefname{example}{Example}{Examples}
\crefname{algorithm}{Algorithm}{Algorithms}
\crefname{problem}{Problem}{Problems}
\crefname{definition}{Definition}{Definitions}
\usepackage{paralist}
\usepackage{turnstile}
\usepackage{mdframed}
\usepackage{tikz}
\usepackage{caption}
\usepackage{newfloat}

\newcommand{\Authornotecolored}[3]{}
\newcommand{\Authorcomment}[2]{}
\newcommand{\Authorfnote}[2]{}

\definecolor{forestgreen(traditional)}{rgb}{0.0, 0.27, 0.13}

\usepackage{boxedminipage}
\newcommand{\paren}[1]{(#1)}
\newcommand{\Paren}[1]{\left(#1\right)}

\newcommand{\abs}[1]{\lvert#1\rvert}
\newcommand{\Abs}[1]{\left\lvert#1\right\rvert}

\newcommand{\set}[1]{\{#1\}}
\newcommand{\Set}[1]{\left\{#1\right\}}

\newcommand{\norm}[1]{\lVert#1\rVert}
\newcommand{\Norm}[1]{\left\lVert#1\right\rVert}

\newcommand{\Esymb}{\mathbb{E}}
\newcommand{\Psymb}{\mathbb{P}}

\DeclareMathOperator*{\E}{\Esymb}

\DeclareMathOperator*{\ProbOp}{\Psymb}
\renewcommand{\Pr}{\ProbOp}

\newcommand{\Mid}{\;\middle\vert\;}

\newcommand{\mper}{\,.}
\newcommand{\mcom}{\,,}
\newcommand\bdot\bullet

\DeclareMathOperator{\val}{val}

\DeclareMathOperator{\poly}{poly}

\DeclareMathOperator{\polylog}{polylog}

\newcommand{\N}{\mathbb N}
\newcommand{\R}{\mathbb R}

\newcommand{\kikuchi}{Kikuchi }

\newcommand{\cB}{\mathcal B}

\newcommand{\cD}{\mathcal D}

\newcommand{\cF}{\mathcal F}
\newcommand{\cG}{\mathcal G}
\newcommand{\cH}{\mathcal H}

\newcommand{\cJ}{\mathcal J}
\newcommand{\cK}{\mathcal K}

\newcommand{\cR}{\mathcal R}

\renewcommand{\S}{\mathbb{S}}

\renewcommand{\leq}{\leqslant}
\renewcommand{\le}{\leqslant}
\renewcommand{\geq}{\geqslant}

\let\epsilon=\varepsilon
\numberwithin{equation}{section}
\newcommand\MYcurrentlabel{xxx}
\newcommand{\MYstore}[2]{%
  \global\expandafter \def \csname MYMEMORY #1 \endcsname{#2}%
}
\newcommand{\MYload}[1]{%
  \csname MYMEMORY #1 \endcsname%
}
\newcommand{\MYnewlabel}[1]{%
  \renewcommand\MYcurrentlabel{#1}%
  \MYoldlabel{#1}%
}
\newcommand{\MYdummylabel}[1]{}
\newcommand{\torestate}[1]{%
  \let\MYoldlabel\label%
  \let\label\MYnewlabel%
  #1%
  \MYstore{\MYcurrentlabel}{#1}%
  \let\label\MYoldlabel%
}
\newcommand{\restatetheorem}[1]{%
  \let\MYoldlabel\label
  \let\label\MYdummylabel
  \begin{theorem*}[Restatement of \cref{#1}]
    \MYload{#1}
  \end{theorem*}
  \let\label\MYoldlabel
}
\newcommand{\restatelemma}[1]{%
  \let\MYoldlabel\label
  \let\label\MYdummylabel
  \begin{lemma*}[Restatement of \cref{#1}]
    \MYload{#1}
  \end{lemma*}
  \let\label\MYoldlabel
}
\newcommand{\restateprop}[1]{%
  \let\MYoldlabel\label
  \let\label\MYdummylabel
  \begin{proposition*}[Restatement of \cref{#1}]
    \MYload{#1}
  \end{proposition*}
  \let\label\MYoldlabel
}
\makeatletter
\newcommand{\subalign}[1]{%
  \vcenter{%
    \Let@ \restore@math@cr \default@tag
    \baselineskip\fontdimen10 \scriptfont\tw@
    \advance\baselineskip\fontdimen12 \scriptfont\tw@
    \lineskip\thr@@\fontdimen8 \scriptfont\thr@@
    \lineskiplimit\lineskip
    \ialign{\hfil$\m@th\scriptstyle##$&$\m@th\scriptstyle{}##$\hfil\crcr
      #1\crcr
    }%
  }%
}
\makeatother

\newcommand{\restatefact}[1]{%
  \let\MYoldlabel\label
  \let\label\MYdummylabel
  \begin{fact*}[Restatement of \prettyref{#1}]
    \MYload{#1}
  \end{fact*}
  \let\label\MYoldlabel
}
\newcommand{\restate}[1]{%
  \let\MYoldlabel\label
  \let\label\MYdummylabel
  \MYload{#1}
  \let\label\MYoldlabel
}

\newcommand{\eps}{\epsilon}

\allowdisplaybreaks
\newcommand*{\Id}{\mathrm{Id}}

\newcommand*{\tr}{\mathrm{tr}}
\newcommand*{\zo}{\set{0,1}}

\newcommand*{\on}{\{\pm 1\}}

\newcommand*{\U}{\mathcal{U}}

\DeclareMathOperator{\pE}{\tilde{\mathbb{E}}}

\newcommand{\1}{\bm{1}}

\def\var#1{\mbox{\bf Var}[ #1 ]}

\def\floor#1{\left\lfloor #1 \right\rfloor}
\def\ceil#1{\left\lceil #1 \right\rceil}

\def\abs#1{\left|#1  \right|}

\def\norm#1{\left\| #1 \right\|}
\def\smallnorm#1{\| #1 \|}

\def\var#1{\mbox{\bf Var}[ #1 ]}

\newcommand{\setQ}{Q}

\newcommand*{\threefrac}[3]{%
  \ensuremath{%
    \vcenter{%
      \halign{\hfil$\,##\,$\hfil\cr
        \scriptstyle{#1}\cr
        \noalign{\kern\threefracLineSep}%
        \hline
        \noalign{\kern\threefracLineSep}%
        \scriptstyle{#2}\cr
        \noalign{\kern\threefracLineSep}%
        \hline
        \noalign{\kern\threefracLineSep}%
        \scriptstyle{#3}\cr
      }%
    }%
  }%
}

\newcommand{\peter}[1]{\dtcolornote[Peter]{blue}{#1}}

\allowdisplaybreaks
\newcommand{\FormatAuthor}[3]{
\begin{tabular}{c}
#1 \\ {\small\texttt{#2}} \\ {\small #3}
\end{tabular}
}

\newcommand{\keywords}[1]{\bigskip\par\noindent{\footnotesize\textbf{Keywords\/}: #1}}

 \newcommand{\boolnorm}[1]{{\lVert #1 \rVert}_{\infty \to 1}}

 \newcommand{\Fits}{\{-1,1\}}

\renewcommand{\pE}{\widetilde{\mathbb E}}

\newcommand{\algval}{\mathrm{alg}\text{-}\mathrm{val}}

\renewcommand{\star}{\odot}
\renewcommand{\Xi}{\xi}
\let\svthefootnote\thefootnote
\newcommand\blfootnote[1]{%
  \let\thefootnote\relax%
  \footnotetext{#1}%
  \let\thefootnote\svthefootnote%
}

\begin{document}


\title{Algorithms and Certificates for Boolean CSP Refutation: ``Smoothed is no harder than Random''}

\author{
\begin{tabular}[h!]{ccc}
    \FormatAuthor{Venkatesan Guruswami\thanks{Supported in part by NSF grants CCF-CCF-2228287 and CCF-2211972 and a Simons Investigator award.}}{venkatg@berkeley.edu}{UC Berkeley}
          \FormatAuthor{Pravesh K.\ Kothari\thanks{Supported in part by an NSF CAREER Award \#2047933, a Google Research Scholar Award, and a Sloan Fellowship.}}{praveshk@cs.cmu.edu}{Carnegie Mellon University}
                \FormatAuthor{Peter Manohar\thanks{Supported in part by an ARCS Scholarship, NSF Graduate Research Fellowship (under grant numbers DGE1745016 and DGE2140739), and NSF CCF-1814603.}}{pmanohar@cs.cmu.edu}{Carnegie Mellon University}
\end{tabular}
} 




\maketitle\blfootnote{Any opinions, findings, and conclusions or recommendations expressed in this material are those of the author(s) and do not necessarily reflect the views of the National Science Foundation.}
	\thispagestyle{empty}

\begin{abstract}
We present an algorithm for strongly refuting \emph{smoothed} instances of all Boolean CSPs. The smoothed model is a hybrid between worst and average-case input models, where the input is an arbitrary instance of the CSP with only the negation patterns of the literals re-randomized with some small probability. For an $n$-variable smoothed instance of a $k$-arity CSP, our algorithm runs in $n^{O(\ell)}$ time, and succeeds with high probability in bounding the optimum fraction of satisfiable constraints away from $1$, provided that the number of constraints is at least $\tilde{O}(n) (\frac{n}{\ell})^{\frac{k}{2} - 1}$. This matches, up to polylogarithmic factors in $n$, the trade-off between running time and the number of constraints of the state-of-the-art algorithms for 
refuting \emph{fully random} instances of CSPs \cite{DBLP:conf/stoc/RaghavendraRS17}.

\smallskip

We also make a surprising connection between the analysis of our refutation algorithm in the significantly ``randomness starved'' setting of semi-random $k$-XOR and the existence of even covers in \emph{worst-case} hypergraphs. We use this connection to positively resolve Feige's 2008 conjecture -- an extremal combinatorics conjecture on the existence of even covers in sufficiently dense hypergraphs that generalizes the well-known Moore bound for the girth of graphs. As a corollary, we show that polynomial-size refutation witnesses exist for arbitrary smoothed CSP instances with number of constraints a polynomial factor below the ``spectral threshold'' of $n^{k/2}$, extending the celebrated result for random 3-SAT of Feige, Kim and Ofek \cite{FKO06}.
\end{abstract}

	\clearpage
\setcounter{tocdepth}{2}
{\small \tableofcontents}
	\thispagestyle{empty}
	
	\clearpage
	
	\pagestyle{plain}
	\setcounter{page}{1}



\section{Introduction}
\label{sec:intro}

Worst-case complexity theory paints a grim picture for solving Constraint Satisfaction Problems~(CSPs). For a large class~\cite{Cha13,MR2683587-Moshkovitz10} of Max CSPs with $k$-ary Boolean predicates ($k$-CSPs), the Exponential Time Hypothesis (ETH)~\cite{DBLP:journals/jcss/ImpagliazzoP01} implies that for sparse  instances, i.e., with $m =O(n)$ constraints in $n$ variables, there is no sub-exponential time approximation algorithm that beats simply returning a random assignment. While fully-dense instances (i.e., $m \geq O(n^{k})$) admit~\cite{AroraKK95} a polynomial time approximation scheme (PTAS), ETH implies that lowering $m$ to just $\sim n^{k-1}$ makes the problem APX-hard~\cite{DBLP:conf/stacs/FotakisLP16} even for sub-exponential time algorithms. In fact, for instances with $m \leq O(n^{k-1})$, we suspect that even efficiently verifiable \emph{certificates} of non-vacuous upper bounds on the value, i.e., max fraction of constraints satisfiable, do not exist. 

The study of \emph{random} CSPs, on the other hand, offers a stark contrast. Max $k$-CSPs with any strictly super-linear number of, say, $m \geq n^{1.1}$ randomly generated constraints\footnote{i.e., uniformly random and independently chosen variables and ``literal patterns'' in each constraint.} admit~\cite{BM16,AOW15,DBLP:conf/stoc/RaghavendraRS17} sub-exponential time \emph{tight refutation}\footnote{Such algorithms correctly \emph{certify} an upper bound on the value within an arbitrarily small additive $\epsilon$ w.h.p.} algorithms. These are based on \emph{spectral methods} that exploit problem structure in non-trivial ways. Further, when $m \sim \tilde{O}(n^{k/2}) \ll n^{k-1}$, such algorithms in fact yield a PTAS for certifying the value of the input instance correctly. In fact, a considerably more fine-grained, predicate-specific and likely sharp picture~\cite{BCK15,DBLP:conf/stoc/KothariMOW17} of the trade-off between running time and number of constraints has emerged in the last decade. Adding to this rich theory is the fascinating work of~\cite{FKO06} that shows that random CSPs admit polynomial-time verifiable certificates of non-trivial upper bounds on the value even when $m \sim n^{k/2-\delta_k}$ -- i.e., when number of constraints are polynomially smaller than the threshold for efficient refutation.

How does the complexity landscape of CSPs -- for both algorithms and certificates --  interpolate between these two extremes? Is the worst-case understanding too pessimistic? Is the average-case understanding too idealistic? And are the sophisticated algorithmic tools and the structural properties that govern their success for random CSPs relevant to more general instances? 

\medskip\noindent\textbf{Refutation algorithms in the smoothed model.}
To formally study these questions, in 2007, Feige~\cite{Fei07} introduced a natural ``hybrid'' model in between worst-case and random instances (in the spirit of the pioneering work of Spielman and Teng~\cite{MR2078601-Spielman03}). In this \emph{smoothed model}, an instance is generated by starting from an arbitrary (i.e., worst-case) instance, and then negating each literal in each clause independently with some small, constant probability. In contrast to random CSPs where the clause structure (i.e., $k$-tuples describing the constraints) and the literal patterns (i.e., which variables are negated in a constraint) are chosen uniformly at random and independently, the clause structure in smoothed CSPs is completely arbitrary (i.e., worst-case) and only a small constant fraction of the literal patterns are random. In~\cite{Fei07}, Feige combined semidefinite programming with a new combinatorial certificate based on a natural notion of cycles in hypergraphs, and proved that polynomial algorithms succeed in weakly refuting (i.e., certifying a $1-o_n(1)$ upper bound on value, \cref{def:ref-algos}) \emph{smoothed} $3$-SAT formulas with $m \geq \tilde{O}(n^{1.5})$ constraints. 

Feige's techniques, however, appear fundamentally limited to weak refutation and specialized to $3$-CSPs. As a result, there is no known strong refutation algorithm (i.e., certifying a $1-\Omega(1)$ upper bound on value) for smoothed instances of 3-SAT and no known (even weak) refutation algorithm for smoothed instances of any nontrivial $4$-CSP.

In this work, we develop new techniques that yield strong refutation algorithms for all smoothed Boolean $k$-CSPs with (a possibly sharp) trade-off between running time and number of constraints matching that of fully random $k$-CSPs~\cite{DBLP:conf/stoc/RaghavendraRS17}, up to polylogarithmic factors. In particular, our results show that the algorithmic task of strong refutation in the significantly ``randomness starved'' setting of smoothed instances is no harder than in a fully random instance. 

\medskip\noindent\textbf{Refutation witnesses below spectral threshold: Feige's conjecture.} The work~\cite{FKO06} (and extensions~\cite{Witmer17}), prove that there are efficiently verifiable witnesses of unsatisfiability for \emph{fully random} $k$-CSPs with $n^{\frac{k}{2}-\delta_k}$ constraints for some constant $\delta_k>0$; when $k = 3$, this threshold is $n^{1.4}$. These witnesses are based on certain natural analogs of cycles in hypergraphs called \emph{even covers}. In an effort to understand if such witnesses exist in more general instances, Feige~\cite{MR2484644}  conjectured a trade-off between number of constraints and size of a smallest even cover. This conjecture formally generalizes the Moore bound~\cite{AHL02} on girth of graphs to hypergraphs. 

In this work, we prove Feige's conjecture by a new \emph{spectral double counting} argument that relates sub-exponential time smoothed refutation algorithms and the existence of even covers in hypergraphs. As a consequence, we derive that there are efficiently verifiable witnesses of unsatisfiability for smoothed instances of all $k$-CSPs with $m \sim n^{k/2 -\delta_k}$ constraints, for some constant $\delta_k$, which is polynomially smaller than the threshold at which efficient refutation algorithms exist even for random $k$-CSPs. 

\medskip\noindent\textbf{Summary.} Taken together, our main results can be interpreted as suggesting that the worst-case picture of complexity of CSPs arises entirely because of \emph{islands of pathology}: most instances ``around'' the worst-case hard ones are in fact essentially as easy as random, for both refutation algorithms as well as existence of refutation witnesses. 
Further, in a precise sense, the difficulty of worst-case instances can be attributed to the worst-case literal patterns, rather than the clause structure.

Our contribution is shown visually in \cref{fig:plot}. \cref{fig:plot} plots the time vs.\ \# constraints trade-off for refuting random and smoothed $3$-SAT instances (along with the analogous trade-off for approximation schemes for worst case instances). Our contribution is the smoothed case (blue line), which shows that smoothed $3$-SAT instances can be refuted with the same trade-off as random ones (green line). We also show that there exist efficiently verifiable refutation witnesses for smoothed instances at $n^{1.4}$ constraints (purple line), matching the result for random instances due to \cite{FKO06}.

\begin{figure}
\centering
\includegraphics[width=\textwidth]{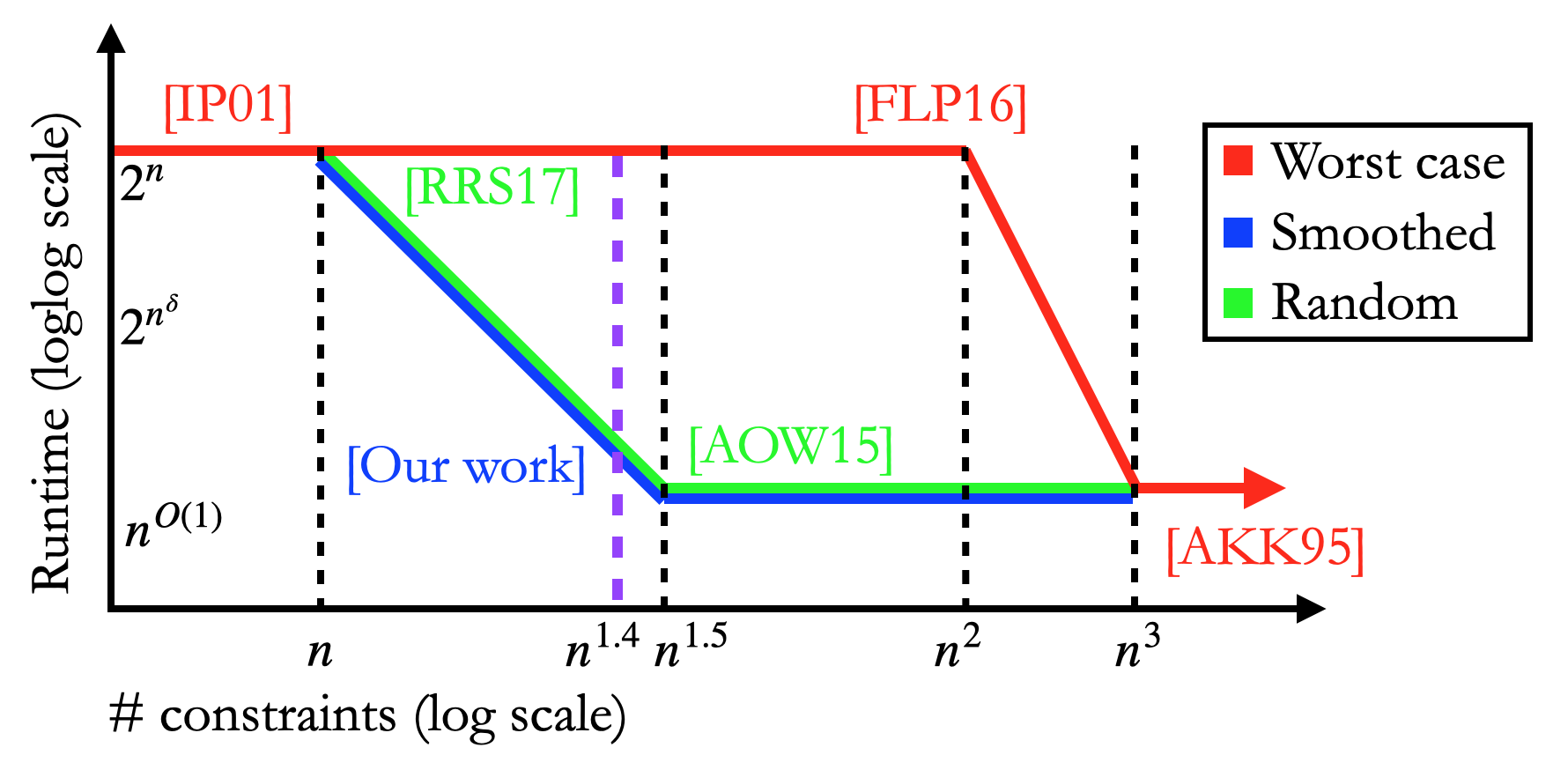}
\caption{Time vs.\ \# constraints trade-off for refuting random and smoothed $3$-SAT instances, and for approximation schemes for worst-case instances. The smoothed case is our contribution. We also prove that refutation witnesses exist for smoothed instances at the purple line, i.e., $n^{1.4}$ constraints.}
\label{fig:plot}
\end{figure}

\subsection{Our results} 
We now discuss our results on algorithms and certificates, as well as the interconnected techniques and insights that go into them. Let us recall the standard notation to talk about CSPs.
\begin{definition}[$k$-ary Boolean CSPs, random, semirandom, and smoothed instances]
\label{def:csp}
A CSP instance $\phi$ on $n$ variables with a $k$-ary predicate $P:\on^k \rightarrow \zo$ is a set of $m$ constraints on $n$ variables $x_1, \dots, x_n$ taking values in $\Fits^n$ of the form $P(\Xi(C)_1 x_{C_1}, \Xi(C)_2 x_{C_2}, \ldots, \Xi(C)_k x_{C_k})=1$. Here, $C = (C_1, C_2, \dots, C_k)$ ranges over a collection $\cH$ of \emph{scopes} (a.k.a.\ clause structure) of $k$-tuples of $n$ variables such that $C_i \neq C_j$ for any $i,j$\peter{why do we need this...?} and $\Xi:\cH \rightarrow \on^k$ are ``literal negation patterns'' one for each $C$ in $\cH$. The \emph{value} of $\phi$, $\val(\phi)$, is the maximum fraction of constraints satisfied by any assignment to the $n$ variables. 

In a \emph{random} (sometimes, \emph{fully random} in order to disambiguate from related models) instance, $\cH$ is a collection of $m$ uniformly random and independently chosen $k$-tuples and the $\Xi(C)$'s are chosen uniformly at random and independently from $\on^k$ for each $C$. 

In a \emph{semirandom} instance, $\cH$ is arbitrary (i.e., worst-case) and $\Xi(C) \in \on^k$ are uniformly at random and independent for each $C$. 

In a \emph{smoothed} instance, $\cH$ is arbitrary (i.e., worst-case) and $\Xi(C) \in \on^k$ are obtained by starting with arbitrary (i.e., worst-case) $\Xi'(C) \in \on^k$ for each $C$ and then for each $C,i$, setting $\Xi(C)_i = \Xi'(C)_i$ with probability $0.99$ and $\Xi(C)_i = -\Xi'(C)_i$ with probability $0.01$, independently. 
\end{definition}
We note that the semirandom model is more general than the random model, and the smoothed model is more general than the semirandom model.
 \begin{definition}[Weak, Strong and Tight refutation algorithms] \label{def:ref-algos}
 A refutation algorithm takes as input a CSP instance $\phi$ and outputs a value $\algval(\phi) \in [0,1]$ with $\algval(\phi) \geq \val(\phi)$ for all $\phi$. For a distribution $\cD$ over $\phi$, we say that the refutation algorithm \emph{weakly refutes} instances drawn from $\cD$ if with high probability over $\phi \sim \cD$, $\algval(\phi) < 1$. We also define \emph{strong refutation} ($\algval(\phi) < 1-\delta$ for some absolute constant $\delta >0$) and \emph{$\eps$-tight refutation} ($\algval(\phi) < \val(\phi)+\eps$, where $\eps$ is a parameter of the algorithm that can be made arbitrarily small) analogously.
 \end{definition}

\subsubsection{Algorithms for smoothed refutation} 
Our first main result gives a (possibly sharp) trade-off between running time and number of constraints for strongly refuting \emph{smoothed} CSP instances. 
\begin{mtheorem}[Smoothed refutation, informal \cref{thm:smoothed-refutation}] \label{thm:smoothed-main-intro}
For every $\ell = \ell(n)$, there is a $n^{O(\ell)}$-time strong refutation algorithm for smoothed CSPs with $m \geq m_0 = \tilde{O}(n) \cdot \Paren{\frac{n}{\ell}}^{(\frac{t}{2}-1)}$ constraints. That is, for any CSP instance $\phi$ with $m \geq m_0$ constraints, with probability $0.99$ over the smoothing $\phi_s$ of $\phi$, the algorithm outputs $\algval(\phi_s) \leq 1 - \delta$ for some absolute constant $\delta > 0$.

Here, $t = t(P)\leq k$ is the ``degree of uniformity'' of $P$ -- the smallest integer $t\leq k$ such that there is no $t$-wise uniform distribution (\cref{def:t-wise-uniform}) on $\on^k$ supported entirely on the satisfying assignments $P^{-1}(1) \subseteq \on^k$.
\end{mtheorem}
In order to understand the trade-off described by the theorem, let us apply it to two examples.
\begin{example}
For $k$-SAT, $P$ is the Boolean OR function. We thus have $t(P)=k$, as the uniform distribution on odd-parity strings is supported on $P^{-1}(1)$ and is $(k-1)$-wise uniform. Our result gives a polynomial time algorithm to strongly refute smoothed instances of $k$-SAT whenever the number of constraints $m \geq \tilde{O}(n^{\frac{k}{2}})$. More generally, for any $\delta >0$, in time $2^{O(n^{\delta})}$ the algorithm strongly refutes smoothed instances with $\geq \tilde{O}(n^{(1 - \delta)\frac{k}{2} + \delta})$ constraints.
\end{example}

\begin{example}
Consider the ``Hadamard predicate'' $P$ on $k = 2^{2^q -1}$ bits where $P(x)=1$ if and only if $x$ is a codeword of the truncated Hadamard code, i.e., $x$ is a truth table of a linear function, excluding the all $0$'s function. Hadamard CSPs naturally appear in the design of query efficient PCPs. Here, $t(P)=3 \ll k$, so our theorem gives a polynomial-time algorithm to strongly refute smoothed instances of the Hadamard CSP with at least $\tilde{O}(n^{1.5})$ constraints, and a $2^{n^{\delta}}$-time algorithm for instances with at least $\tilde{O}(n^{1.5-\delta/2})$ constraints $\forall \delta \in (0,1]$. 
\end{example}

\smallskip\noindent\textbf{Comparison with prior results.}
\cref{thm:smoothed-main-intro} can be directly compared to works on refuting random, semirandom and smoothed (in the order of increasing generality) CSPs. 

Building on~\cite{AOW15,BM16}, Raghavendra, Rao and Schramm~\cite{DBLP:conf/stoc/RaghavendraRS17} proved the same trade-off  (up to a $\polylog(n)$ factor in $m$) between running time and number of constraints required as in \cref{thm:smoothed-main-intro} for the significantly simpler special case of \emph{fully random} CSPs -- when the clause structure and the literal patterns are chosen uniformly at random from the respective domains. Our result shows that the same trade-off holds for \emph{smoothed} instances -- i.e., with worst-case clause structure and small random perturbations of worst-case literal patterns. All known efficient refutation algorithms, including ours and that of~\cite{DBLP:conf/stoc/RaghavendraRS17}, can in hindsight be interpreted as an analysis of the canonical sum-of-squares (SoS) relaxation (\cref{sec:sum-of-squares}) for the max $k$-CSP problem. For random CSPs (and thus also for the more general smoothed instances we study) the trade-off we obtain is known to be essentially tight~\cite{DBLP:conf/stoc/KothariMOW17,BCK15} for such ``SoS-encapsulated'' algorithms: this fact is often taken as evidence of sharpness of this trade-off.  

Much less is known about refuting CSPs in the more general \emph{semirandom}  and \emph{smoothed} models. Feige~\cite{Fei07} gave a \emph{weak} refutation algorithm for refuting smoothed and semirandom instances of $3$-SAT. His techniques apply to all $3$-CSPs but do not seem to extend to either strong refutation or $4$-CSPs. More recently, in a direct precursor to this work, Abascal, Guruswami and Kothari~\cite{AbascalGK21} gave a polynomial time algorithm for refuting \emph{semirandom} instances of all CSPs -- thus obtaining one of the extreme points (corresponding to $\ell=O(1)$) in the trade-off in \cref{thm:smoothed-main-intro} above. \cref{thm:smoothed-main-intro} relies on a key idea from their work (row bucketing) along with several new ideas discussed below.

\medskip\noindent\textbf{Algorithms for refuting \emph{semirandom} $k$-XOR.} Our main technical result is an algorithm for \emph{tight} refutation of \emph{semirandom} instances of $k$-XOR. \cref{thm:smoothed-main-intro} then follows by a simple  blackbox reduction (see \cref{sec:smoothed}) that relies on a dual polynomial introduced in~\cite{AOW15}. For the special case of $k$-XOR, an instance $\phi$ is completely described by an arbitrary $k$-uniform instance hypergraph $\cH$ and a collection of ``right-hand sides'' $b_C \in \on$, one for each $C \in \cH$; in the notation of \cref{def:csp}, we have $b_C = \prod_{i = 1}^k \Xi(C)_i$. One can associate to $\phi$ a homogeneous degree $k$ polynomial $\phi(x)$ on the hypercube $\on^n$:
\begin{equation*}
\phi(x) = \frac{1}{m} \sum_{C \in \cH} b_C \prod_{i \in C}x_i \mper
\end{equation*}
This polynomial $\phi(x)$ computes the ``advantage over $1/2$'' of an assignment $x$. That is, the value of the associated instance is $\frac{1}{2} + \frac{1}{2}\max_{x \in \on^n} \phi(x)$. Tight refutation corresponds to certifying that $\phi(x) \leq \epsilon$ for arbitrary $\epsilon>0$.

\begin{theorem}[Tight refutation of semirandom $k$-XOR, informal  \cref{thm:mainpolyrefute}] \label{thm:main-poly-refute-intro}
For every $k \in \N$ and $\ell = \ell(n)$ and every $\epsilon >0$, there is a $n^{O(\ell)}$ time $\epsilon$-tight refutation algorithm for homogeneous degree $k$ polynomials that succeeds with probability at least $0.99$ over the draw of the coefficients i.i.d.\ uniform on $\{-1,1\}$, whenever the associated hypergraph $\cH$ has $m \geq n \Paren{\frac{n}{\ell}}^{\frac{k}{2}-1} \cdot \poly(\frac{\log n}{\epsilon})$ hyperedges.

In particular, for every $\delta > 0$, we obtain a $2^{O(n^{\delta})}$-time $\epsilon$-tight refutation algorithm for semirandom $k$-XOR instances with $m \gg \tilde{O}(n) \cdot n^{(1-\delta)(\frac{k}{2}-1)} \poly(\frac{1}{\eps})$-constraints.
\end{theorem}

\smallskip\noindent\textbf{Prior works and brief comparison of techniques.} The trade-off above (up to $\polylog(n)$ factors in $m$) matches the one obtained for refuting fully random $k$-XOR~\cite{DBLP:conf/stoc/RaghavendraRS17}. Our techniques, however, necessarily need to be  significantly different, as the analysis in~\cite{DBLP:conf/stoc/RaghavendraRS17} (and related works it built on~\cite{Coja-OghlanGL04,BM16,AOW15}) crucially rely on the randomness of the hypergraph $\cH$. In particular, the refutation in~\cite{DBLP:conf/stoc/RaghavendraRS17} uses the spectral norm of a certain ``symmetric tensor power'' of the canonical matrix obtained from the instance. They analyze this matrix using a technical tour-de-force argument using the trace moment method.\footnote{Just the technical argument in~\cite{DBLP:conf/stoc/RaghavendraRS17} runs over 20 pages!} A couple of follow-up works have attempted to simplify the analyses in~\cite{DBLP:conf/stoc/RaghavendraRS17}. Wein, Alaoui and Moore~\cite{WeinAM19} succeeded in giving a simpler proof (introducing the \emph{Kikuchi matrix}, a variant of which is central to this work) for the case of random $k$-XOR for \emph{even} $k$, and they also suggest that a natural generalization of their Kikuchi matrix for random odd $k$ will work (their suggestion does not pan out, as we prove in \cref{sec:wam-sug-does-not-work}). In a recent work, Ahn~\cite{Ahn20} simplified some aspects of the analysis of the ``symmetric tensor power'' matrix in the analysis of~\cite{DBLP:conf/stoc/RaghavendraRS17}. To summarize, the tools in prior works on random CSPs for analyzing the spectra of relevant correlated random matrices seem to use the randomness of the hypergraph both heavily and in a rather opaque manner.

For the more general setting of semirandom $k$-XOR refutation, the best known result~\cite{AbascalGK21} obtained an extreme point in the trade-off (i.e., the case of $\ell=O(1)$). That work analyzes the $\infty \to 1$-norm of the canonical matrix associated with the CSP instance. In this special case when $\ell=O(1)$, it turns out that handling $3$-XOR instances allows deriving all larger $k$ as a corollary. For the case of $3$-XOR, their analysis relies on a new \emph{row bucketing} step according to the \emph{butterfly degree} of a pair of vertices (a new notion that they define), along with a certain pseudo-random vs structure decomposition for arbitrary $3$-uniform hypergraphs associated with the $3$-XOR instance.  

To prove \cref{thm:main-poly-refute-intro}, we build on~\cite{AbascalGK21} and introduce a few new tools. For even $k$, the Kikuchi matrix of~\cite{WeinAM19} analyzed using the row bucketing idea (with an appropriate generalization of the butterfly degree) of~\cite{AbascalGK21} yields a correct trade-off (see \cref{sec:random4xor,sec:semirandom4xor}). The case of odd $k$ turns out to be significantly more challenging (as has always been the case in CSP refutation) and needs new ideas. We introduce a variant of the Kikuchi matrix for this purpose. Unlike the case of even $k$ (and the algorithm in~\cite{AbascalGK21}), the spectral norm of this matrix is provably too large to yield a refutation -- even for \emph{random} instances. Indeed, this is why the strategy suggested by~\cite{WeinAM19} does not pan out, as we show in \cref{sec:wam-sug-does-not-work}.  Instead, we use the spectral norm of a matrix obtained by pruning away appropriately chosen rows. We then show that the number of pruned rows is not too large, and so does not contribute too much to the $\infty \to 1$-norm of the full matrix.

 The row pruning step motivates a definition of \emph{regularity}, a collection of natural pseudorandom properties that relate to \emph{well-spreadness} in the intersection structure of the hyperedges in the instance hypergraph.\footnote{This is closely related to the notion of spread encountered in recent work on the sunflower conjecture~\cite{AlweissLWZ20,Rao19}.} 
 We then show that the hyperedges in every $k$-uniform hypergraph can be decomposed, via a \emph{regularity decomposition} lemma, into $k'$-uniform hypergraphs for $k' \leq k$, along with some ``error'' hyperedges, such that (i) each of the $k'$-uniform hypergraphs satisfies regularity, and (ii) refuting all of these $k'$-XOR instances provides a refutation for the original instance. We explain our row pruning and the regularity decomposition steps in more detail in \cref{sec:overview}.

\subsubsection{Short refutations below spectral threshold: proving Feige's conjecture} In a one-of-a-kind result, Feige, Kim and Ofek~\cite{FKO06} (henceforth, FKO) proved that with high probability over the draw of a fully random $3$-SAT instance $\psi$, there is a polynomial size \emph{witness} that weakly refutes $\psi$ if $\psi$ has $m \sim \tilde{O}(n^{1.4})$ constraints. Formally, there is a polynomial time \emph{non-deterministic} refutation algorithm that succeeds in finding a refutation with high probability over the drawn of a fully random $3$-SAT instance with $m \sim \tilde{O}(n^{1.4})$ constraints.  On the other hand, all known polynomial time \emph{deterministic} refutation algorithms require the input random instance to have $\Omega(n^{1.5})$ constraints -- this bound is often called the \emph{spectral threshold}. The fastest known refutation algorithm~\cite{DBLP:conf/stoc/RaghavendraRS17} for instances with $\sim n^{1.4}$ constraints runs in time $2^{n^{0.2}}$, matching the SoS lower bound~\cite{DBLP:conf/stoc/KothariMOW17}. Thus, intriguingly, the FKO result shows the existence of polynomial time verifiable refutation witnesses (i.e., certificates of an upper bound of $1-o_n(1)$ on the value) at a constraint density at which there are no known $2^{n^{o(1)}}$-time refutation algorithms. Does such a ``gap'' between thresholds for existence vs efficient computability of refutation witnesses persist for semirandom and smoothed instances, i.e., instances with \emph{worst-case} constraint hypergraphs?

In 2008, Feige~\cite{MR2484644} made an elegant conjecture on the existence of even covers in sufficiently dense hypergraphs. This conjecture can be interpreted as generalizing to hypergraphs the classical Moore bound on the girth of graphs with a given number of edges. If true, Feige's conjecture implies that the FKO result holds for all semirandom and smoothed CSP instances -- in particular, the FKO result does not rely on the properties of the underlying hypergraph at all. Let us explain this conjecture below. 

\begin{definition}[Even Cover and Girth]
For a $k$-uniform hypergraph $\cH$ on $[n]$, an \emph{even cover} of length $t$ is a collection of $t$ distinct hyperedges $C_1, C_2, \ldots, C_t$ in $\cH$ such that every vertex in $[n]$ appears in an even number of $C_i$'s. The \emph{girth} of $\cH$ is the length of the smallest even cover in $\cH$.
\end{definition}

\begin{conjecture}[Feige's conjecture, Conjecture 1.2 in~\cite{MR2484644}] \label{conj:feige}
Every $k$-uniform hypergraph $\cH$ on $[n]$ with $m \geq m_0 = O(n) \Paren{\frac{n}{\ell}}^{\frac{k}{2}-1}$ hyperedges has an even cover of length $O(\ell \log n)$. 
\end{conjecture}

\medskip\noindent\textbf{A brief history of the conjecture.} For $k=2$, an even cover is a $2$-regular subgraph (and thus a union of cycles) in a graph and thus, the conjecture above reduces to the question of determining the maximum girth (the length of the smallest cycle) in a graph with $n$ vertices and $nd/2$ edges for parameter $d$. The best known bound is due to Alon, Hoory and Linial~\cite{AHL02} who proved that for every graph on $n$ vertices with $nd/2$ edges for $d>2$, there is a cycle of length at most $c \log_{d-1}n$ for $c \leq 2$. The best known lower bound on the girth is $c \log_{(d-1)}n$ for $c\geq 4/3$ by Margulis~\cite{MR939574} and Lubotzky, Philips and Sarnak~\cite{MR963118} via explicit constructions of Ramanujan graphs. Obtaining a tight bound on $c$ has been an outstanding open problem for the last 3 decades. 

Much less is known for hypergraphs. When $k$ even and $\ell=O(1)$, Naor and Verstraete~\cite{MR2399017-Naor08} proved the conjecture. They were motivated by a natural coding theory interpretation: viewing each hyperedge as describing the non-zero coefficients of linear equations over $\mathbb{F}_2$, an even cover is a \emph{sparse linear dependency} and thus, the conjecture gives the rate-distance trade-off for linear codes with column-sparse parity check matrices. In the more challenging case when $k$ is odd, the bounds for $\ell=O(1)$ case in~\cite{MR2399017-Naor08} were improved to essentially optimal ones in~\cite{MR2484644}. For $\ell \gg 1$, the best previous bound for $3$-uniform hypergraphs is due to a simple argument of Alon and Feige~\cite{MR2809335-Alon09} (Lemma 3.3), who proved that every $3$-uniform hypergraph with $\tilde{O}(n^2/\ell)$ hyperedges has an even cover of size $\ell$ (this is off by $\sim \sqrt{n}$ factor in $m$). For $3$-uniform hypergraphs with $m \gg n^{1.5+\epsilon}$ (and the case when $m \gg n^{k/2}$ in general), \cite{DBLP:journals/cpc/DellamonicaHLMNPRSV12} proved that there are even covers of size $O(1/\epsilon$). Finally, Feige and Wagner~\cite{Wagner13generalizedgirth} proved some variants (``generalized girth problems'') in order to build tools to approach this conjecture.  

To summarize, prior to this work, the conjecture was known to be true only for $\ell=O(1)$. For larger $\ell$, the only approach was the combinatorial strategy introduced in~\cite{Wagner13generalizedgirth}. In this work, we  prove Feige's conjecture (up to $\poly \log n$ slack in $m$) via a new \emph{spectral double counting argument}.  

\begin{mtheorem}[Feige's conjecture is true, informal \cref{thm:feige-conjecture}] \label{thm:fko-main-intro}
For every $k \in \N$ and $\ell = \ell(n)$, every $k$-uniform hypergraph $\cH$ with $m \geq m_0 = \tilde{O}(n) \cdot (\frac{n}{\ell})^{\frac{k}{2}-1}$ hyperedges has an even cover of size $O(\ell \log n)$. 
\end{mtheorem}

Our spectral double counting argument\footnote{Subsequent to our posting of this paper, Tim Hsieh and Sidhanth Mohanty were able to use our spectral double counting technique with the non-backtracking walk matrix of a graph to recover the sharpest known result (matching~\cite{AHL02}) for the Moore bound for irregular graphs. We believe a similar approach might also help achieve sharper results for size of smallest even covers in hypergraphs.} is heavily derived from our analysis for smoothed refutation using our Kikuchi matrices; indeed, our proof of \cref{thm:feige-conjecture} mirrors our steps in the analysis of our refutation algorithm.
In fact, in a precise sense (as we explain in \cref{sec:feigeoverview}), our approach gives a tight connection between even covers in hypergraphs and simple cycles (and in turn, the spectral norm of the corresponding adjacency matrix) in the ``Kikuchi graph'' built from the hypergraph. 

Combining with our smoothed refutation algorithms (\cref{thm:smoothed-main-intro}) we immediately obtain a generalization of the FKO result that yields a polynomial time non-deterministic refutation algorithm for smoothed instances of all $k$-ary CSPs with number of constraints $m$ polynomially below the spectral threshold of $n^{k/2}$. 

\begin{mtheorem}[Informal \cref{thm:ourfko}]
\label{thm:ourfko-main-intro}
There is a \emph{non-deterministic} polynomial time algorithm that weakly refutes smoothed instances of any $k$-CSP with $m \geq m_0 = \tilde{O}(n^{\frac{k}{2}-\frac{k-2}{2(k+8)}})$-constraints. For the special case of $k=3$, $m_0=\tilde{O}(n^{1.4})$. 
\end{mtheorem}

\section{Overview of our Techniques} \label{sec:overview}
In this section, we illustrate our key ideas by giving essentially complete proofs of some special cases of our main results along with expository comments.

This overview is structured as follows: we will first give an essentially complete proof for refuting \emph{semirandom} instances of \emph{even-arity} $k$-XOR. As has been the trend in all the refutation results, the even-arity case happens to be significantly simpler but allows us to showcase two key ideas: 

\medskip
\parhead{(1) The power of the Kikuchi matrix.} In fact, this work can be thought of as a paean to the beautiful structure and the applications of the Kikuchi matrix and its variant that we introduce for odd-arity $k$-XOR. Combined with the \emph{row bucketing} idea from~\cite{AbascalGK21}, we can easily resolve the case of even arity $k$-XOR. The Kikuchi matrix was introduced by~\cite{WeinAM19} to give a simpler proof of the result of~\cite{DBLP:conf/stoc/RaghavendraRS17} for refuting \emph{fully random} instances of \emph{even-arity} $k$-XOR. They left open the question of finding an analogous proof for the odd-arity case (again, for fully random CSPs) and even suggested an approach. Their approach, however, does not pan out, as we prove in \cref{sec:wam-sug-does-not-work}. Our Kikuchi matrix for the odd-arity case along with our analysis technique (that does not directly work with spectral norms) allows us to prove sharp trade-offs for refuting random CSPs and with additional ideas, make them work even for the significantly randomness starved semirandom and smoothed settings.  

\medskip
\parhead{(2) The connection between ``Kikuchi matrix refutations'' and even covers in hypergraphs.} In this overview, we will use this connection to give a single page proof of Feige's conjecture for $k$-hypergraphs for $k$ \emph{even}. We note that this gives an interesting instance of the phenomenon where the analysis of an algorithm in a reduced-randomness setting can be used to infer a purely combinatorial property of worst-case structures.

We will then discuss our ideas for the odd-arity case at a high-level by focusing on $3$-XOR. As is usual in CSP refutation, even for the special case of \emph{fully random} instances, refuting odd-arity XOR is significantly more challenging~\cite{Coja-OghlanGL04,BM16,AOW15}. We introduce several new ideas to tackle the semirandom (and thus also the smoothed) case: \begin{inparaenum}[(1)] \item a new, suitable variant of the Kikuchi matrix, \item the idea of \emph{row pruning} combined with \emph{row bucketing}, and \item a new \emph{regularity decomposition} for arbitrary hypergraphs\end{inparaenum}. 

Our proof of Feige's conjecture for odd-$k$-uniform hypergraphs is conceptually similar to the even case -- in that it mimics the refutation argument closely -- but needs all the new machinery for refutation introduced above for handling semirandom odd-arity $k$-XOR and must use the trace moment method (instead of the matrix Bernstein) in the step that upper bounds the spectral norm of appropriate sequence of matrices produced in our analysis. The combinatorial argument required in analyzing the trace method turns out to be somewhat more intricate in the odd arity case. We will not discuss it in this overview.  

Our reduction from smoothed CSP refutation to semirandom CSP refutation is short and elementary, and we present it in full in \cref{sec:smoothed}. We will not discuss this argument in this overview.

\subsection{Random $4$-XOR via the Kikuchi matrix of ~\cite{WeinAM19}}
\label{sec:random4xor}
Let's start by defining the Kikuchi matrix and showing how it gives a simple refutation algorithm with the optimal trade-off for random instances of even-arity $k$-XOR. We will focus on $k=4$ here. 

\begin{definition}[Kikuchi Matrix] \label{def:kikuchi-overview-even}
Let $N = {n \choose \ell}$. For a $4$-XOR instance described by $\cH$ and $b_C$'s for $C \in \cH$, we define the matrices $A_C \in \R^{N \times N}$ for each $C \in \cH$ as follows. Let $A_C \in \R^{N \times N}$ be the matrix indexed by all possible subsets of $[n]$ of size exactly $\ell$. The entry of $A_C$ at any $(S,T)$ where $S,T \in {{[n]} \choose \ell}$ is defined by:
\[
A_C(S,T) = \begin{cases} b_C  & \text{ if } S \oplus T = C\\
						0 & \text{ otherwise }
\end{cases}
\] 
Here, $S \oplus T$ is the symmetric difference of the sets $S,T$. The level $\ell$ \kikuchi matrix of the instance is then simply $A = \sum_{C \in \cH} A_C$.
\end{definition}

\medskip\noindent\textbf{Quadratic forms of the Kikuchi matrix.} The quadratic forms of this matrix are closely related to the polynomial $\phi(x)$ associated with the input $4$-XOR instance: namely, $\phi(x) := \frac{1}{m}\sum_{C \in \cH} b_C \prod_{i \in C} x_i$. Notice that the non-zero entries of the matrix $A$ correspond to pairs of sets $(S,T)$ such that the symmetric difference of $S,T$ is one of the clauses in the input $4$-XOR instance. Observe that if $S \oplus T = C$, then $|S \cap C| = 2$, $|T \cap C|=2$, and $|S \cap T| = \ell-2$. In particular, each $b_C$ appears in ${4 \choose 2} \cdot {{n-4} \choose {\ell-2}}$ different entries of $A$. Now, let $x^{\star \ell}$ be the ${n \choose \ell}$-dimensional vector of degree $\ell$ monomials in $x$. That is, the entries of $x^{\star \ell}$ are indexed by subsets of size $\ell$ of $[n]$ and the $S$-th entry of $x^{\star \ell}$ is given by $\prod_{i \in S} x_i$. Then, we must have:
\begin{equation} \label{eq:kikuchi-quadratic-overview}
{4 \choose 2} \cdot {{n-4} \choose {\ell-2}} \phi(x) = \frac{1}{m} \Paren{x^{\star \ell}}^{\top} A x^{\star \ell}
\end{equation}

This immediately provides a certificate of upper bound on the value of the input instance as it must hold that 
\begin{equation}\label{eq:kikuchi-norm-overview}
\max_{x \in \{-1,1\}^n} \phi(x) \leq \frac{1}{6m} \cdot {{n-4} \choose {\ell-2}}^{-1} {n \choose \ell} \Norm{A}_2 \leq O\Bigl(\frac{n^2}{m \ell^2}\Bigr) \cdot \Norm{A}_2\mcom
\end{equation}
where $\Norm{A}_2$ is the spectral norm of the matrix $A$. If we can show that $\Norm{A}_2 \leq \tilde{O}(\ell)$ w.h.p.\ over the draw of the hypergraph $\cH$ and the $b_C$'s, then, whenever $m \gg \tilde{O}(n) \cdot \frac{n}{\ell}$, the spectral norm of $A$ provides a certificate that $\phi(x)\leq 0.01$ for every $x \in \on^n$.

It is in the ease of establishing such an upper bound on the spectral norm that the choice of Kikuchi matrix really shines! Observe that $A_C$'s are a sequence of \emph{independent, random} matrices and thus, one can try to apply off-the-shelf matrix concentration inequalities to bound the spectral norm of $A$. Instead of using the matrix Chernoff inequality as in~\cite{WeinAM19}, we will use the matrix Bernstein inequality below as it turns out to generalize better. We also give a completely elementary trace moment based proof of the same fact (see \cref{sec:Gspecboundtrace}).  

\begin{fact}[Matrix Bernstein Inequality] \label{fact:matrix-bernstein-overview}
Let $M_1, M_2, \ldots, $ be independent random $N \times N$ matrices with mean $0$ such that $\Norm{M_i}_2 \leq R$ almost surely. 
Let $\sigma^2 = \max\{ \Norm{ \E [\sum_i M_i M_i^{\top}]}_2, \Norm{ \E [\sum_i M_i^{\top} M_i]}_2\}$ be the \emph{variance} term. 
Then, with probability at least $1-1/n^{100}$, 
\[
\Norm{\sum_i M_i}_2 \leq O(R \log N + \sigma \sqrt{ \log N })\mper
\] 
\end{fact}
\medskip\noindent\textbf{Spectral norm of the Kikuchi matrix.} Let's analyze $\Norm{A}_2$ using this inequality. First, observe that any row of $A_C$ has at most $1$ non-zero entry of magnitude $1$. Since the spectral norm of a symmetric matrix is upper bounded by the maximum $\ell_1$-norm of any of its rows, this immediately yields that $\Norm{A_C}_2 \leq 1$. Let's now compute the ``variance'' term. Here's the key observation about the Kikuchi matrix that makes this analysis so simple: the matrix $A_C^2$ is \emph{diagonal} for every $C$. To see this, observe that the entry at any $(S,T)$ of this matrix is given by $\sum_{U} A_C(S,U) A_C(U,T)$. A term in the summation is non-zero only if $S \oplus U = U \oplus T = C$ which can happen if and only if $T=S$. 

Let's now compute the diagonals of $\E \sum_{C} A_C^2$. Notice that $A_C^2(S,S)$ equals either $1$ or $0$ for every $C$.  Thus, $\sum_{C} A_C^2(S,S) = \deg(S)$ where 
\[
\deg(S) := \Abs{\Set{C \mid |S \cap C| = 2}}\mcom
\] and so the variance term $\sigma^2$ is $\max_{S} \deg(S)$. 

How large can this be? Since each constraint contributes ${4 \choose 2} \cdot {{n-4} \choose {\ell-2}}$ non-zero entries to $A$,  $\sum_{S \in {n \choose \ell}} \deg(S) = {4 \choose 2} \cdot {{n-4} \choose {\ell-2}}m$. Thus, on average $\deg(S)$ is $\approx m \ell^2/n^2$. When $m \sim n^2/\ell$, this is $\sim \ell$. 

When $\cH$ is a \emph{random hypergraph} with $\sim n^2/\ell$ hyperedges, we expect $\deg(S)$ to not deviate too much from its expectation. In fact, using the Chernoff bound yields $\deg(S) \leq O(\ell \log n)$ for all $S$ whp. Since $N = {n \choose \ell}$, this yields that $\Norm{A}_2 \leq O(\log N) + O(\sqrt{\ell \log n \cdot \log N}) = \tilde{O}(\ell)$, as desired.

\subsection{Semirandom instances of $4$-XOR via row bucketing from~\cite{AbascalGK21}} 
\label{sec:semirandom4xor}
Let us now conduct a post-mortem of the above proof to see where we used the randomness of the hypergraph $\cH$. Even after fixing $\cH$, the $A_C$'s are  independent random matrices, with all the randomness coming from the $b_C$'s. Thus, we can still apply the matrix Bernstein inequality. The only point in the proof where we used the randomness of the hypergraph $\cH$ was to establish that $\deg(S) = O(\ell \log n)$ for every $S$. So, our proof immediately extends to semirandom instances where the instance hypergraph $\cH$ is such that $\deg(S) = O(\ell \log n)$ for every $S$. 

This bound is delicate: when $\deg(S)=\Omega(\ell^{2})$, we obtain no non-trivial refutation guarantee and even $\deg(S) \sim \ell^{1.1}$ results in a suboptimal trade-off. On the other hand, in arbitrary $\cH$, $\deg(S)$ can be as large as $m$ (but no larger). Further, this is a ``real'' issue (and not an artefact of the use of Matrix Bernstein inequality): when $\deg(S)$ is large, so is the spectral norm of $A$. 

\medskip\noindent\textbf{Key observation: only sparse vectors cause large quadratic forms.} Our way forward builds on that of~\cite{AbascalGK21} who recently gave a polynomial time algorithm for (strongly) refuting semirandom instances of $k$-XOR with $\geq \tilde{O}(n^{k/2})$ constraints. The key observation is when $\deg(S)$ is large, the spectral norm of $A$ is high but intuitively, the ``offending'' large quadratic forms are induced only by ``sparse'' vectors, i.e., vectors where the $\ell_2$ norm is contributed by a small fraction of the coordinates. On the other hand, we only care about upper bounding quadratic forms of $A$ on vectors where all coordinates are $\pm 1$ and are thus are maximally ``non-sparse'' or ``flat''.

\medskip\noindent\textbf{Row bucketing.} We can formalize this observation via \emph{row bucketing}. Let $d_0 \sim m \cdot \ell^2/n^2$ be the average value of $\deg(S)$. Let's partition the row indices in ${n \choose \ell}$ into multiplicatively close buckets $\cF_0, \cF_1, \cdots, \cF_t$ so that for each $i \geq 1$, 
\[
\cF_i = \Set{ S \mid 2^{i-1} d_0 < \deg(S) \leq 2^{i} d_0}\mper
\] 
and $\cF_0 = \Set{S \mid \deg(S) \leq d_0}$.
Then, since $\deg(S)\leq m$ and $d_0 \geq 1$ (as $m \sim n^2/\ell$), we can take $t \leq \log_2 m$. Further, by Markov's inequality, $|\cF_i|\leq 2^{-i} {n \choose \ell} = 2^{-i}N$. For each $i,j \leq t$, let $A_{i,j}$ be the matrix obtained by zeroing out all rows not in $\cF_i$ and all columns not in $\cF_j$ from the Kikuchi matrix $A$. Then,  $A = \sum_{i,j \leq t} A_{i,j}$. 

The key observation is the following: while $A_{i,j}$ has non-zero rows and columns where $\deg(S)$ is larger by a $2^{i}$ ($2^j$, respectively) factor than the average, we are compensated for this by a reduction in the number of non-zero rows and columns. 

Let $y \in \R^{N}$ be any vector with entries in $\on^N$, and let $y_{\cF_i}$ be the vector obtained by zeroing out all coordinates of $y$ that are not indexed by elements of $\cF_i$. Then, by Cauchy-Schwarz, we must have: 
\begin{equation}
\max_{y \in \on^N} y^{\top} A_{i,j} y = \max_{y \in \on^N} (y_{\cF_i})^{\top}A_{i,j} (y_{\cF_j}) \leq \sqrt{|\cF_i| |\cF_j|} \cdot \Norm{A_{i,j}}_2 \mper
\end{equation} 

We apply the Matrix Bernstein inequality in a similar manner to the previous analysis. The ``variance'' term grows by a factor of $\max \{2^{i}, 2^{j}\}$ over the bound obtained for the random case. As a result, the spectral norm of $A_{i,j}$ is higher by a factor of $\max \{2^{i/2}, 2^{j/2}\}$. On the other hand, the effective $\ell_2$ norm of the vector drops by $2^{-(i+j)/2}$. The trade-off ``breaks in our favor'' and the dominating term in the bound is $A_{0,0}$ -- the spectral norm of which is at most of the same order as that of the $A$ in the case of the previous random $4$-XOR analysis! We thus obtain that $\max_{y \in \on^N} y^T A y$ is $\tilde{O}(\frac{n^2}{m \ell^2} \cdot \ell)$, and so we certify that $\phi(x) \leq 0.01$ for every $x \in \on^n$.

\subsection{Proving Feige's conjecture for $4$-uniform hypergraphs}
\label{sec:feigeoverview}
We now discuss how the analyses of the Kikuchi matrix from the previous section relates to Feige's conjecture on even covers in $4$-uniform (and in general, any even-uniform) hypergraphs. A priori, such a connection may appear rather surprising that the analysis of a super-polynomial size matrix introduced for refuting $k$-XOR can shed light on a purely combinatorial fact. But we will soon see that this is yet another instance of the Kikuchi matrix doing its magic. 

Recall that Feige's conjecture suggests a trade-off between the number of hyperedges and an appropriate notion of \emph{girth} (i.e., length of the smallest cycle, or \emph{even cover}) in hypergraphs that generalizes the classical Moore bound~\cite{AHL02}, which asserts that every graph on $n$ vertices with $nd/2$ edges has a cycle of length $\leq 2 \log_{d-1}(n)$. To explain our \emph{spectral double counting} argument to prove this conjecture, it is helpful to first use it to prove a (significantly weaker) version of the Moore bound and then generalize to hypergraphs $H$ via the ``Kikuchi graph'' derived from $H$.

\begin{proposition}[Weak Moore bound in irregular graphs] \label{prop:moore-bound}
Every graph $G$ on $n$ vertices and $nd/2$ edges for $d \geq O(\log_2^3 (n))$ has a cycle of length $\leq 2\lceil \log_2 n \rceil$. 
\end{proposition} 

Our \emph{spectral double counting} argument counts the number of edges of $G$ in two different ways: let $A$ be the $0$-$1$ adjacency matrix of $G$. Then, we have $\1^{\top} A \1 = nd$. We will show that if $G$ does not have a cycle of size $\leq 2\lceil \log_2 n \rceil$, then, all $\pm 1$-coordinate quadratic forms of $A$ are at most $n \cdot \tilde{O}(\sqrt{d})$. Together, these two bounds yields the desired contradiction. 

\begin{claim}[Trace Method in the absence of even covers]
Let $A$ be the $0$-$1$ adjacency matrix of a graph $G$ on $n$ vertices with $nd/2$ edges with no cycle of length $\leq 2r$ for $r = \lceil \log_2 n \rceil$. Then, for every $y \in \on^n$, 
\[
y^{\top}Ay \leq n \sqrt{d} \cdot O(\log_2^{1.5}(n))\mper
\]
\end{claim} 

Notice that this claim immediately yields a contradiction if $nd > n \sqrt{d} \cdot O(\log_2^{1.5}(n))$, which holds if $d \geq O(\log_2^{3} n)$, thus proving \cref{prop:moore-bound}. Let's now see how to prove this claim.
\begin{proof}

The average degree of vertices in $G$ is $d$. Let $\cF_i = \{v \mid 2^{i}d \leq \deg(v) \leq 2^{i+1} d\}$ for each $1 \leq i \leq \log_{2} n$. Let $A_{i,j}$ be obtained by zeroing out all rows not in $\cF_i$ and all columns not in $\cF_j$ from $A$. Then, $A = \sum_{i,j} A_{i,j}$. 

By a similar observation as in the previous subsection, we have:
\begin{equation} \label{eq:partitioned-bound}
y^{\top} Ay \leq \sum_{i,j} \sqrt{|\cF_i||\cF_j|} \Norm{A_{i,j}}_2\mper
\end{equation}

Let's now bound $\Norm{A_{i,j}}_2$. The idea is to use the trace moment method on the matrix $A_{i,j}$: for every $r$, $\tr((A_{i,j} A_{i,j}^{\top})^{r}) \geq \Norm{A_{i,j}}_2^{2r}$.  This method is typically employed in analyzing the spectral norm of \emph{random} matrices. But notice that $A_{i,j}$ is a \emph{fixed} matrix -- nothing random in it. Nevertheless, our key observation is if $G$ has no cycle of length $\leq 2r$, then one can derive the same \emph{exact upper bound} on $\tr(A_{i,j}^{2r})$ \emph{as if it was a random ``signing''} of the adjacency matrix of $G$. 

We have: 
\[
\tr((A_{i,j}A_{i,j}^{\top})^{r}) = \sum_{v_1, v_2, \ldots, v_{2r} \in [n]} A_{i,j}(v_1, v_2) A_{i,j}(v_3,v_2) \cdots A_{i,j}(v_{2r-1},v_{2r})A_{i,j}(v_{1},v_{2r})\mper
\]
The term corresponding to $(v_1, v_2, \ldots, v_{2r})$ contributes a non-zero value (of at most $1$) to the right hand side above only if the sequence $\{v_i, v_{i+1}\}$ is an edge, say $e_i$ in $G$ for each $i \leq 2r$. Consider now the multiset of edges $E' = \{e_1, e_2, \ldots, e_r\}$. Since these are edges on a walk, viewing the $e_i$'s as subsets of $[n]$ of size exactly $2$, we must have that $\oplus_{i = 1}^{2r} e_i = 0$. Let's now prune $E'$ by removing any $e_i,e_j$ that are equal. We must be able to remove all edges in this procedure, as otherwise we are left with a $2$-regular induced subgraph inside $G$, and so $G$ must have a cycle of length $\leq 2r$. Thus, each edge of $G$ occurs an even number of times in the multiset $E'$. 

Let's now use this observation to count the number of returning walks beginning with a fixed vertex $v_1$. For each edge, we ``match'' its first occurrence along the walk with the last occurrence. There are $\frac{(2r)!}{r! 2^r}$ different ways to select this matching. Given a matching, there are at most $r$ distinct choices of edges to be made. We make these choices inductively along the path from $v_1$ to $v_{2r}$. At each step we can make a new choice (i.e., we are not traversing an edge that is already matched to a previously chosen edge) given our previous choices, there are at most $\Delta = \max\{2^i,2^j\}d$ choices for the edge. Summing up over all choices for $v_1$, we obtain that the number of non-zero contributing $2r$ length walks is at most $n \cdot \Delta^r 2^{r} r!$. Thus, 
\[
\Norm{A_{i,j}}_2 \leq \max\{2^{i/2}, 2^{j/2}\} \cdot n^{1/2r} d^{1/2} 2^{1/2} \sqrt{r} \leq 2d^{1/2} \max\{2^{i/2}, 2^{j/2}\} \sqrt{2 \log_2 n} \mcom
\] 
for $r = 2\lceil \log_2 n \rceil$ and large enough $n$.  

Plugging back in \eqref{eq:partitioned-bound} yields that 
\[
y^{\top}Ay \leq 2\sum_{i\leq j} 2^{-(i+j)/2}n 2^{j/2} \cdot \sqrt{2d \log_2 n} \leq nd^{1/2} O(\log_2^{1.5} n)\mper\qedhere
\]
\end{proof}

Let's summarize the idea of the proof: analyzing the quadratic forms on the hypercube of adjacency matrix with row bucketing yields a (significantly weaker but still non-trivial) bound on the girth of a graph with a given number of edges. This argument can possibly be sharpened (to only an absolute constant factor loss) by switching to the non-backtracking walk matrix of $G$ (instead of the adjacency matrix) and dropping the row bucketing step. The above loose argument, however, generalizes to hypergraphs as we show below.

\begin{lemma}[Feige's Conjecture for $4$-Uniform Hypergraphs]
\label{lem:feigeconj4}
Every $4$-uniform hypergraph $\cH$ on $[n]$ with $m \geq O(\frac{n^2}{\ell}\log_2^3 n)$ hyperedges has an even cover of length $O(\ell \log_2 n)$. 
\end{lemma}

For every $C \in \cH$, let $b_C =1$ and consider the Kikuchi matrix $A$ of the $4$-XOR instance specified by $\cH$ and $b_C$'s. Equivalently, $A$ is simply the adjacency matrix of the ``Kikuchi graph'' on vertex set ${{[n]} \choose \ell}$ where edges correspond to pairs $(S,T)$ such that $S \oplus T = C$ for some $C \in \cH$. The idea is to repeat the argument for the adjacency matrix above but this time on the Kikuchi graph. The ``win'' in this scheme is a reduction of the problem on hypergraphs to a related problem on the associated Kikuchi graph that is significantly easier to reason about. 

As in the previous section, each $C \in \cH$ corresponds to ${4 \choose 2} \cdot {{n-4} \choose {\ell-2}}$ different non-zero entries in $A$ and in particular, we have for $x = 1^n$, 
\[
(x^{\star \ell})^{\top} A x^{\star \ell} = 6 {{n-4} \choose {\ell-2}} |\cH|\mper
\]

Our proof exactly mirrors the proof of the above weak Moore bound for graphs. We will show that if $\cH$ has no even cover of length $2r$ for $r = 0.5\log_2 N$, then, $y^{\top} A y \leq {n \choose \ell} \tilde{O}(\ell)$ for any $y \in \{-1,1\}^N$. 

Let $\deg(S) = |\{ C \mid |S \cap C| = 2\}|$. For every $i \leq \lceil \log_2 m \rceil$, let $\cF_i = \{ S \mid 2^{i-1} d_0 < \deg(S) \leq 2^{i} d_0\}$ ($\cF_0 = \{S \mid \deg(S) \leq d_0\}$) denote the $i$-th row bucket, where $d_0 \sim m \ell^2/n^2$. Note that $\deg(S) \leq m$ and $d_0 \geq 1$ so the number of buckets is indeed at most $\lceil \log_2 m \rceil$.
Write $A = \sum_{i,j} A_{i,j}$ where $A_{i,j}$ has all rows not in $\cF_i$ and all columns not in $\cF_j$ zeroed out. We can now argue:
\[
(y^{\top}A y) \leq \sum_{i,j} \Norm{A_{i,j}}_2 \cdot \sqrt{ |\cF_i| |\cF_j|}\mper
\]

In the previous section, when $b_C$'s were independent, random bits, we used the matrix Bernstein inequality to bound $\Norm{A_{i,j}}_2$. Here, $b_C$'s are fixed (and equal to $1$) so, of course, that strategy cannot work. Instead, our proof uses the trace moment method as in the proof of the weak Moore bound.

\begin{proposition}
Suppose $\cH$ has no even cover of length $2r$ for $r \leq \log_2 N$. Then, $\Norm{A_{i,j}}_2 \leq O(\ell \log_2 n)$. 
\end{proposition}

\begin{proof}[Proof of Proposition]

As before, we use $\Norm{A_{i,j}}_2^{2r} \leq \tr((A_{i,j} A_{i,j}^{\top})^{r}))$ for any $r \in \N$. 
We then have:
\[
\tr((A_{i,j} A_{i,j}^{\top})^{r})) = \sum_{S_1, S_2, \ldots, S_{2r} \in {{[n]} \choose \ell}} A_{i,j}(S_1, S_2) \cdot A_{i,j}(S_3, S_2) \cdots A_{i,j}(S_{2r-1},S_{2r})A_{i,j}(S_{2r+1},S_{2r})\mcom
\]
where we adopt the convention that $S_{2r+1}=S_1$.
Let us now analyze the right hand side of this equality. Each term in the RHS corresponds to a $2r$-tuple $(S_1,S_2,\ldots,S_{2r})$ of sets from ${{[n]} \choose \ell}$ and contributes either $0$ or $1$. 

If a term corresponding to $(S_1, S_2,\ldots, S_{2r})$ contributes a $+1$, then, for each $t \leq 2r$, there must be a $C_t \in \cH$ such that $S_t \oplus S_{t+1} = C_t$. Thus, each non-zero term is in bijection with $(S_1,C_1, C_2, \ldots, C_{2r})$. On the other hand, we must have that
 $\emptyset =\oplus_{t =1}^{2r} S_t \oplus S_{t+1} = \oplus_{t = 1}^{2r} C_t$, as each $S_t$ appears twice in $\oplus_{t =1}^{2r} S_t \oplus S_{t+1}$, and thus the total symmetric difference is $\emptyset$. Hence, a non-zero term  $(S_1,C_1, C_2, \ldots, C_{2r})$ must satisfy $\oplus_{t = 1}^{2r} C_t = \emptyset$. 

Let us analyze such a $2r$-tuple of hyperedges. By removing equal pairs repeatedly as in the previous proof, we can conclude that since $\cH$ has no even cover of length $\leq 2r$, each hyperedge in $\cH$ occurs an even number of times in the (multi)set $\{C_1, C_2, \ldots, C_{2r}\}$. 

We now count the number of $(S_1, C_1, \ldots, C_{2r})$ such that each $C_t$ occurs an even number of times. Since $C_t$'s occur in pairs, we can match the first occurrence of the hyperedge in the ordered set $(C_1, C_2,\ldots, C_{2r})$ to the last. There are $\leq 2^r r!$ different ways of selecting this matching. Given $S_1$ and the matching, there are at most $r$ unique $C_t$'s to choose. When making a choice of $C_t$ (say), $S_t$ is already determined by the previous choices. Thus, we have at most $\deg(S_t) \leq \Delta := \max\{2^i, 2^j\}d_0$ unique choices for the hyperedge $C$. In total, there are  $\leq N \cdot 2^r r! \Delta^r$ non-zero terms, and so
\[
\Norm{A_{i,j}}_2 \leq N^{1/2r} 2^{1/2} \sqrt{r} \max\{2^{i/2}, 2^{j/2}\} \sqrt{d_0} \leq \max\{2^{i/2}, 2^{j/2}\} 2 \sqrt{\log_2 N} \sqrt{d_0} \mcom
\]
for $r = 0.5 \log_2 N$ and large enough $n$. The remaining calculation now mimics the one for \cref{prop:moore-bound} (recalling that $d_0 \sim m \ell^2/n^2$), and finishes the proof of \cref{lem:feigeconj4}
\end{proof}

\subsection{Refuting semirandom $3$-XOR via row pruning}
The case of odd arity XOR refutation is lot more challenging. Even in the well-studied special case of random CSP refutation and the special case of $\ell= O(1)$ (i.e., polynomial time refutation), the case of odd arity CSPs turns out to be significantly more challenging than the even case. So let us start by focusing on the case of random $3$-XOR first. 

As in the case of $4$-XOR, we would like to begin by finding a simpler argument (compared to \cite{DBLP:conf/stoc/RaghavendraRS17}) for the special case of \emph{random} $3$-XOR using some appropriate variant of the Kikuchi matrix. In fact, ~\cite{WeinAM19} attempted this by introducing a variant of the Kikuchi matrix, and suggested an explicit approach (see Section F.1 of~\cite{WeinAM19}) to prove that the spectral norm of that matrix yields a refutation, but this does not work (see \cref{sec:wam-sug-does-not-work}). Indeed, we do not know of any reasonable variant of the Kikuchi matrix whose spectral norm yields a refutation for even \emph{fully random} $3$-XOR instances with the expected trade-off. 

Instead, we will introduce a variant of the Kikuchi matrix and use it to give a refutation algorithm for \emph{random} $3$-XOR instances by relying not on the spectral norm (which is too large) but, instead, the spectral norm of a ``pruned'' version of the matrix. We will then discuss the remaining key ideas of \emph{regularity decomposition} combined with row bucketing to refute semirandom odd-arity XOR. 

\medskip\noindent\textbf{Bipartite $3$-XOR.}  The Kikuchi matrix we introduce relates directly to a polynomial obtained by applying the standard ``Cauchy-Schwarz trick'' to the input polynomial. Consider the polynomial $\psi(x) = \frac{1}{m} \sum_{C \in \cH} b_C x_C$ associated with a $3$-XOR instance described by a $3$-uniform hypergraph $\cH$ with $m$ hyperedges and ``right-hand sides'' $b_C$'s. Here, for a set $R$ we define $x_R := \prod_{i \in R} x_i$, and in particular, $x_C = \prod_{i \in C} x_i$. For each $C \in \cH$, let $C_{\min}$ be the minimum indexed element in $C$ (using the natural ordering on $[n]$). Then, 
\[
\max_{x\in \on^n} \psi(x) \leq \max_{x,y \in \on^n} \frac{1}{m} \sum_{C \in \cH} b_C y_{C_{\min}} x_{C \setminus C_{\min}} \mcom
\]
where each $y_u$ is formally a new variable, but we think of $y_u$ as equal to $x_u$.
Let us reformulate this expression a bit: let $\cH_u = \{ C \mid C'=(C,u) \in \cH, C'_{\min} = u \}$. Then, 
\[
\max_{x \in \on^n} \psi(x) \leq \max_{x,y \in \on^n} \frac{1}{m} \sum_{u \in  [n]} y_u \sum_{C \in \cH_u} b_{u,C} x_{C}\mper
\]
One can think of the RHS as the polynomial associated with a \emph{bipartite} instance of the $3$-XOR problem on $2n$ variables, since every constraint uses one $y$ variable and two $x$ variables. Our refutation algorithm works for such bipartite instances more generally. 

For such a bipartite instance, using the Cauchy-Schwarz inequality, we can derive: 
\begin{multline} \label{eq:cauchy-schwarz-overview}
\Paren{\frac{1}{m} \sum_{u \in  [n]} y_u \sum_{C \in \cH_u} b_{u,C} x_{C}}^2 \leq \frac{n}{m^2} \sum_{u} \sum_{C,C' \in \cH_u} b_{u,C} b_{u,C'} x_C x_{C'} \\ = \frac{nm}{m^2} + \frac{n}{m^2}\sum_{u} \sum_{C \neq C' \in \cH_u} b_{u,C} b_{u,C'} x_C x_{C'} := \frac{n}{m} + f(x)
\end{multline}
The first term on the RHS is $\leq \epsilon^2/2$ if $m \geq 2n/\epsilon^2$. The second term produces a $\leq 4$-XOR instance. 

We thus end up with a $4$-XOR instance -- an even arity instance -- albeit with significantly less randomness than required in the argument from previous section. So, we need some different tools to refute such instances. The first of this is the following variant of the Kikuchi matrix that is designed specifically for ``playing well'' with the symmetries produced by the squaring step above.  

\medskip\noindent\textbf{Our Kikuchi matrix.} Our Kikuchi matrix is indexed by subsets of size $\ell$ on a universe of size $2n$ -- corresponding to two labeled copies of each of the original $n$ $x$ variables.  For each $C \in \cH$, let $C^{(1)}$ be the subset of $[n] \times [2]$ where every variable is labeled with ``$1$'', and similarly for $C^{(2)}$. This trick is done to ensure that the clauses $x_{C^{(1)}} x_{C'^{(2)}}$ form a $4$-XOR instance, as now $C^{(1)}$ and $C'^{(2)}$ \emph{by definition} cannot intersect.

For even $k$, the ``independent'' pieces in the Kikuchi matrix were the matrices $A_C$, one for each $C \in \cH$. For odd $k$, the independence pieces will be $A_u$ -- one for each $y_u$ because of the loss of independence due to the Cauchy-Schwarz step above. 

\begin{definition}[Kikuchi Matrix, $3$-XOR]
Let $N = {{[2n]} \choose \ell}$. For every $u \in [n]$, let $A_u \in \R^{N \times N}$ be defined as follows: for each $S,T \subseteq [n] \times [2]$ of size $\ell$, we will set $A_u(S,T)$ to be non-zero if there are $C,C' \in \cH_u$ such that $S \oplus T = C^{(1)} \oplus {C'}^{(2)}$ and $1=|S \cap C^{(1)}| = |S \cap {C'}^{(2)}| = |T \cap C^{(1)}| = |T \cap {C'}^{(2)}|$. That is, $A_u(S,T)$ is non-zero if each of $S,T$ contain one variable from each of $C^{(1)}$ and ${C'}^{(2)}$. In that case, we will set $A_u(S,T) = b_{u,C} \cdot b_{u,C'}$.
Finally, set $A = \sum_{u} A_u$. 
\end{definition}
Equivalently, $A_u(S,T)$ is non-zero if there are $C,C' \in \cH_u$ such that the $1$-labeled (respectively, $2$-labeled) elements in $S,T$ have symmetric difference $C$ ($C'$, respectively).  This construction is important for the success of our row pruning step (which we will soon discuss) and at the same time ensures that every pair $(C,C')$ of constraints in $\cH_u$ contributes an equal number of non-zero entries in the Kikuchi matrix $A$. We note that if we do not introduce the $2$ copies of each variable, the number of times a pair $(C,C')$ appears in the matrix would depend on $\abs{C \cap C'}$.

The quadratic forms of $A$ relate to the value of the underlying $4$-XOR instance: for $D = 4 {{2n-4} \choose {\ell-2}}$, 

\[
\val(\phi)^2\leq \frac{n}{m} + \val(f) \leq \frac{n}{m} + \frac{n}{m^2 D} \paren{\max_{z \in \on^N} z^{\top} A z}\mper
\]

\medskip\noindent\textbf{Bounding $z^{\top} A z$.} 
In the even arity case, we were able to obtain a refutation at this point by simply using the spectral norm of $A$ to bound the right hand side above. However, this turns out to provably fail here. To see why, let us define the relevant notion of degree -- the count of the number of non-zero entries in each row of $A_u$:
\[
\deg(S) = |\{ C,C' \in \cH_u \mid |S \cap C^{(1)}| = |S \cap C'^{(2)}|= 1\}|
\]
If we were to apply the matrix Bernstein inequality, the ``almost sure'' upper bound on $A_u$ for all $u$ is at least as large as $\sim \max_{S} \sqrt{\deg(S)}$ and it's not too hard to show that there are $S$ for which this bound is at least $\ell$. As a result, the best possible spectral norm upper bound that we can hope to obtain on $A$ is $\Omega(\ell \log_2 N) = \tilde{\Omega}(\ell^2)$ -- a bound that gives us no non-trivial refutation algorithm. 

\medskip\noindent\textbf{Row pruning.} The key observation that ``rescues'' this bad bound is that $\deg(S)$ cannot be large for too many rows. To see why, consider the random variable that selects a uniformly random $S \in {{[2n]} \choose \ell}$ and outputs $\deg(S)$. This can be well approximated (for our purposes) by a random set where every element is included independently with probability $\sim \ell/2n$. The expectation of $\deg(S)$ on this distribution is $O(1)$. By relying on the fact that $|C \cap C'| = \emptyset$ in $\cH_u$ for almost all pairs with high probability, $\var{\deg(S)} = O(1)$. A Chernoff bound yields that the fraction of $S$ for which $|\{C \in \cH_u \mid |S \cap C|>O(\log n)\}|$ is inverse polynomially small in $n$. A union bound on all $u$ then shows the fraction of rows that are ``bad'' for any $u$ is at most an inverse polynomial. 

It turns out we can ignore such ``bad'' rows with impunity. This is because we are interested in certifying upper bounds on quadratic forms of $A$ over ``flat'' vectors again and we can argue that removing ``bad'' rows cannot appreciably affect them.  For the ``residual matrix'', we can now apply the matrix Bernstein inequality and finish off the proof! The execution here requires \emph{row bucketing} with respect to a combinatorial parameter called the butterfly degree (generalizing a similar notion in~\cite{AbascalGK21}) that controls the variance term in the analysis. 

\medskip\noindent\textbf{Extending to semirandom instances.}
Looking back, the previous analysis uses that the graphs $\cH_u$'s obtained from the random $3$-uniform hypergraph $\cH$ satisfy a ``spread'' condition: there are few to none distinct pairs $C,C' \in \cH_u$ such that $C \cap C' \neq \emptyset$. This notion of \emph{regularity} is the precise pseudo-random property of $\cH$ that is enough for our argument (i.e.\ the row pruning step) above to go through. 

 For the case of $3$-XOR, such a regularity property is relatively easy to ensure by a certain ad hoc argument: if too many pairs $C,C' \in \cH_u$ happen to share a variable, then, ``resolving'' them yields a system of $2$-XOR constraints. Refutation in the special case of $2$-XOR is easy  using the Grothendieck inequality; this has been observed in several works, including~\cite{Fei07,AbascalGK21}. Indeed, this was roughly the strategy employed in the recent work~\cite{AbascalGK21} for the case of $\ell=O(1)$ for  semirandom $k$-XOR. In fact, in the $\ell=O(1)$ regime, it turns out that one can reduce $k$-XOR for all $k$ to the case of $3$-XOR and get the right trade-off; thus, such a decomposition for $3$-XOR is enough for the argument of~\cite{AbascalGK21} to go through for all $k$. 

\subsection{Handling $k$-XOR for $k>3$: hypergraph regularity} 
When $\ell\gg O(1)$, the case of higher arity $k$ does not reduce to $k=3$. Once again, working through the case of random $k$-XOR inspires our more general argument. We work with a generalization of the Kikuchi matrix introduced in the previous section for the case of $k=3$. When analyzing the row pruning step, we need to rely on certain tail inequalities for low-degree polynomials that depends on the ``spread'' of the hypergraph defined by the indices of the non-zero coefficients in the polynomial. We use the result of Schudy and Sviridenko~\cite{schudysviridenko} that builds on an influential line of work on concentration inequalities for polynomials with combinatorial structure in the monomials begun by~\cite{KimV00}. Our application of this inequality is rather delicate and as a result, we need a significantly stricter notion of \emph{regularity} -- we call this $(\epsilon,\ell)$-regularity -- for our row pruning argument to go through. 

\medskip\noindent\textbf{Hypergraph regularity decomposition.} Roughly speaking the notion of $(\epsilon,\ell)$-regularity (indexed by the parameter $\ell$ and an accuracy bound $\epsilon$) we need demands that for each subset $Q \subseteq [n]$, the number of hyperedges $C \in \cH_u$ such that $Q \subseteq C$ is bounded above by an appropriate function of $m, n$ and $\ell$. Random hypergraphs $\cH$ satisfy such a regularity property naturally. 

In order to handle arbitrary hypergraphs, we introduce a new \emph{regularity decomposition} for hypergraphs. Our regularity decomposition is based on a certain \emph{bipartite contraction} operation that takes a bipartite hyperedge $(u,C) \in \cH$ and a subset $Q \subseteq C$ and replaces it with $((u,Q), C\setminus Q)$. This operation should be thought of as ``merging'' all the elements in $Q$ and $u$ into a new single element $(u,Q)$ and obtaining a smaller arity hyperedge in a variable extended space. 

We give a greedy (and efficient) algorithm that starts from a $k$-uniform hypergraph and repeatedly applies bipartite contraction operations to obtain a sequence of $k'$-uniform hypergraphs for $k' \leq k$ along with some ``error'' hyperedges, with the property that each of the $k'$-uniform hypergraphs produced are $(\epsilon,\ell)$-regular. Each of the $k'$-uniform hypergraphs produced is naturally associated with a $k'$-XOR instance related to the input $k$-XOR instance. We show that refuting each of these output instances yields a refutation for the original $k$-XOR instance. 

\medskip\noindent\textbf{Cauchy-Schwarz even in the even-arity setting.} Unlike in the case of $3$-XOR where the resulting bipartite $3$-XOR instance had an equal number of $y$ and $x$ variables above, the bipartite $k'$-XOR instances produced via our regularity decomposition are \emph{lopsided} -- the number of $y$ variables can be polynomially larger in $n$ than the number $n$ of the $x$ variables. A naive bound on the number of constraints required to refute such instances is too large to yield the required trade-off, even in the case for even $k$.

Instead (and in contrast to all previous works on CSP refutation), we show that an appropriate application of the ``Cauchy-Schwarz'' trick above to even-arity $k$-XOR instances allows us to ``kill'' the $y_u$'s appearing in the polynomial, leaving us with only a polynomial in the $x_i$'s. This is a rather different usage of the technique -- in prior works (and as in the case of $3$-XOR highlighted above), it was instead used to build the right ``square'' matrices for obtaining spectral refutations of the associated CSP instances when $k$ is odd.

\subsection{Organization}
	The rest of the paper is organized as follows. In \cref{sec:prelims}, we introduce some notation, and recall the various concentration inequalities and facts that we will use in our proofs. In \cref{sec:decomposition}, we state and prove our hypergraph decomposition lemma. In \cref{sec:poly-refute}, we begin the proof of \cref{thm:main-poly-refute-intro}, reducing to the case of $k$-XOR to handling ``lopsided'' polynomials. In \cref{sec:partition-refute}, we handle the ``lopsided'' polynomials, finishing the proof of \cref{thm:main-poly-refute-intro}. In \cref{sec:smoothed}, we use \cref{thm:main-poly-refute-intro} to prove \cref{thm:smoothed-main-intro}. In \cref{sec:fko}, we prove Feige's conjecture (\cref{thm:fko-main-intro}), and finally in \cref{sec:fko2} we use \cref{thm:fko-main-intro,thm:smoothed-main-intro} to prove \cref{thm:ourfko-main-intro}.

\section{Preliminaries}
\label{sec:prelims}

\subsection{Basic notation} We let $[n]$ denote the set $\{1, \dots, n\}$. For two subsets $S, T \subseteq [n]$, we let $S \oplus T$ denote the symmetric difference of $S$ and $T$, i.e., $S \oplus T := \{i : (i \in S \wedge i \notin T) \vee (i \notin S \wedge i \in T)\}$.

For a rectangular matrix $A \in \R^{m \times n}$, we let $\norm{A}_2 := \max_{x \in \R^m, y \in \R^n: \norm{x}_2 = \norm{y}_2 = 1} x^{\top} A y$ denote the spectral norm of $A$, and $\boolnorm{A} := \max_{x \in \on^m, y \in \on^n} x^{\top} A y$ denote the $\infty\to1$ norm of $A$. We note that $\boolnorm{A} \leq \sqrt{n m} \norm{A}_2$.

Given a multiset $\cH$, we will use the notation $C \in \cH$ to refer to a distinct element of $C$, and $C \ne C'$ for $C, C' \in \cH$ to denote that $C$ and $C'$ are distinct elements in $\cH$ (even if they are two different copies of the same element).

Given a set $R$ and variables $x_1, \dots, x_n$, we will let $x_R := \prod_{i \in R} x_i$. In particular, $x_C := \prod_{i \in C} x_i$.

\subsection{Concentration inequalities}
We will rely on the following concentration inequalities. The first is the standard rectangular matrix Bernstein inequality.
\begin{fact}[Rectangular matrix Bernstein, Theorem 1.6 of \cite{Tropp2012}]
\label{thm:matrixbernstein}
Let $X_1, \dots, X_k$ be independent random $d_1 \times d_2$ matrices with $\E[X_i] = 0$ and $\norm{X_i} \leq R$ for all $i$. Let $\sigma^2$ be such that $\sigma^2 \geq \max(\smallnorm{\E[\sum_{i = 1}^k X_i X_i^{\top}]}, \smallnorm{\E[\sum_{i = 1}^k X_i^{\top} X_i]})$. Then for all $t \geq 0$, $\Pr[\smallnorm{\sum_{i = 1}^k X_i} \geq t] \leq (d_1 + d_2) \exp(\frac{-t^2/2}{\sigma^2 + R t/3})$.
\end{fact}

The second concentration inequality is a result for combinatorial polynomials due to Schudy and Sviridenko \cite{schudysviridenko} that is the culmination of an influential line of work begun by Kim and Vu~\cite{KimV00}.
\begin{fact}[Concentration of polynomials, Theorem 1.2 in \cite{schudysviridenko}, specialized]
\label{fact:schudy-sviridenko}
Let $\cH \subseteq {[n] \choose t}$ be a collection of multilinear monomials of degree $t$ in $n$ $\zo$-valued variables, and let $f(x) := \sum_{C \in \cH} \prod_{i \in C} x_i$. Let $Y_1, Y_2, \ldots, Y_n$ be independent and identically distributed Bernoulli random variables with $\Pr[Y_i =1]=\tau$. Then, for some absolute constant $R \geq 1$,
\begin{equation*}
\Pr[ |f(Y)-\E f(Y)| \geq\lambda ] \leq e^2 \max\Set{ \max_{r = 1,2,\ldots,t} e^{-\lambda^2/\nu_0 \nu_r R^{t}}, \max_{r =1,2,\ldots,t} e^{-(\frac{\lambda}{\nu_r R^{t}})^{1/r}}}\mcom
\end{equation*}
where, for every $r \leq t$,
$\nu_r = \tau^{t-r} \max_{h_0 \subseteq [n], |h_0|=r} \abs{\{h \in \cH: h \supseteq h_0\}}$. 
\end{fact}

\subsection{The sum-of-squares algorithm}
\label{sec:sum-of-squares}
We briefly define the key sum-of-squares facts that we use. These facts are all taken from~\cite{BarakS16,TCS-086}.
\begin{definition}[Pseudo-expectations over the hypercube] \label{def:pseudo-expectation}
A degree $d$ pseudo-expectation $\pE$ over $\on^n$ is a linear operator that maps degree $\leq d$ polynomials on $\on^n$ into real numbers with the following three properties:
\begin{enumerate}
    \item (Normalization) $\pE[1] = 1$.
	\item (Booleanity) For any $x_i$ and any polynomial $f$ of degree $\leq d-2$, $\pE[f x_i^2] = \pE[f]$. 
	\item (Positivity) For any polynomial $f$ of degree at most $d/2$, $\pE[f^2] \geq 0$. 
\end{enumerate} 
\end{definition}
We note that if $\E$ is the expectation operator of a distribution over $\on^n$, then $\E$ is a degree $d$ pseudo-expectation (for any $d$), and thus $\max_{x \in \on^n} f(x) \leq \max_{\pE} \pE[f]$, where the second max is taken over all degree $d$ pseudo-expectations $\pE$.

The SoS algorithm shows that we can efficiently maximize $\pE[f]$ over degree $d$ pseudo-expectations $\pE$ for a polynomial $f$. 
\begin{fact}[Sum-of-squares algorithm, Corollary 3.40 in \cite{TCS-086}]
\label{fact:sosalg}
Let $f(x_1, \dots, x_n)$ be a polynomial of degree $k$, where the coefficients of $f$ are rational numbers with $\poly(n)$ bit complexity. Let $d \geq k$. There is an algorithm that, on input $f,d$, runs in time $n^{O(d)}$ and outputs a value $\alpha$ such that $\beta + 2^{-n} \geq \alpha \geq \beta$, where $\beta$ is the maximum, over all degree $d$ pseudo-expectations $\pE$ over $\on^n$, of $\pE[f]$. 
\end{fact}
\ignore[proof sketch]{
\begin{proof}[Proof sketch of \cref{fact:sosalg}]
Fix $\eps = 2^{-n^{d}}$, and let $f(x) := \sum_{\abs{S} \leq k} f_S x_S$, where each $f_S \in \R$ and $x_S := \prod_{i \in S} x_i$.
Let $\cK$ denote the set of PSD matrices $M$ with real entries, indexed by sets $S \subseteq [n]$ with $\abs{S} \leq d/2$, such that $\abs{M(S,T) - M(S',T')} \leq \eps$ if $S \oplus T = S' \oplus T'$. Note that, given a degree-$d$ pseudo-expectation $\pE$ over $\on^n$, the \emph{moment matrix} $M_{\pE}$ of $\pE$, defined as $M_{\pE}(S,T) := \pE[x_S x_T]$, is in $\cK$.

For each set $S$ with $\abs{S} \leq d$, fix disjoint $S_1, S_2$ of size $\leq d/2$ each such that $S = S_1 \oplus S_2$.
For a matrix $M \in \cK$, define $\alpha(M) := \sum_{\abs{S} \leq k} f_S M(S_1, S_2)$. Note that if $M = M_{\pE}$ for some $\pE$, then $\alpha(M) = \pE[f]$.

Let $\alpha^* := \sup_{M \in \cK} \alpha(M)$.
By \cite[Corollary 3.40]{TCS-086}, in $n^{O(d)} \log(1/\eps) = n^{O(d)}$ time we can find $M \in \cK$ such that $\abs{\alpha(M) - \alpha^*} \leq \eps$. We now show that $\alpha^*$ is close to $\beta := \max \pE[f]$, where the max is taken over all degree $d$ pseudo-expectations $\pE$ over $\on^n$.

First, we observe that $\beta \leq \alpha^*$, as the moment matrix $M_{\pE}$ for each $\pE$ is in $\cK$. Next, let $M \in \cK$ be such that $\abs{\alpha(M) - \alpha^*} \leq \eps$. We will find a $\pE$ such that $\pE[f]$ is close to $\alpha(M)$. Let $M'$ be the matrix where $M'(S,T) := M( (S \oplus T)_1, (S\oplus T)_2)$. As $M \in \cK$, we have $\abs{M'(S,T) - M(S,T)} \leq \eps$ for all $S,T$, which implies that $\abs{\alpha(M') - \alpha(M)} \leq 2^{\poly(n)} n^{O(d)} \eps$ (as each $f_S$ has $\poly(n)$ bit complexity) and $M' \succeq -2^{\poly(n)} n^{O(d)}  \eps \cdot \Id$, as $M \succeq 0$. By taking a convex combination of $M'$ with $M_0 = \Id$, the moment matrix of the uniform distribution, we get $M''$ with $M'' \succeq 0$, $\abs{\alpha(M'') - \alpha(M)} \leq 2^{\poly(n)} n^{O(d)} \eps$, and $M''(S,T) = M''(S', T')$ if $S \oplus T = S' \oplus T'$. It then follows that $M''$ is the moment matrix of some $\pE$, and so $\beta \geq \pE[f] = \alpha(M'') \geq \alpha^* - 2^{\poly(n)} n^{O(d)} \eps$.

Finally, as $2^{\poly(n)} n^{O(d)} \eps \ll 2^{-n-1}$, we thus have that $\abs{\alpha(M) - \beta} \leq 2^{-n-1}$. Outputting, $\alpha := \alpha(M) + 2^{-n-1}$, we must have $\alpha \geq \beta$ and $\alpha \leq \beta + 2^{-n}$, as required.
\end{proof}
}

We now list the other key properties of pseudo-expectations that we will use. First, we note that pseudo-expectations satisfy the Cauchy-Schwarz inequality.
\begin{fact}[SoS Cauchy-Schwarz inequality]
\label{fact:sos-cauchy-schwarz}
Let $f,g$ be polynomials with $\deg(f), \deg(g) \leq d/2$, and let $\pE$ be a degree $d$ pseudo-expectation. Then $\pE[f g] \leq \sqrt{\pE[f^2] \pE[g^2]}$.
\end{fact}

Next, we observe that SoS captures Grothendieck's inequality, which we recall below.
\begin{fact}[Grothendieck's inequality]
\label{fact:grothendieck}
Let $A$ be an $n \times n$ matrix and let $s = \max_{Z \in \R^{n \times n}, Z \succeq 0, Z_{i,i} = 1 \forall i} \tr(A \cdot Z)$. Then, $s \leq K_G \boolnorm{A}$, where $K_G \leq 1.8$ is a universal constant independent of $A$.  
\end{fact}
\begin{fact}[SoS ``knows of" Grothendieck]
\label{lem:sos-vs-grothendieck}
Let $A \in \R^{n \times n}$. Let $\pE$ be a pseudo-expectation over $\on^{n}$ of degree $\geq 2$. Then
\begin{equation*}
\pE[ x^{\top}Ax] \leq K_G \boolnorm{A} \leq 1.8 \boolnorm{A} \mper
\end{equation*}
\end{fact}
\begin{proof}
Since $\pE$ is a pseudo-expectation of degree $\geq 2$, the pseudo-moment matrix  $\pE[xx^{\top}] \succeq 0$. Further, since $\pE$ is over $\on^{n}$, $\pE[x_i^2] = 1$ for every $i \in [n]$. Thus, the matrix $Z = \pE[xx^{\top}] \succeq 0$, and has $Z_{i,i} =1$. Applying \cref{fact:grothendieck} completes the proof. 
\end{proof}

Finally, we observe that $\pE[f] \geq 0$ holds for all nonnegative $f$ on $k$ variables, provided that the degree $d$ is at least $2k$.
\begin{fact}
\label{fact:sos-completeness}
Let $f(x_1, \dots, x_k)$ be a \emph{non-negative} degree $\leq k$ multilinear polynomial in $x_1, \dots, x_k$, i.e., $f(x_1, \dots, x_k) \geq 0$ for all $x_1, \dots, x_k \in \on^k$. Let $\pE$ be a pseudo-expectation of degree $d$ over $\on^n$, where $d \geq 2k$. Then, $\pE[f] \geq 0$.
\end{fact}

\section{A Hypergraph Decomposition Lemma}
\label{sec:decomposition}
A key ingredient in our proof of \cref{thm:smoothed-main-intro} is a \emph{regular hypergraph decomposition} algorithm that takes an arbitrary $k$-uniform hypergraph and decomposes it into a $k-1$ different \emph{regular} sub-hypergraphs (after removing a small fraction of the hyperedges). In this section, we present this decomposition step. We first introduce some notation, and then explain the decomposition.

\begin{definition}[Uniform hypergraphs]
A $k$-uniform hypergraph $\cH$ on $n$ vertices is a collection $\cH$ of subsets of $[n]$ of size exactly $k$. For a set $\setQ \subseteq [n]$, we define $\deg(\setQ) := \abs{\{C \in \cH : \setQ \subseteq C\}}$.
\end{definition}
\begin{remark}
\label{remark:nonsimplehypergraph}
We will \emph{not} assume that $\cH$ is simple, i.e., $\cH$ can be a multiset. For simplicity, we will abuse notation and let $C \in \cH$ refer to an element of the \emph{multiset} $\cH$. We will say that $C \ne C'$ if $C$ and $C'$ are different elements of the multiset $\cH$, even if $C$ and $C'$ are equal as sets, i.e., they are distinct copies of the same element in the underlying set of $\cH$. As an example, we use the above definition of $\deg(\setQ)$ to refer to the number of $C \in \cH$ with $Q \subseteq C$, \emph{counted with multiplicity}. We encourage the reader to assume that $\cH$ is simple, and then observe that nothing changes if $\cH$ is a multiset, and definitions are changed appropriately to count multiplicities.
\end{remark}

Our decomposition lemma will decompose a uniform hypergraph into \emph{bipartite} hypergraphs, which we introduce.
\begin{definition}[Bipartite hypergraphs]
\label{def:bipartitehypergraph}
 A $p$-bipartite $t$-uniform hypergraph on $n$ vertices is a collection $\{\cH_u\}_{u \in [p]}$, where each $\cH_u$ is a collection of subsets of $[n]$ of size exactly $t-1$. We call each $\cH_u$, or just $u$, a \emph{partition} of the bipartite hypergraph. A set $C \in \cH_u$ corresponds to the hyperedge $(u, C)$. For a set $\setQ \subseteq [n]$ and $u \in [p]$, we define $\deg_u(\setQ) := \abs{\{C \in \cH_u : \setQ \subseteq C\}}$. When $p$ is clear from context or not relevant, we just use the terminology ``bipartite $t$-uniform hypergraph''.
\end{definition}
One should think of a bipartite hypergraph $\{\cH_u\}_{u \in [p]}$ as a hypergraph $\cH$ on two sets of vertices, $[p]$ and $[n]$, where each hyperedge $(u,C) \in \cH$ contains one vertex $u \in [p]$ and $k-1$ vertices in $[n]$; for $u \in [p]$, the $(k-1)$-uniform hypergraph $\cH_u$ contains all hyperedges $C$ such that the hyperedge $(u,C)$ is in the hypergraph $\cH$.

\begin{definition}[Hypergraph regularity]
\label{def:hypergraphregularity}
 We say that a $p$-bipartite $k$-uniform hypergraph $\{\cH_u\}_{u \in [p]}$ is $(\eps,\ell)$-regular if $\deg_u(Q) \leq \frac{1}{\eps^2} \max(\Paren{\frac{n}{\ell}}^{\frac{k}{2} - 1 - \abs{Q}}, 1)$ for all $Q \subseteq [n]$ of size at most $k-1$ and all $u \in [p]$. For convenience, we will say $\{\cH_u\}_{u \in [p]}$ is regular when $\eps, \ell$ are clear from context.
\end{definition}

\begin{remark}[Regularity is a pseudorandom property]
Informally speaking, a collection of $k$-tuples is regular if the number of $k$-tuples in $\cH_u$ that all contain a fixed set of size $j$ is appropriately upper bounded. It is not hard to show that if $\cH = \cup_{u \in [p]} \cH_u$ is a \emph{uniformly random} bipartite hypergraph with $p = n$ partitions and $m = \ell (\frac{n}{\ell})^{\frac{k}{2}}$ random $k$-tuples, then with high probability, for every $u \in [p], \setQ$, $\deg_u(\setQ) \leq \max(\frac{m}{p n^{\abs{Q}}}, 1) \cdot O(\log n) \leq 
\max(\Paren{\frac{n}{\ell}}^{\frac{k}{2} - 1 - \abs{Q}}, 1) \cdot O(\log n)$, which is the same condition of regularity, up to the $O(\log n)$ extra factor. Thus, regularity can be seen as a (weak) pseudorandom property of a bipartite hypergraph.
\end{remark}

Next, we define a notion of hypergraph decomposition that we call a bipartite contraction.
\begin{definition}[Bipartite contractions]
Let $\cH$ be a $k$-uniform hypergraph on $n$ vertices. We say that a pair of subsets $(Q, C')$ (of $[n]$) is a \emph{contraction} of the hyperedge $C \in \cH$ if $C = Q \cup C'$ and $Q, C'$ are disjoint. It is sometimes useful to think of this pair as denoting a set of size $1 + k - \abs{Q}$, where the first ``element'' of the set is the entire \emph{set} $Q$, and the remaining $k - \abs{Q}$ elements come from the set $C \setminus Q$.

A \emph{bipartite contraction} of $\cH$ is a collection of $k-1$ bipartite hypergraphs $\{\cH^{(t)}_u\}_{u \in [p^{(t)}]}$ for $t = 2, \dots, k$, along with a set $\cH^{(1)}$ of ``discarded edges'' where:
 \begin{enumerate}[(1)] 
 \item each $\{\cH^{(t)}_u\}_{u \in [p^{(t)}]}$ is a bipartite $t$-uniform hypergraph, 
 \item each $u \in [p^{(t)}]$ corresponds to a subset $Q_u \subseteq [n]$ of size $k + 1 - t$ (it is possible that $Q_u =Q_{u'}$ for distinct $u,u'$), 
 \item every hyperedge in any $\cH^{(t)}_u$ is a bipartite contraction of some hyperedge in $\cH$, i.e., for every $t$ and any $u \in [p^{(t)}]$ and $R \in \cH^{(t)}_u$, the set $Q_u \cup R = C$ for some $C  \in \cH$, so that the hyperedge $(Q_u,R)$ is a contraction of $C$, 
 \item every hyperedge $C$ is contracted exactly once, i.e., for each $C \in \cH$, either $C \in \cH^{(1)}$ or there exists unique $t$, $u \in [p^{(t)}], R \in \cH^{(t)}_{u}$ such that $Q_{u} \cup R = C$.
 \end{enumerate}
\end{definition}

Our hypergraph contraction lemma shows that for any $k$-uniform hypergraph $\cH$, we can efficiently find a bipartite contraction of $\cH$ such that each of the resulting bipartite hypergraphs is regular.
\begin{lemma}[Hypergraph contraction lemma]
\label{lem:decomposition}
Let $\cH$ be a $k$-uniform hypergraph on $n$ vertices with $k \geq 2$ and $\abs{\cH} = m$. Then, there is a bipartite contraction of $\cH$ such that
\begin{enumerate}[(1)]
\item $m^{(1)} := \abs{\cH^{(1)}} \leq \frac{n}{k \eps^2} \Paren{\frac{n}{\ell}}^{\frac{k}{2} - 1}$. 
\item For $t \geq 2$, each bipartite $t$-uniform hypergraph $\{\cH^{(t)}\}_{u \in [p^{(t)}]}$ is 
\begin{enumerate}[(a)] 
\item $(\eps, \ell)$-regular, 
\item $\abs{\cH^{(t)}_u} = m^{(t)}/p^{(t)} = \floor{ \frac{1}{\eps^2} \max(\Paren{\frac{n}{\ell}}^{t - \frac{k}{2} - 1}, 1)}$ for all $u \in [p^{(t)}]$, where $m^{(t)} := \sum_{u \in [p^{(t)}]} \abs{\cH^{(t)}_u}$.
\end{enumerate}
\end{enumerate}
Further, given $\cH$, the decomposition itself can be computed by an algorithm running in time $O(n^k \abs{\cH}^2)$. 
\end{lemma}
Observe that the lemma does not assume any lower bound on $m$. Indeed if $m$ is too small then we will have $m^{(t)} = 0$ for all $t \geq 2$.

\begin{proof}[Proof of \cref{lem:decomposition}]
We prove \cref{lem:decomposition} by analyzing the following greedy algorithm to construct the bipartite contraction. Before stating the formal algorithm, we first explain the high level idea of the algorithm, as it is very simple. 

If $\cH$ does not have enough hyperedges, then we set $\cH^{(1)} = \cH$ and are done.
Otherwise, there must be some ``violating'' set $\setQ$: namely, a set $\setQ$ where $\deg(\setQ)$ is above a threshold $\tau$ (related to the definition of regularity). We choose a ``maximal'' such violating $\setQ$, i.e., no set containing $\setQ$ is a violation, and then \begin{inparaenum}[(1)] \item remove an arbitrary $\tau$ hyperedges of the form $\setQ \cup C$ from $\cH$, \item take bipartite contractions $(\setQ, C \setminus Q)$ of all such hyperedges, and \item add them all to $\cH_u^{(k+1-|\setQ|)}$ where $u$ is ``new'' partition where $Q_u := Q$\end{inparaenum}. Notice that we may pick the same $\setQ$ more than once since we only decrease $\deg(\setQ)$ by $\tau$ in one such step. We repeatedly fix such violations greedily until we cannot and stop. Notice that this procedure is ``one-shot'' -- we do not recursively operate on the $\cH^{(t)}_u$'s produced, as (we will show) that they are guaranteed to $(\epsilon,\ell)$-regular by the design of our decomposition procedure.

We now state and analyze the greedy algorithm.
\begin{mdframed}
  \begin{algorithm}
    \label{algo:reg-hypergraph-decomp}\mbox{}
    \begin{description}
    \item[Given:]
       A $k$-uniform hypergraph $\cH$ over $n$ vertices, where $m = \abs{\cH}$.
    \item[Output:]
       A bipartite contraction $\{\{\cH^{(t)}_u\}_{u \in [p^{(t)}]}\}_{t = 2, \dots, k}$ of $\cH$.
           \item[Operation:]\mbox{}

    \begin{enumerate}
    	\item \textbf{Initialize:} $p^{(t)} = 0$ for $t = 2, \dots, k$.
	    \item \textbf{Fix violations greedily}:
			\begin{enumerate}
    			\item Find a maximal nonempty violating $\setQ$. That is, find $Q \subseteq [n]$ of size $1 \leq \abs{Q} \leq k - 1$ such that $\deg(Q) = \abs{\{C \in \cH : Q \subseteq C\}} > \frac{1}{\eps^2} \max(\Paren{\frac{n}{\ell}}^{\frac{k}{2}  - \abs{Q}}, 1)$, and $\deg(Q') \leq  \frac{1}{\eps^2} \max(\Paren{\frac{n}{\ell}}^{\frac{k}{2} - \abs{Q'}}, 1)$ for all $Q' \supsetneq Q$. 
    			\item Let $q = \abs{Q}$. Let $u = 1 + p^{(k + 1 - q)}$ be a new ``label'', and define $\cH'$ to be an arbitrary subset of $\{C \in \cH : Q \subseteq C\}$ of size exactly $\floor{ \frac{1}{\eps^2} \max(\Paren{\frac{n}{\ell}}^{\frac{k}{2} - q}, 1)}$. Let $Q$ be the set $Q_u$ associated with $u$, and define $\cH^{(k + 1 - q)}_u := \{C \setminus Q : C \in \cH'\}$.
       
			\item Set $p^{(k + 1 - q)} \gets 1 + p^{(k + 1 - q)}$,  and $\cH \gets \cH \setminus \cH'$.
    		\end{enumerate}
		\item If no such $\setQ$ exists, then put the remaining hyperedges in $\cH^{(1)}$. 
      \end{enumerate}
    \end{description}
  \end{algorithm}
  \end{mdframed}

First, we argue that $m^{(1)}$ is small. By construction, $\cH^{(1)}$ is the set of remaining hyperedges when the inner loop terminates, and so we must have $\deg(\{i\}) \leq \frac{1}{\eps^2} \max(\Paren{\frac{n}{\ell}}^{\frac{k}{2} - 1}, 1) =  \frac{1}{\eps^2} \Paren{\frac{n}{\ell}}^{\frac{k}{2} - 1}$ for every $i \in [n]$; we abuse notation and let $\deg$ only count hyperedges remaining in $\cH$. We then have $\sum_{i \in [n]} \deg(\{i\}) = k \abs{\cH^{(1)}}$, as every $C \in \cH^{(1)}$ is counted exactly $k$ times in the sum. Hence, $m^{(1)} \leq \frac{n}{k \eps^2}\Paren{\frac{n}{\ell}}^{\frac{k}{2} - 1}$.

We now argue that for each $t$, the bipartite hypergraphs $\{\cH^{(t)}_u\}_{u \in [p^{(t)}]}$ have the desired properties. Fix $t \in \{2, \dots, k\}$. By construction, each $\cH^{(t)}_u$ has the same size, namely $\floor{ \frac{1}{\eps^2} \max(\Paren{\frac{n}{\ell}}^{ t - \frac{k}{2} - 1}, 1)}$. It then follows that $m^{(t)} := \sum_{u \in [p^{(t)}]} \abs{\cH^{(t)}_u} = p^{(t)} \cdot \floor{ \frac{1}{\eps^2} \max\Paren{\Paren{\frac{n}{\ell}}^{ t - \frac{k}{2} - 1}, 1}}$, and so $p^{(t)} \leq \eps^2 m^{(t)}$ and $\abs{\cH^{(t)}_u} = \frac{m^{(t)}}{p^{(t)}}$. This proves property (b) in Item (2).

It remains to show property (a), that $\{\cH^{(t)}_u\}_{u \in [p^{(t)}]}$ is $(\eps, \ell)$-regular. To see this, let $u \in [p^{(t)}]$, and let $Q_u$ be the set associated with the label $u$. Note that we must have $\abs{Q_u} = k + 1 - t$. Let $\cH'$ denote the set of constraints in $\cH$ at the time when $u$ and $\cH^{(t)}_u$ are added to the bipartite hypergraph. Namely, we have that for every $C \in \cH^{(t)}_u$, $Q_u \cup C \in \cH'$. Now, let $R \subseteq [n]$ be a nonempty set of size at most $t - 1$. First, observe that if $R \cap Q_u$ is nonempty, then we must have $\deg_u(R) = 0$ (this degree is in the hypergraph $\cH_u^{(t)}$). Indeed, this is because $C \cap Q_u = \emptyset$ for all $C \in \cH^{(t)}_u$. So, we can assume that $R \cap Q_u = \emptyset$. Next, we see that $\deg_u(R) \leq \deg_{\cH'}(Q_u \cup R)$ (where $\deg_{\cH'}$ is the degree in $\cH'$), as $Q_u \cup C \in \cH'$ for every $C \in \cH^{(t)}_u$. Because $Q_u$ was maximal whenever it was processed in our decomposition algorithm and $Q_u \subsetneq Q_u \cup R$ as $R$ is nonempty and $R \cap Q_u = \emptyset$, it follows that 
\begin{flalign*}
&\deg_{\cH'}(Q_u \cup R) \leq \frac{1}{\eps^2} \max(\Paren{\frac{n}{\ell}}^{\frac{k}{2} - \abs{Q_u \cup R}}, 1) = \frac{1}{\eps^2} \max(\Paren{\frac{n}{\ell}}^{\frac{k}{2} - \abs{Q_u} - \abs{R} }, 1) \\
&= \frac{1}{\eps^2} \max(\Paren{\frac{n}{\ell}}^{t  - \frac{k}{2} - 1 - \abs{R} }, 1) \leq \frac{1}{\eps^2} \max(\Paren{\frac{n}{\ell}}^{\frac{t}{2} - 1 - \abs{R} }, 1) \mcom
\end{flalign*}
where the last inequality follows because $t - \frac{k}{2} - 1 - \abs{R} \leq \frac{t}{2} - 1 - \abs{R}$ always holds, as $t \leq k$. This finishes the proof.

Finally, when $R= \emptyset$, we trivially have $\deg_u(\emptyset) = \abs{\cH_u^{(t)}} = \floor{ \frac{1}{\eps^2} \max(\Paren{\frac{n}{\ell}}^{ t - \frac{k}{2} - 1}, 1)} \leq  \frac{1}{\eps^2} \max(\Paren{\frac{n}{\ell}}^{ t - \frac{k}{2} - 1}, 1) \leq \frac{1}{\eps^2} \max(\Paren{\frac{n}{\ell}}^{\frac{t}{2} - 1 }, 1)$, where we use again that $t - \frac{k}{2} \leq \frac{t}{2}$ as $t \leq k$.

To argue the runtime bound, we simply observe that each iteration takes $O(\abs{\cH} n^{k})$ time via brute-force, and there are clearly at most $\abs{\cH}$ iterations.
\end{proof}


\section{Refuting Semirandom Sparse Polynomials over the Hypercube}
\label{sec:poly-refute}
In this section, we describe an algorithm to tightly refute semirandom instances of homogenous, multilinear degree-$k$ polynomials. Concretely, our algorithm takes as input a homogenous, multilinear degree-$k$ polynomial $\phi$ in $n$ variables $x_1, \dots, x_n$ and outputs a correct upper bound on $\val(\phi) := \max_{x \in \Fits^n} \phi(x)$. Whenever the coefficients of the polynomial are generated from independent random probability distributions on $[-1,1]$ and the (multi-)hypergraph of coefficients has sufficiently many hyperedges, with high probability, the algorithm outputs a value that is smaller than a target $\epsilon$. The guarantees of our algorithm are captured by the theorem below.

\begin{theorem}[Refuting semirandom sparse polynomials] \label{thm:mainpolyrefute}
Let $k \in \N$ and $\ell \colon \N \to \N$ be a function such that $2(k-1) \leq \ell(n) \leq n$. There is an algorithm that takes as input a homogeneous, multilinear polynomial $\phi$ in $n$ variables $x_1, x_2,\ldots, x_n$ of total degree $k$ specified by a $k$-uniform multi-hypergraph $\cH$ and a collection of rational numbers $\{b_{C}\}_{C \in \cH}$:
\begin{equation}
\phi(x) = \frac{1}{m} \sum_{C \in \cH} b_C \cdot \prod_{i \leq k} x_{C_i}\mcom
\end{equation} 
and the algorithm outputs a value $\algval(\phi) \in [-1,1]$ in time $n^{O(\ell)}$ satisfying the following:
\begin{enumerate}[(1)]
 \item $1\geq \algval (\phi) \geq \val(\phi)$. 
\item  There is an absolute constant $\Gamma > 0$ such that if $n^{\log_2 n} \geq \abs{\cH} = m \geq m_0= \Gamma^k \cdot \Paren{\frac{n}{\ell}}^{\frac{k}{2}} \ell \cdot \frac{(\log_2 n)^{4k + 1}}{\eps^5}$ and the $b_C$'s are independent, mean $0$ random variables supported in $[-1, 1]$, then with probability $1-1/\poly(n)$ over the draw of $b_C$'s, it holds that $\algval(\phi) \leq \eps + 2^{-n}$.
\end{enumerate}

Moreover, our algorithm is ``captured'' by the canonical degree $2\ell$ sum-of-squares relaxation of polynomial maximization problem over the hypercube. Specifically, under the same hypothesis on $\phi$ as above, for every pseudo-expectation $\pE$ of degree $\geq 2\ell$ over $\on^n$, it holds that $\pE[\phi] \leq \epsilon$.
\end{theorem}
As is the case in \cref{sec:decomposition}, we will \emph{not} assume that $\cH$ is simple, and we will adopt the same notational conventions as in \cref{remark:nonsimplehypergraph}.

\subsection{Regular bipartite polynomials} 
Our proof of \cref{thm:mainpolyrefute} goes via a reduction to refuting sparse polynomials with additional structure that we call \emph{bipartite} polynomials. Bipartite polynomials can be seen as a generalization of partitioned $2$-XOR instances introduced in~\cite{AbascalGK21}. We next present this class of polynomials and identify a \emph{regularity} property of such polynomials that will be a key technical ingredient in our algorithm. 

\begin{definition}[$p$-bipartite polynomials]
\label{def:bipartitepoly}
Let $k \in \N$. A $p$-bipartite polynomial $\psi$ is a homogeneous degree $k$ polynomial in $p + n$ variables $y=\{y_u\}_{u \in [p]}$ and $x=\{x_j\}_{j \in [n]}$ defined by 
\begin{equation*}
\psi(y,x) = \frac{1}{m} \sum_{u = 1}^p y_u \sum_{C \in \cH_u} b_{u,C}  x_C\mcom
\end{equation*}
where $\{\cH_u\}_{u \in [p]}$ is a $p$-bipartite $k$-uniform hypergraph (\cref{def:bipartitehypergraph}), $b_{u,C} \in [-1,1]$ for every $C \in \cH$, $x_C := \prod_{i \in C} x_i$, and $m := \sum_{u \in [p]} \abs{\cH_u}$. The \emph{value} of $\psi$, denoted by $\val(\psi)$, is $\max_{y \in \on^{p}, x \in \on^n} \psi(y,x)$. Note that $\val(\psi) \in [-1,1]$ always. We also note that $\psi$ is a homogeneous degree $1$ polynomial in $y$. 
\end{definition}

\begin{definition}[Regular $p$-bipartite polynomials]
\label{def:partitioned-poly-regular-instances}
We say that a $p$-bipartite polynomial $\psi$ is $(\eps,\ell)$-regular if the underlying $p$-bipartite $k$-uniform hypergraph $\{\cH_u\}_{u \in [p]}$ is $(\eps,\ell)$-regular (\cref{def:hypergraphregularity}). When $\eps,\ell$ are clear from context, we will simply say that $\psi$ is regular.
\end{definition}

The bulk of the technical work in proving \cref{thm:mainpolyrefute} is in analyzing a refutation algorithm for regular instances of $p$-bipartite polynomials encapsulated in the following theorem.   

\begin{theorem}[Refuting regular bipartite polynomials]
\label{thm:bipartite-poly-refute}
\MYstore{thm:bipartite-poly-refute}{
Let $k \in \N$. For any $\ell:\N \rightarrow \N$ with $2(k-1) \leq \ell(n) \leq n$ for all $n \in \N$, there is an algorithm with the following properties: the algorithm takes as input a $p$-bipartite, homogeneous, polynomial $\psi=\psi(y,x)$ in variables $y=\{y_u\}_{u \in [p]}$ and $x=\{x_i\}_{i \in [n]}$ of total degree $k$:
\begin{equation*}
\psi(y,x) = \frac{1}{m}\sum_{u = 1}^p y_u \sum_{C \in \cH_u} b_{u,C} x_C\mcom
\end{equation*}
specified by a collection of $(k-1)$-uniform hypergraphs $\{\cH_u\}_{u \in [p]}$ and rational numbers in $[-1,1]$ $\{b_{u,C}\}_{u \in [p], C \in \cH_u}$.  The algorithm runs in time $(p+n)^{O(\ell)}$ time and outputs $\algval(\psi) \in [-1,1]$ satisfying the following:
\begin{enumerate}
\item For every $\psi$, $\algval(\psi) \geq \val(\psi)$.
\item Whenever $\psi$ and $b_{u,C}$'s satisfy:
\begin{enumerate} 
    \item $\psi$ is $(\eps, \ell)$-regular, 
    \item $\abs{\cH_u} \leq \frac{2m}{p}$ for all $u \in [p]$,
    \item $n^{\log_2 n} \geq m \geq \max\left \{\Gamma^k \cdot \Paren{\frac{n}{\ell}}^{\frac{k-1}{2}} \sqrt{p \ell} \cdot \frac{(\log_2 n)^{2k + 0.5}}{\eps^3}, p/\epsilon^2\right \}$, where $\Gamma$ is an absolute constant, and 
    \item Each $b_{u,C}$'s is chosen from (possibly different) independent mean zero distributions on $[-1,1]$.
\end{enumerate}
Then with probability $1 - 1/\poly(n)$ over the draw of $b_{u,C}$'s, $\algval(\psi) \leq \sqrt{2.8} \cdot \eps + 2^{-n}$.
\end{enumerate}
Further, our algorithm is ``captured'' by the sum-of-squares algorithm of degree $2\ell$: for every pseudo-expectation $\pE$ in variables $x,y$ of degree $2\ell$ over $\on^{p + n}$, $\pE[\psi(x,y)] \leq \sqrt{2.8} \cdot \eps$.
}
\MYload{thm:bipartite-poly-refute}
\end{theorem}
We defer the proof of~\cref{thm:bipartite-poly-refute} to \cref{sec:partition-refute}.
\subsection{Reduction to regular bipartite polynomials} 
We now use \cref{lem:decomposition} along with \cref{thm:bipartite-poly-refute} to complete the proof of \cref{thm:mainpolyrefute} by analyzing the following algorithm:
\begin{mdframed}[frametitle = {Main Refutation Algorithm}, frametitlealignment=\centering]
  \begin{algorithm}
    \label{algo:main-refutation}\mbox{}
    \begin{description}
    \item[Given:]
       A polynomial $\phi$ specified by a  $k$-uniform multi-hypergraph $\cH$ over $n$ vertices and rational numbers $\{b_{C}\}_{C \in \cH}$. 
    \item[Output:]
       A value $\algval \in [-1,1]$. 
    \item[Operation:]\mbox{}
    \begin{enumerate}
        \item Apply the decomposition algorithm from \cref{lem:decomposition} to construct bipartite hypergraphs $\{\cH^{(t)}_u\}_{u \in [p^{(t)}]}$ for $2 \leq t \leq k$, and a set of discarded edges $\cH^{(1)}$. 
        \item For every $t$, $u \in [p^{(t)}]$ and for every hyperedge $C \in  \cH^{(t)}_u$, set $b_{u,C} = b_{Q_u \cup C}$.
        \item For $2 \leq t \leq k$, apply the refutation algorithm for regular bipartite polynomials from~\cref{thm:bipartite-poly-refute} to the degree $t$ $p^{(t)}$-bipartite polynomial specified by the bipartite hypergraph $\{\cH^{(t)}_u\}_{u \in [p^{(t)}]}$ and $b_{u,C}$'s to obtain $\algval_t$. Set $\algval_1 = 1$.
        \item Output $\algval = \frac{1}{m} \sum_{t = 1^k} m^{(t)} \cdot \algval_t$, where $m^{(t)} = \sum_{u \in [p^{(t)}]}\abs{\cH_u^{(t)}}$.
      \end{enumerate}
    \end{description}
  \end{algorithm}
  \end{mdframed}
\begin{proof}[Proof of \cref{thm:mainpolyrefute} from \cref{lem:decomposition,thm:bipartite-poly-refute}] 
First, without loss of generality we will assume that $\eps \leq \frac{1}{\sqrt{2}}$, so that $\frac{1}{\eps^2} \geq 2$. This is without loss of generality, as it only changes the universal constant in \cref{thm:mainpolyrefute}.

For each $t$ and $u \in [p^{(t)}]$, let $Q_u \subseteq [n]$ denote the subset of size $k + 1 - t$ associated to $u$, and let $\psi_t$ be the polynomial associated with the $t$-uniform $(\eps, \ell)$-regular bipartite hypergraph $\{\cH_u^{(t)}\}_{u \in [p^{(t)}]}$ obtained from the hypergraph $\cH$ specifying the input polynomial $\phi$ by applying the decomposition algorithm from \cref{lem:decomposition}. Thus, $\psi_t$ is a polynomial in the $p^{(t)} + n$ variables $\{y^{(t)}_u\}_{u \in [p^{(t)}]} \cup \{x_i\}_{i \in [n]}$, and $\psi_t(\{y^{(t)}_u\}_{u \in [p^{(t)}]}, x) := \frac{1}{m^{(t)}} \sum_{u \in [p^{(t)}]} y^{(t)}_u \prod_{C \in \cH^{(t)}_u} b_{Q_u \cup C} x_C$. We then have that
\begin{equation}
\label{eq:decomp}
\phi(x) = \frac{1}{m} \sum_{t = 2}^k m^{(t)} \psi_t( \{x_{Q_u}\}_{u \in [p^{(t)}]}, x)  + \frac{1}{m} \sum_{C \in \cH^{(1)}} b_C x_C\mper
\end{equation}
Indeed, this follows immediately from the definition of a bipartite contraction, because when we substitute $x_{Q_u}$ for $y_{u}$ for some $u \in [p^{(t)}]$, then $y_u x_C = x_{Q_u \cup C} = x_{C'}$ for $C' \in \cH$. 

Let $\algval_t = \algval(\psi_t)$ be the output of the refutation algorithm from \cref{thm:bipartite-poly-refute} applied to $\psi_t$. Then, $\val(\psi_t) \leq \algval_t$. Thus, using \eqref{eq:decomp}, $\val(\phi) \leq \frac{1}{m} \sum_{t = 1}^k m^{(t)} \algval_t = \algval$. 

Next, if for some $t$, $m^{(t)} \leq \eps m$, then using the trivial bound of $\algval(\psi_t) \leq 1$ yields $m^{(t)} \algval(\psi_t) \leq \eps m$. Note that in particular, $m^{(1)} \leq \eps m$ always holds, as $m \geq \frac{1}{\eps^3} \Paren{\frac{n}{\ell}}^{\frac{k}{2}} \cdot \ell$ and $m^{(1)} \leq \frac{n}{k \eps^2} \Paren{\frac{n}{\ell}}^{\frac{k}{2} - 1}$.

Now, suppose that for some $t$, $m^{(t)} \geq \eps m$. Notice that $m^{(t)} \leq m \leq n^k \leq n^{\log_2 n}$. We now prove that in this setting, $m^{(t)} \geq \Gamma^t \cdot \Paren{\frac{n}{\ell}}^{\frac{t-1}{2}} \sqrt{p^{(t)} \ell} \cdot \frac{(\log_2 n)^{2t + 0.5}}{\eps^3}$. We know that $m^{(t)} = p^{(t)} \cdot \floor{ \frac{1}{\eps^2} \max(\Paren{\frac{n}{\ell}}^{t - \frac{k}{2} - 1}, 1)}$. Hence, it suffices to show 
\begin{equation*}
\eps m \geq \Gamma^{2t} \cdot \Paren{\frac{n}{\ell}}^{t-1}\ell \cdot \frac{(\log_2 n)^{4t + 1}}{\eps^6} \cdot \frac{1}{ \frac{1}{2\eps^2} \max(\Paren{\frac{n}{\ell}}^{t - \frac{k}{2} - 1}, 1)} \mcom
\end{equation*}
where we use that $\floor{ \frac{1}{\eps^2} \max(\Paren{\frac{n}{\ell}}^{t - \frac{k}{2} - 1}, 1)} \geq \floor{\frac{1}{\eps^2}} \geq \frac{1}{2 \eps^2}$ as $\frac{1}{\eps^2} \geq 2$.

Hence, for $t \geq \frac{k}{2} + 1$, it suffices to have
\begin{equation*}
\eps m \geq 2\Gamma^{2t} \cdot \Paren{\frac{n}{\ell}}^{\frac{k}{2}}\ell \cdot \frac{(\log_2 n)^{4t + 1}}{\eps^4} \mcom
\end{equation*}
and for $t < \frac{k}{2} + 1$, it suffices to have
\begin{equation*}
\eps m \geq 2\Gamma^{2t} \cdot \Paren{\frac{n}{\ell}}^{t-1}\ell \cdot \frac{(\log_2 n)^{4t + 1}}{\eps^4} \mper
\end{equation*}
As $m \geq {\Gamma'}^{k} \cdot \Paren{\frac{n}{\ell}}^{\frac{k}{2}} \ell \cdot \frac{(\log_2 n)^{4k + 1}}{\eps^5}$, for the absolute constant $\Gamma' = 2\Gamma^2$, both conditions are satisfied.

We have thus shown that if $m^{(t)} \geq \eps m$, then $\psi_t$ satisfies the conditions of \cref{thm:bipartite-poly-refute}, and so we have $m^{(t)} \algval_t \leq \eps m^{(t)} \leq \eps m$ with probability $1 - 1/\poly(n)$ over the draw of $b_C$'s. By union bound over all $t$, we thus get that $\algval(\phi) \leq O(k \eps)$ with probability $1 - k/\poly(n) \geq 1-1/\poly(n)$ over the draw of $b_C$'s. This completes the analysis of the second guarantee. 

The running time of the algorithm is dominated by the time required to apply the refutation algorithm from~\cref{thm:bipartite-poly-refute} to each of the bipartite polyomials produced by the decomposition algorithm. This cost is bounded above by $n^{O(\ell)}$. 

Finally, the fact that this algorithm is ``captured'' by SoS follows because \cref{thm:bipartite-poly-refute} is ``captured'' by SoS and the linearity of the pseudo-expectations.
\end{proof}


\section{Refuting Regular Bipartite Polynomials}
\label{sec:partition-refute}
In this section, we prove \cref{thm:bipartite-poly-refute}. 
Our algorithm is based on the semidefinite programming relaxation of the ``$\infty \to 1$''-norm of an appropriate matrix associated with the polynomial $\psi$. The analysis of the algorithm will naturally establish the ``Further,...'' part of the statement. 

As in several prior works starting with~\cite{Coja-OghlanGL04}, our proof of \cref{thm:bipartite-poly-refute} applies the ``Cauchy-Schwarz'' trick in order to work with an even-degree polynomial associated with $\psi$. 

\begin{lemma}[Cauchy-Schwarz trick]
\label{lem:cauchyschwarztrick}
Let $\psi$ be a $p$-bipartite, homogeneous, polynomial $\psi=\psi(y,x)$ in variables $y=\{y_u\}_{u \in [p]}$ and $x=\{x_i\}_{i \in [n]}$ of total degree $k$: 
\begin{equation*}
\psi(y,x) = \frac{1}{m}\sum_{u = 1}^p y_u \sum_{C \in \cH_u} b_{u,C} x_C\mper
\end{equation*}
Let $f$ be the following polynomial obtained from $\psi$:
\begin{equation*}
 f(x) =\frac{p}{m^2} \sum_{u = 1}^p \sum_{(C,C') \in \cH_u \times \cH_u, C \ne C'} b_{u,C} b_{u,C'} x_C x_{C'}\mper
\end{equation*} 
Then $\val(\psi)^2 \leq  \frac{p}{m} + \val(f)$. Further, for every pseudo-expectation $\pE$ of degree $\geq 2k$ over $\on^{p + n}$, $\pE[\psi]^2 \leq  \frac{p}{m} + \pE[f]$. 
\end{lemma}
\begin{proof}
Fix an assignment in $\on$ to the $y_u$'s and $x_i$'s. We then have
\begin{align*}
\psi^2(y,x) &= \Paren{\frac{1}{m} \sum_{u = 1}^p y_u \sum_{C \in \cH_u} b_{u,C} x_C}^2 \leq \frac{1}{m^2} \Paren{\sum_{u = 1}^p y_u^2} \Paren{ \sum_{u = 1}^p \Paren{\sum_{C \in \cH_u} b_{u,C} x_C}^2}\\
&\leq \frac{p}{m^2}  \cdot \sum_{u = 1}^p \sum_{C \in \cH_u} b_{u,C}^2 x_C^2 + \frac{p}{m^2}  \sum_{u \leq p } \sum_{(C,C') \in \cH_u \times \cH_u, C \ne C'} b_{u,C} b_{u,C'} x_C x_{C'}\\
&\leq \frac{p}{m} + \frac{p}{m^2} \sum_{u = 1}^p \sum_{(C,C') \in \cH_u \times \cH_u, C \ne C'} b_{u,C} b_{u,C'} x_C x_{C'}\mcom
\end{align*}
where the first inequality above uses the Cauchy-Schwarz inequality, the second uses that $y_u^2 =1$ for every $u$, and the third uses that $b_{u,C}^2 \leq 1$ and $x_C^2=1$. Further, observe that by using the SoS version of the Cauchy-Schwarz inequality (\cref{fact:sos-cauchy-schwarz}) and the fact that $\pE$ is over $\on^{p + n}$, we see that the above also holds for all degree $d \geq 2(k-1)$ pseudo-expectations $\pE$.

Taking the maximum over $x$ and $y$ on both sides then yields that $\val(\psi)^2 \leq \frac{p}{m} +  \val(f)$.
Taking the maximum over all pseudo-expectations $\pE$ on $\on^{p + n}$ and using \cref{fact:sos-cauchy-schwarz} yields that $\pE[\psi]^2 \leq \pE[\psi^2] \leq  \frac{p}{m} + \pE[f]$.
\end{proof} 

\subsection{Our Kikuchi matrix and algorithm} 
As \cref{lem:cauchyschwarztrick} shows, it suffices to upper bound $\val(f)$.
Our certificate of an upper bound on $\val(f)$ is based on an appropriate variant of the Kikuchi matrix of~\cite{WeinAM19}. To define our matrix, it is convenient to think of having two clones of each of the $n$ possible ``$x$'' variables. For every $i$, we will use $(i,1)$ and $(i,2)$ to denote the two clones of the $i$-th variable below. For any set $C \subseteq [n]$, we will use $C^{(1)}$ to denote the set $\{(i,1) \mid i \in C\}$, i.e., the clause $C$ using the first type of clones, and $C^{(2)}$ to be the clause $C$ using the second type of clones. Recall that for any sets $S,T$, let $S \oplus T$ denote the symmetric difference of the two sets. More generally, let $S_1 \oplus S_2 \oplus \cdots \oplus S_t$ denote the set of all elements that occur in an odd number of different $S_i$'s. 

\begin{definition}[Our \kikuchi Matrix] \label{def:kikuchi-matrix}
Let $\ell \in \N$ and let $N := {2n \choose \ell}$. 

Fix a $p$-bipartite $k$-uniform hypergraph $\{\cH_u\}_{u \in [p]}$. For each $u \in [p]$, define the $N \times N$ matrix $A_u$, indexed by sets $S \subseteq [n] \times [2]$ of size $\ell$, as follows. For any two sets $S,T \subseteq [n] \times [2]$ of size $\ell$ and sets $C \ne C' \in \cH_u$ of size $k-1$, we say that $S \overset{C,C'}{\leftrightarrow} T$ if 
\begin{enumerate} 
    \item $S \oplus T = C^{(1)} \oplus C'^{(2)}$, 
    \item $k$ is odd, and $\abs{S \cap C^{(1)}} = \abs{S \cap C'^{(2)}} = \abs{T \cap C^{(1)}} = \abs{T \cap C'^{(2)}} = \frac{k-1}{2}$, or, 
    \item $k$ is even, and $\abs{S \cap C^{(1)}} = \abs{T \cap C'^{(2)}} = \frac{k}{2}$ and $\abs{S \cap C'^{(2)}} = \abs{T \cap C^{(1)}} = \frac{k-2}{2}$, or,  
    \item $k$ is even, and $\abs{S \cap C^{(1)}} = \abs{T \cap C'^{(2)}} = \frac{k-2}{2}$ and $\abs{S \cap C'^{(2)}} = \abs{T \cap C^{(1)}} = \frac{k}{2}$.
\end{enumerate}
Note that $C^{(1)} \oplus C'^{(2)} = C^{(1)} \cup C'^{(2)}$, as $C^{(1)}$ and $C'^{(2)}$ are disjoint by construction.

We define \begin{equation} 
\label{eq:defn-of-A_u}
A_u(S,T) = \begin{cases} b_{u,C} \cdot b_{u,C'} \text{ if } \exists C,C' \in \cH_u, \text{ s.t. } S \overset{C,C'}{\leftrightarrow} T, \\
0 \text{ otherwise. } \end{cases}  
\end{equation}
 If $\cH$ is not simple, then the nonzero entry above is replaced with $\sum_{C \ne C' \in \cH_u : S \overset{C,C'}{\leftrightarrow} T} b_{u,C} \cdot b_{u,C'}$. Note that the the sum is over pairs of different elements $C,C'$ of the multiset $\cH$ (which may nonetheless be equal as sets).
 
Our (overall) \kikuchi matrix $A$ for the polynomial $f$ is defined as 
\begin{equation}
    \label{eq:summed-kikuchi}
    A := \sum_{u = 1}^p A_u\mper 
\end{equation}
\end{definition}

\smallskip\noindent 
The matrix $A$ allows us to write $f$ as a quadratic form, as the following lemma shows.
\begin{lemma}
\label{lem:kikuchi-quad-form}
Let $N := {2n \choose \ell}$ and let $A$ be the Kikuchi matrix in \cref{def:kikuchi-matrix} associated with an arbitrary $p$-bipartite $\psi$ specified by a bipartite hypergraph $\cH$ and coefficients $\{b_{u,C}\}_{u \in [p], C \in \cH}$. For any $x \in \Fits^n$, let $x^{\star \ell} \in \Fits^{N}$ be the vector where the $S$-th entry of $x^{\star \ell}$ is $x_S := \prod_{b \in [2]} \prod_{(i, b) \in S} x_{i} $. Then, 
\begin{equation} \label{eq:quadratic-form-kikuchi}
(x^{\star \ell})^{\top} A x^{\star \ell} = \frac{m^2 D}{p} \cdot  f(x)
\end{equation}
for $D$ as defined in~\cref{eq:defn-of-D}.
As a consequence, since $x^{\star \ell}$ has $\pm 1$-valued entries, $\val(f) \leq \frac{p}{m^2 D} \boolnorm{A}$. Furthermore, for every pseudo-expectation $\pE$ of degree $\geq 2\ell$ over $\on^n$, $$\pE[f] = \frac{p}{m^2D} \pE[(x^{\star \ell})^{\top} A x^{\star \ell}] \leq K_G \cdot \frac{p}{m^2 D} \boolnorm{A} \mcom$$
where $K_G \leq 1.8$ is the universal constant in \cref{fact:grothendieck}.
\end{lemma}
\begin{proof}
To see \eqref{eq:quadratic-form-kikuchi}, observe that by definition of $A$, if $k$ is odd then every pair $(C,C')$ in $\cH_u$ with $C \neq C'$ appears exactly ${k - 1 \choose \frac{k-1}{2}}^2 {2n - 2(k-1) \choose \ell - (k-1)} = D$ times when we expand the LHS. This is because we can choose $S$ by first picking its size $\frac{k-1}{2}$ intersection with $C^{(1)}$ and its intersection with $C'^{(2)}$ (${k-1 \choose \frac{k-1}{2}}^2$ choices) and then picking the rest of the set (${2n - 2(k-1) \choose \ell - (k-1)}$ choices), and this also completely determines $T$. A similar calculation yields the value of $D$ when $k$ is even, and so
 \cref{eq:quadratic-form-kikuchi} then follows. This is the place where we crucially use the ``clones'' of the variables to ensure that each pair $(C,C')$ appears the same number of times on the LHS. Without this trick, the number of times a pair $(C,C')$ appears would instead depend on $\abs{C \cap C'}$.

 The ``As a consequence,...'' part now follows by the definition of the $\infty \to 1$ norm. The ``furthermore'' follows by \cref{fact:grothendieck} and \cref{lem:sos-vs-grothendieck}.
\end{proof}

Below, we summarize the definitions that we have made so far.
\begin{mdframed}[frametitle = {Key Notation}, frametitlealignment=\centering]
\label{box:key-notation}
\begin{enumerate}
    \item The input polynomial $\psi$
    \begin{equation}
\psi(y,x) = \frac{1}{m}\sum_{u = 1}^p y_u \sum_{C \in \cH_u} b_{u,C} x_C\mcom
\end{equation} is $(\eps,\ell)$-regular, and $p$-bipartite, homogeneous of total degree $k$ and is described by a collection of $(k-1)$-uniform hypergraphs $\{\cH_u\}_{u \in [p]}$ one for every $u \in [p]$ and a collection of rationals $\{b_{u,C}\}_{u \in [p], C\in \cH_u}$. 
    \item The polynomial $f$ obtained after the Cauchy-Schwarz trick applied to $\psi$:
\begin{equation} \label{eq:poly-after-C-S}
 f(x) =\frac{p}{m^2} \sum_{u = 1}^p \sum_{(C,C') \in \cH_u \times \cH_u, C \ne C'} b_{u,C} b_{u,C'} x_C x_{C'}\mcom
\end{equation} 
is homogeneous of total degree $2(k-1)$. Furthermore, $\val(\psi)^2 \leq \val(f) + \frac{p}{m} \leq \val(f) + \eps^2$.
\item The Kikuchi matrix $A = \sum_u A_u$ of $f$ is an $N \times N$ matrix for $N = {{2n} \choose \ell}$. The entries of $A$ are indexed by sets $S,T \subseteq [n] \times [2]$ of size $\ell$ and the entry $A_u(S,T)$ is non-zero (and equal to $b_{u,C} b_{u,C'}$) if and only if $S \overset{C,C'}{\leftrightarrow} T$ for some distinct pair $C,C' \in \cH_u$. Each pair $(C,C')$ from $\cH_u$ contributes $D$ non-zero entries in $A$ where 
\begin{equation}
    \label{eq:defn-of-D}
 D = \left\{ \begin{array}{ll}
{k - 1 \choose \frac{k-1}{2}}^2 {2n - 2(k-1) \choose \ell - (k-1)} & \mbox{if $k$ is odd} \qquad \\
2 {k - 1 \choose \frac{k}{2}}{k - 1 \choose \frac{k-2}{2}} {2n - 2(k-1) \choose \ell - (k-1)} & \mbox{if $k$ is even.} \qquad
\end{array}
\right. 
\end{equation}
Furthermore, $\val(f) \leq \frac{p}{m^2 D} \boolnorm{A}$.
\end{enumerate}
\end{mdframed}

We now describe our algorithm in the box below.
\begin{mdframed}[frametitle = {Refutation Algorithm for Regular Polynomials}, frametitlealignment=\centering]
  \begin{algorithm}
    \label{algo:refute-regular}\mbox{}
    \begin{description}
    \item[Given:]
       An $(\epsilon,\ell)$-regular, $p$-bipartite polynomial $\psi = \sum_{u} \sum_{C \in \cH_u} b_{u,C} y_u x_C$ in variables $x,y$ specified by a collection of $(k-1)$-uniform hypergraphs $\{\cH_u\}_{u \in [p]}$ on $[n]$ and rational numbers $\{b_{u,C}\}_{u \in [p], C \in \cH_u}$ in $[-1,1]$. 
    \item[Output:]
       A value $\alpha \in [-1,1]$ such that $\alpha \geq \val(\psi)$. 
    \item[Operation:]\mbox{}

    \begin{enumerate}
        \item Construct $A$, the $N \times N$ Kikuchi matrix from \cref{def:kikuchi-matrix}. 
        \item Compute the value of the following SDP: $s= \max_{Z \in R^{ N \times N}, Z \succeq 0, Z_{S,S} = 1 \ \forall S} \tr(A \cdot Z)$. 
        \item Output $\alpha = \sqrt{\frac{p}{m^2 D} \cdot s + \frac{p}{m}}$. 
    \end{enumerate}
    \end{description}
  \end{algorithm}
  \end{mdframed}
The crux of the analysis of the algorithm is captured in the following lemma that we establish in the remaining part of this section. 

\begin{lemma}[Bounding $\boolnorm{A}$] \label{lem:bool-norm-A}
Let $A$ be the Kikuchi matrix defined in \cref{def:kikuchi-matrix}. Then with probability $1 - 1/\poly(n)$ over the draw of the $b_{u,C}$'s,
\begin{equation*}
\boolnorm{A} \leq \frac{m^2 D \eps^2}{p}\mper
\end{equation*}
\end{lemma}

Observe that this lemma immediately finishes the proof of \cref{thm:bipartite-poly-refute}. Indeed, we clearly have $s \geq \val(f) \frac{D m^2}{p}$ because $Z = x^{\odot \ell} (x^{\odot \ell})^{\top}$ is a valid SDP solution with this value, and so by \cref{lem:cauchyschwarztrick}, $\alpha \geq \val(\psi)$ always holds. By \cref{fact:grothendieck}, we have $s \leq 1.8 \boolnorm{A}$. We already argued that $\frac{p}{m} \leq \eps^2$, and so the output of our algorithm is at most $\sqrt{2.8} \eps$. We note that we additionally require an additive $2^{-n}$ error in the final algorithm because we can only efficiently solve SDPs up to an exponentially small error.

\subsection{Bounding $\boolnorm{A}$: proof plan} 

Using \cref{lem:kikuchi-quad-form}, our task reduces to proving that $\boolnorm{A} \leq \frac{m^2 D \eps^2}{p}$ whenever $b_{u,C}$'s are chosen independently at random from distributions supported on $[-1,1]$. Our proof proceeds in three conceptual steps: 

\begin{enumerate}
\item \parhead{Row pruning.} First, we remove all rows in $A$ that have too large $\ell_1$ norm in any $A_u$ and show that this only incurs a small additive loss in our bound on $\boolnorm{A}$. This is somewhat delicate and crucially relies on regularity of the $\cH_u$'s and a careful application of the celebrated Schudy-Sviridenko polynomial concentration inequality for combinatorial polynomials~\cite{schudysviridenko}.

\item \parhead{Row bucketing.} The row pruning ensures that no row has a large $\ell_1$-norm in any single $A_u$. Taking inspiration from spectral analyses of combinatorial random matrices, one might expect that the spectral norm of $A$ after row pruning is upper bounded. However, this turns out not to be true when the $\cH_u$'s are arbitrary regular hypergraphs. Instead, we show that one can partition the row and columns of $A$ so that in each bucket of the partition, all the rows/columns have roughly equal contribution to the ``variance term''.  

\item \parhead{Spectral norm bound.} Our final step involves proving a spectral norm upper bound on each piece of the partition in order to upper bound its $\infty \to 1$ norm. This is the only step where we use randomness of the right-hand sides $b_C$'s. While different parts of the partition can have larger spectral norm, this is compensated for by the fact that these partitions will have a proportionally smaller number of rows/columns, thus yielding a good bound on the $\infty\to1$ norm of $A$.
\end{enumerate}

Let us now proceed with the details of each of the three steps above.  
\subsection{Row pruning} 
In order to implement our row pruning step, we will define \emph{bad} rows/columns of $A_u$ for each $u$. The following key definition abstracts out the property (of the hypergraphs defining the input polynomial) that decides which rows are bad:

\begin{definition}[Butterfly Degree]  \label{def:butterfly-degree}
Let $\cH_u$ be a $(k-1)$-uniform hypergraph on $[n]$. For any $C,C' \in \cH_u$, let
$$
\cR_{(C,C')}  = \Set{ R \subseteq [n] \times [2] \Mid |R| = k-1, \text{ } \Set{\abs{R \cap C^{(1)}}, \abs{R \cap C'^{(2)}}} = \Set{\Big\lceil\frac{k-1}{2} \Big\rceil, \Big\lfloor \frac{k-1}{2} \Big\rfloor}}\mper
$$
For any $S \subseteq [n] \times [2]$, and $(k-1)$-uniform hypergraph $\cH_u$ on $[n]$, the \emph{butterfly degree} of $S$ in $\cH_u$ is defined by:
$$
\gamma_u(S) = \sum_{(C,C') \in \cH_u \times \cH_u, C \ne C'} \sum_{R \in \cR_{(C,C')}} \1(S \cap (C^{(1)} \cup C'^{(2)}) = R) \mper
$$
For a collection of $(k-1)$-uniform hypergraphs $\cH_u$ on $[n]$ for $u \in [p]$, the \emph{total butterfly degree} of $S$ is defined by $\gamma(S) = \sum_{u \in [p]} \gamma_u(S)$. 
\end{definition}

We note that the notion of total butterfly degree above generalizes the notion of butterfly degree studied in~\cite{AbascalGK21}; the original notion of ``butterfly degree'' is so named because it counts numbers of butterfly-shaped graphs.

The following lemma shows that the butterfly degree characterizes the $\ell_1$-norm of the rows of the Kikuchi matrix $A_u$. 
\begin{lemma}[Butterfly Degree and the $\ell_1$ norm of rows of the Kikuchi Matrix] \label{lem:characterizing-butterfly-degree}
Let $\cH_u$ be a $(k-1)$-uniform hypergraph on $[n]$ and $A_u$ be the associated matrix in \cref{def:kikuchi-matrix}. Then, for any $S \subseteq [n] \times [2]$, we have:
\[
  \gamma_u(S) \geq 
  \sum_T | A_u(S,T)| \mper
\]
\end{lemma}

\begin{proof}
If $k$ is odd, we observe that $\gamma_u(S)$ is the number pairs $(C,C') \in \cH_u \times \cH_u$ with $C \ne C'$ such that $\abs{S \cap C^{(1)}} = \abs{S \cap C'^{(2)}} = \frac{k-1}{2}$, and if $k$ is even, $\gamma_u(S)$ is the number of pairs such that $\abs{S \cap C^{(1)}} = \frac{k}{2}$ and $\abs{S \cap C'^{(2)}} = \frac{k-2}{2}$ or $\abs{S \cap C^{(1)}} = \frac{k - 2}{2}$ and $\abs{S \cap C'^{(2)}} = \frac{k}{2}$. The lemma now follows.
\end{proof}

We now identify ``bad rows'' in $A$ as those that have too large total butterfly degrees.

\begin{definition}[$\Delta$-Bad rows in $A$]
\label{def:badrows}
We define the set of $\Delta$-bad rows in $A$ to be: 
\[
\cB := \{S : \exists u \in [p] \text{, } \gamma_u(S) > \Delta\}\mper
\] 
Note that the set $\cB$ \emph{does not} depend on the values of the $b_{u,C}$'s. 
\end{definition}

Observe that by \cref{lem:characterizing-butterfly-degree}, every row that is not bad has an $\ell_1$-norm that is not too large. The following lemma bounds the number of bad rows in the Kikuchi matrix $A$. We defer the proof of \cref{lem:numbadrows} to \cref{sec:numbadrows}.

\begin{lemma}[Bound on bad rows]
\label{lem:numbadrows}
Let $A$ be the Kikuchi matrix associated with the polynomial $f$ obtained from an $(\epsilon,\ell)$-regular $p$-bipartite polynomial $\psi$ of total degree $k$ defined by $(k-1)$ uniform hypergraphs $\{\cH_u\}_{u \in [p]}$. Let $\cB$ be the set of $\Delta$-bad rows in $A$ for \begin{equation}
\label{eq:defn-of-Delta}
    \Delta = c^{k-1} \frac{1}{\eps^4} \left(\ln\left(\frac{32 p N}{\eps^2 D}\right)\right)^{2(k-1)} \mcom
\end{equation}
where $c$ is an absolute constant. Then $\abs{\cB} \leq \eps^2 D/16$.
\end{lemma}

This immediately implies the following corollary.
\begin{corollary}[Row pruning error]
\label{cor:rowpruning}
Let $A_{\cG, \cG}$ be the matrix obtained by ``zeroing out'' $A$ on all rows/columns in $\cB$. Then $\boolnorm{A - A_{\cG, \cG}} \leq \frac{m^2 D \eps^2}{2p}$.
\end{corollary}
\begin{proof}[Proof of \cref{cor:rowpruning} from \cref{lem:numbadrows}]
Let $B = A - A_{\cG, \cG}$. Let $S \subseteq [n] \times [2]$ be an arbitrary row (or column). We observe that the $\ell_1$ norm of the $S$-th row (or column) in $B$ (or even in $A$) is naively at most $\sum_{u = 1}^p \abs{\cH_u}^2$. This is because each (ordered) pair $(C,C') \in \cH_u \times \cH_u$ can contribute at most one nonzero entry to the $S$-th row, namely to the $T$-th entry where $T = S \oplus C^{(1)} \oplus C'^{(2)}$ (and this is only a valid entry if $\abs{T} = \ell$). As $\abs{\cH_u} \leq 2m/p$ for all $u$, the $\ell_1$ norm of the $S$-th row is at most $p \cdot \frac{4m^2}{p^2} = 4m^2/p$.

We next observe that if $B(S,T) \ne 0$, then at least one of $S,T$ is in $\cB$. Hence,
\begin{flalign*}
\boolnorm{B} \leq \sum_{S,T} \abs{B(S,T)} \leq \sum_{S \in \cB} \sum_{T} \abs{B(S,T)} + \sum_{T \in \cB} \sum_S \abs{B(S,T)} \leq 2 \abs{\cB} \cdot \frac{4m^2}{p} \mper
\end{flalign*}
As $\abs{\cB} \leq \eps^2 D/16$, this is at most $m^2 D \eps^2/2p$, as required. 
\end{proof}

We will now finish the proof, using the following bound on $\boolnorm{A_{\cG, \cG}}$ that we will prove.
\begin{lemma}
\label{lem:boolnormgood}
Let $A$ be the Kikuchi matrix associated with the polynomial $f$ obtained from an $(\epsilon,\ell)$-regular $p$-bipartite polynomial $\psi$ of total degree $k$ defined by $(k-1)$ uniform hypergraphs $\{\cH_u\}_{u \in [p]}$ and coefficients $\{b_{u,C}\}_{u \in [p], C \in \cH_u}$. Then, with probability $1 - 1/\poly(n)$ over the draw of $b_{u,C}$'s, it holds that 
\begin{equation*}
\boolnorm{A_{\cG, \cG}} \leq O(\log^2 m) N \Delta \cdot (\log N + \log \log m) + O(\log m) \sqrt{\frac{N D m^2 (\log N + \log \log m)}{p}} \mper
\end{equation*}
\end{lemma}

\begin{proof}[Finishing the proof of \cref{lem:bool-norm-A}]
By \cref{cor:rowpruning} and \cref{lem:boolnormgood}, we have with probability $1 - 1/\poly(n)$,
\begin{align*}
\boolnorm{A} & \leq \boolnorm{A - A_{\cG, \cG}} + \boolnorm{A_{\cG, \cG}} \\
& \leq \frac{m^2 D \eps^2}{2p} + O\Paren{\log^2 m \cdot N \Delta \cdot  (\log N + \log \log m)} + O\Paren{ \log m \sqrt{\frac{N D m^2 (\log N + \log \log m)}{p}}} \mper
\end{align*}

We now bound $\frac{N}{D}$.
\begin{claim}
\label{prop:NDbound}
$\frac{N}{D} \leq 16^{k-1} \cdot (\frac{n}{\ell})^{k-1}$, where $D$ is defined as in \cref{eq:defn-of-D}.
\end{claim}
\begin{proof}
We have
\begin{flalign*}
&\frac{N}{D} \leq \frac{{2n \choose \ell}}{{2n - 2(k-1) \choose \ell - (k-1)}} = \frac{(\ell - (k-1))!}{\ell!} \cdot \frac{(2n)!}{(2n - 2(k-1))!} \cdot \frac{(2n - \ell - (k-1))!}{(2n - \ell)!} \\
&\leq \left(\frac{n}{\ell}\right)^{k-1} \cdot \left(\frac{\ell}{\ell - (k-1)} \cdot 4 \cdot \frac{n}{2n - \ell - (k-1)} \right)^{k-1} \leq \left(\frac{n}{\ell}\right)^{k-1} \cdot 16^{k-1} \mcom
\end{flalign*}
for $n$ sufficiently large, as $\ell \geq 2(k-1)$.
\end{proof}

By \cref{prop:NDbound}, we thus have that $O(\log^2 m \cdot N \Delta (\log N + \log \log m))$ is at most $\frac{m^2 D \eps^2}{4p}$. Indeed, using that $m \leq n^{\log_2 n}$, we have
\begin{flalign*}
&O(1) \frac{p }{\eps^2 D} \cdot (\log_2 m)^2 \cdot N \Delta (\log N + \log \log m) \\
&\leq O(1)^{k-1} \frac{p }{\eps^2 D} \cdot (\log_2 n)^5 \cdot N \ell \cdot \frac{1}{\eps^4} \left(\ln(\frac{32 p N}{\eps^2 D})\right)^{2(k-1)} \\
&\leq O(1)^{k-1} \frac{\ell pN}{\eps^6 D} \cdot (\log_2 n)^5  \cdot (\ln^2 n)^{2(k-1)} \ \text{(as $p \leq \eps^2 m$ and $m \leq n^{\log_2 n}$)} \\
&\leq O(1)^{k-1} \frac{\ell p}{\eps^6} \left(\frac{n}{\ell}\right)^{k-1} \cdot (\log_2 n)^{4k + 1} \leq m^2 \mcom
\end{flalign*}
for $n$ sufficiently large, using the lower bound on $m$ in \cref{thm:bipartite-poly-refute}.

Similarly, we also have $O\Paren{ \log m \sqrt{\frac{N D m^2(\log N + \log \log m)}{p}}}$ is at most $\frac{m^2 D \eps^2}{4p}$, as
\begin{flalign*}
&O(1) \cdot\frac{p}{\eps^2 D} \cdot \log_2 m \sqrt{\frac{N D m^2 (\log N + \log \log m)}{p}} \\
&\leq O(1) \cdot (\log_2 n)^{2.5} \cdot \frac{m}{\eps^2}\cdot  \sqrt{\frac{p N \ell}{D}} \\
&\leq O(1)^{k-1} \cdot (\log_2 n)^{2.5} \cdot \frac{m}{\eps^2}\cdot  \sqrt{p \ell} \cdot \left(\frac{n}{\ell}\right)^{\frac{k-1}{2}} \leq m^2 \mcom
\end{flalign*}
again using the lower bound on $m$ in \cref{thm:bipartite-poly-refute}.
Hence, $\boolnorm{A} \leq \frac{m^2 D \eps^2}{p}$, which finishes the proof.
\end{proof}
We now prove \cref{lem:boolnormgood} (bounding $\boolnorm{A_{\cG, \cG}}$) and \cref{lem:numbadrows} (bound on bad rows).

\subsection{Bounding the $\infty\to1$ norm of the ``good rows'': proof of \cref{lem:boolnormgood}}
\label{sec:boolnormgood}
Let us denote $A_{\cG, \cG}$---the matrix obtained by zeroing out all rows in $\cB$ from the Kikuchi matrix $A$---by $G$ in this subsection for ease of notation. Similarly, we let $G_u := (A_u)_{\cG, \cG}$ be the matrix obtained by zeroing out all rows and columns in $\cB$ from the Kikuchi matrix $A_u$. Since $A = \sum_{u = 1}^p A_u$, we must have $G = \sum_{u = 1}^p G_u$.

At a high level, the idea of the proof is to split $G = \sum_{i,j} G^{(i,j)}$ into $O(\log^2 m)$ submatrices $G^{(i,j)}$ such that \begin{inparaenum}[(1)]\item each entry $(S,T)$ is non-zero in exactly one of $G^{(i,j)}$ and in that case, equals $G(S,T)$ and \item all non-zero rows (or columns) in any given $G^{(i,j)}$ have roughly the same butterfly degree\end{inparaenum}. This splitting accomplishes our ``row bucketing'' step. The second property above allows us to infer a reasonably good upper bound on the $\infty \to 1$ norm of $G^{(i,j)}$ in terms of an appropriately scaled spectral norm bound on $G^{(i,j)}$ -- we will provide two different proofs of this fact, one using the Matrix Bernstein inequality and the other based on the trace moment method. The first proof is simple but somewhat opaque in that it uses a powerful concentration inequality. The second proof is a little more elaborate but will be directly useful in \cref{sec:fko}.  We will then use the bounds on $\norm{G^{(i,j)}}_2$ to upper bound the $\infty \to 1$ norm of $G = A_{\cG, \cG}$. 

Let us start by defining the row bucketing formally by defining the $G^{(i,j)}$'s. 
\begin{definition}[Row bucketing]
Let $d = \frac{4 m^2 D}{p N} \geq 1$. Define a partition of the rows of the matrix $G$ into $\cF_0 \cup \cF_1 \cup \ldots \cF_t$ as follows: Set $\cF_0 := \{S \in \cG : \gamma(S) \leq d\}$. For each $t \geq i \geq 1$, let 
$$\cF_i := \{S \in \cG : 2^{i - 1} d < \gamma(S) \leq 2^i d\} \mper$$
Observe that since $\gamma(S) \leq \sum_{u = 1}^p \abs{\cH_u}^2 \leq m^2$ and $d \geq 1$, every good row index $S \in \cG$ is in some $\cF_i$ for $i \leq t = 2 \log_2 m$. Thus, the $\cF_i$'s for $i \leq 2 \log_2 m$ form a partition of all the rows of $G$.

For each $i,j \in \{0,1,\ldots, t\}$, let $G^{(i,j)}$ be the submatrix of $G$ such that for any entry $(S,T)$, if $S \in \cF_i, T \in \cF_j$, $G^{(i,j)}(S,T) = G(S,T)$ and $G^{(i,j)}(S,T) = 0$ otherwise. \label{def:row-bucketing}
\end{definition}

\begin{lemma}[Size of $\cF_i$'s]
\label{claim:partitionsize}
Let $\cF_0 \cup \cF_1 \cup \ldots \cF_t$ for $t \leq 2 \log_2 m$ be the partition of the rows of the matrix $G$ constructed in \cref{def:row-bucketing}. Then, $\abs{\cF_0} \leq N$ and $\abs{\cF_i} \leq 2^{1-i} N$ for each $i \in [t]$.
\end{lemma}
\begin{proof}
The bound on $\abs{\cF_0}$ is trivial. 
For $i \geq 1$, we observe that $2^{i - 1} d \abs{\cF_i} < \sum_{S \in \cF_i} \gamma(S) \leq \sum_{S} \gamma(S) \leq D \sum_{u = 1}^p \abs{\cH_u}^2 = D \cdot \frac{4m^2}{p} = d N$, as every (ordered) pair $(C,C') \in \cH_u \times \cH_u$ with $C \ne C'$ appears in exactly $D$ entries in the original matrix $A$. 
\end{proof}

We now come to the key part of the proof that establishes an upper bound on the spectral norm of each $G^{(i,j)}$. 
\begin{lemma}[Spectral norm of $G^{(i,j)}$'s]
\label{lem:Gspecbound}
Let the $G^{(i,j)}$'s be the matrices defined in \cref{def:row-bucketing}. Then, for each $i,j \in \{0, \dots, t\}$, with probability $1 - \frac{1}{\log_2^2 m \cdot \poly(n)}$ over the draw of the $b_{u,C}$'s, 
\[
\norm{G^{(i,j)}}_2 \leq O(1) \cdot \Delta (\log N + \log \log m) + O(1) \cdot 2^{0.5 \max(i,j)}\sqrt{d (\log N + \log \log m)}\mper
\]
\end{lemma}

This is enough to immediately complete the proof of \cref{lem:boolnormgood}. 

\begin{proof}[Proof of \cref{lem:boolnormgood}]
The total number of pairs of $(i,j)$ such that $i, j \leq t = 2 \log_2 m$ is at most $4\log_2^2 m$. Thus, applying \cref{lem:Gspecbound} and doing a union bound over all $(i,j)$ yields that with probability at least $1-1/\poly(n)$ over the draw of the $b_{u,C}$'s, $\norm{G^{(i,j)}}_2 \leq O(1) \cdot \Delta (\log N + \log \log m) + O(1) \cdot 2^{0.5 \max(i,j)}\sqrt{d (\log N + \log \log m)}$ for every $i,j$ simultaneously. Let us condition on this event in the following. 

The final idea in the proof is to observe the following key fact: for any $y,z \in \on^N$, we must have:
\[
y^{\top} G^{(i,j)} z = y_{\cF_i}^{\top} G^{(i,j)} z_{\cF_j} \leq \Norm{y_{\cF_i}}_2 \Norm{z_{\cF_j}}_2 \Norm{G^{(i,j)}}_2 = \sqrt{ |\cF_i| |\cF_j|} \Norm{G^{(i,j)}}_2 \mper 
\]
In the first equality we used the fact that only the rows in $\cF_i$ (and columns in $\cF_j$, respectively) are non-zero in $G^{(i,j)}$ and in the inequality, we used the definition of the spectral norm. 

Thus, we must have:
\[
\boolnorm{G^{(i,j)}} = \max_{y,z \in \on^N} y^{\top} G^{(i,j)} z \leq \sqrt{ |\cF_i| |\cF_j|} \Norm{G^{(i,j)}}_2 \mper 
\]

Thus, by triangle inequality for $\boolnorm{\cdot}$, we have:
\begin{flalign*}
&\boolnorm{G} \leq \sum_{i = 0}^t \sum_{j = 0}^t \boolnorm{G^{(i,j)}} \leq \sum_{i = 0}^t \sum_{j = 0}^t \sqrt{\abs{\cF_i} \abs{\cF_j}} \norm{G^{(i,j)}}_2 \\
&\leq O( N t^2 \Delta (\log N + \log \log m)) + 2 \sum_{i = 0}^t \sum_{j = i}^t N \sqrt{2^{2 - i - j}} \cdot O(1) \cdot 2^{0.5 j} \sqrt{d (\log N + \log \log m)} \\
&= O(N t^2 \Delta (\log N + \log \log m)) + O(N \sqrt{d (\log N + \log \log m)}) \sum_{i = 0}^t \sum_{j = i}^t 2^{-0.5i} \\
&= O(N t^2 \Delta (\log N + \log \log m)) +  O(N t\sqrt{d (\log N + \log \log m)}) \mper
\end{flalign*}
As $t = O(\log m)$ and $d = \frac{4 m^2 D}{p N}$, \cref{lem:boolnormgood} follows.
\end{proof}

We now complete the proof of \cref{lem:Gspecbound}. We present two different proofs of \cref{lem:Gspecbound}. The first is a simple proof using the Matrix Bernstein inequality. The second proof is based on the trace moment method, and will be important to us in \cref{sec:fko}.

\subsubsection{Proof of \cref{lem:Gspecbound} using Matrix Bernstein inequality}
\begin{proof}
Fix a pair $(i,j)$. We can write $G^{(i,j)}$ as $\sum_{u = 1}^p G^{(i,j)}_u$. Then the $G^{(i,j)}_u$'s are independent random matrices, as $b_{u,C}$ and $b_{u',C'}$ are independent for $u \ne u'$. We will apply Matrix Bernstein (\cref{thm:matrixbernstein}) to the $G^{(i,j)}_u$'s.

Because all nonzero rows (columns) $S$ in the $G^{(i,j)}_u$'s must have $S \in \cG$, it follows that $\gamma_u(S) \leq \Delta$ for every $u$. In particular, the $\ell_1$ norm of any row (column) in $G^{(i,j)}_u$ is at most $\Delta$, and so $\norm{G^{(i,j)}_u}_2 \leq \Delta$ always holds. 

We now compute the ``variance term'' $\sigma^2$ in \cref{thm:matrixbernstein}. Let $M = \E[\sum_{u = 1}^p G_u^{(i,j)} {G_u^{(i,j)}}^{\top}]$, where the expectation is taken over the $b_{u,C}$'s. The $\ell_1$ norm of the $S$-th row in $M$ is
\begin{flalign*}
\sum_{u = 1}^p \sum_{T \in \cF_i} \sum_{R \in \cF_j} \E[G_u^{(i,j)}(S, R) G_u^{(i,j)}(T, R)] \mper
\end{flalign*}
Because the $b_{u,C}$'s are mean zero, $\E[G^{(i,j)}(S, R) G^{(i,j)}(T, R)]$ is nonzero iff there exist $C, C' \in \cH_u$ with $C \ne C'$ such that $S \oplus R = C^{(1)} \oplus C'^{(2)}$ and either $T \oplus R = C^{(1)} \oplus C'^{(2)}$ or $T \oplus R = C^{(2)} \oplus C'^{(1)}$, and when this occurs the expectation of the corresponding term is at most $1$. (If $\cH$ is non-simple, then the expectation will simply be the \emph{sum} over valid choices for $C,C'$.)
For each $u$, there are at most $\gamma_u(S)$ such $R$'s, and each contributes at most $2$ (for the two different choices of $T$) to the sum. Hence, the $\ell_1$-norm of the $S$-th row in $M$ is at most $2 \sum_{u = 1}^p \gamma_u(S) = 2 \gamma(S)$. As $S \in \cF_i$, we must have $\gamma(S) \leq 2^i d$, and so we have $\norm{M}_2 \leq 2^{i+1} d$.

Swapping the roles of $i$ and $j$, we see that we can take $\sigma^2 = 2 \cdot 2^{\max(i,j)} d$. Applying \cref{thm:matrixbernstein} then yields that with probability $1 - \frac{1}{\log_2^2 m \cdot \poly(N)} \geq 1 - \frac{1}{\log_2^2 m \cdot \poly(n)}$, we have $\norm{G^{(i,j)}}_2 \leq O(\Delta (\log N + \log\log m) + \sqrt{2^{\max(i,j)} d (\log N + \log\log m)})$, which finishes the proof.
\end{proof}

\subsubsection{Proof of \cref{lem:Gspecbound} using trace moment method}
\label{sec:Gspecboundtrace}
\begin{proof}
Let $Z = G^{(i,j)}$ and $Z_u = G_u^{(i,j)}$, and let $r \in \N$.
We observe that $\norm{Z}_2^{2r} \leq \tr((ZZ^{\top})^{r})$. 
We will proceed with the proof in two steps. First, we upper bound $\E[\tr((ZZ^{\top})^{r})]$ by a combinatorial quantity: the number of ``even walk sequences'', which we define below. Then, we bound the number of such sequences.

\begin{definition}
\label{def:walksequence}
Let $S \in \cF_i$. We say that a sequence $(u_1, C_1, C'_1), \dots, (u_{2r}, C_{2r}, C'_{2r})$ with $u_h \in [p]$ and $C_h \ne C'_h \in \cH_{u_h}$ is a ``walk sequence'' for $S$ if the sets $T_{h} := S \oplus \bigoplus_{j < h} (C_j^{(1)} \oplus {C'_j}^{(2)})$ each have size \emph{exactly} $\ell$ and the entries $Z_{u_{2h-1}}(T_{2h - 1}, T_{2h})$ and $Z_{u_{2h}}(T_{2h + 1}, T_{2h})$ are nonzero for each $h = 1, \dots, r$.
Moreover, the sequence is \emph{even} if each $(u,Q)$ appears an even number of times in the multiset $\{(u_h, C_h), (u_h, C'_h)\}_{h \in [2r]}$.
\end{definition}

\begin{proposition}
\label{prop:exptracebound}
$\E[\tr((ZZ^{\top})^{r})] \leq \sum_{S \in \cF_i} \text{\# \{even walk sequences $(u_1, C_1, C'_1), \dots, (u_{2r}, C_{2r}, C'_{2r})$ for $S$\}}$.
\end{proposition}

\begin{lemma}[Sequence counting]
\label{lem:sequencecounting}
For each $S \in \cF_i$, the number of even walk sequences $(u_1, C_1, C'_1), \dots, (u_{2r}, C_{2r}, C'_{2r})$ for $S$ is at most $(4 r)^r  (2^{\max(i,j)} d + r \Delta^2)^{r}$.
\end{lemma}

We observe that \cref{prop:exptracebound,lem:sequencecounting} immediately imply \cref{lem:Gspecbound}. Indeed, we have that
\begin{flalign*}
\E[\tr((ZZ^{\top})^{r})] \leq \abs{\cF_i} (4 r)^r  (2^{\max(i,j)} d + r \Delta^2)^{r} \leq N (4 r)^r  (2^{\max(i,j)} d + r \Delta^2)^{r} \mcom
\end{flalign*}
and hence by Markov's inequality,
\begin{flalign*}
\Pr[\norm{Z}_2 \geq \lambda] \leq \frac{\E[\norm{Z}_2^{2r}]}{\lambda^{2r}} \leq \frac{N (4 r)^r  (2^{\max(i,j)} d + r \Delta^2)^{r}}{\lambda^{2r}} \mper
\end{flalign*}
Taking $r = \ceil{\log_2 N + \log_2 \log_2 m}$ and $\lambda = c \sqrt{r} (\sqrt{2^{\max(i,j)} d + r \Delta^2})$ for a large enough absolute constant $c$ thus implies 
\begin{flalign*}
\Pr[\norm{Z}_2 \geq c \sqrt{2^{\max(i,j)} d r + \Delta^2 r^2}] \leq \frac{N 4^r}{c^{2r}} \leq \frac{1}{\poly(N) \cdot \polylog(m)} \mper
\end{flalign*}
Finally, we observe that $\sqrt{2^{\max(i,j)} d r+ \Delta^2 r^2} \leq \sqrt{2^{\max(i,j)} d r} + \Delta r$, 
which finishes the proof of~\cref{lem:Gspecbound}, as $r \le O(\log N + \log \log m)$.
\end{proof}

We now prove \cref{prop:exptracebound,lem:sequencecounting}.
\begin{proof}[Proof of \cref{prop:exptracebound}]
We compute:
\begin{flalign*}
&\E[\tr((ZZ^{\top})^{r})] = \sum_{(u_1, S_1), \dots, (u_{2r}, S_{2r})} \E[\prod_{h = 1}^{r} Z_{u_{2h - 1}}(S_{2h - 1}, S_{2h}) Z_{u_{2h}}(S_{2h + 1}, S_{2h})] \mcom
\end{flalign*}
where we use the convention that $u_{2r + 1} := u_1$ and $S_{2r + 1} := S_{1}$.
Next, we observe that this is equal to
\begin{flalign*}
&=\sum_{S \in \cF_i} \sum_{(u_1, C_1, C'_1), \dots, (u_{2r}, C_{2r}, C'_{2r}) \ \text{walk sequence for $S$}}\E[\prod_{h = 1}^{r} Z_{u_{2h - 1}}(T_{2h - 1}, T_{2h}) Z_{u_{2h}}(T_{2h + 1}, T_{2h})] \\
&= \sum_{S \in \cF_i} \sum_{(u_1, C_1, C'_1), \dots, (u_{2r}, C_{2r}, C'_{2r}) \ \text{walk sequence for $S$}} \E[\prod_{h = 1}^{r} b_{u_{2h - 1}, C_{2h - 1}} b_{u_{2h - 1}, C'_{2h - 1}} b_{u_{2h}, C_{2h}} b_{u_{2h}, C'_{2h}}] \\
&\leq \sum_{S \in \cF_i} \text{\# even walk sequences $(u_1, C_1, C'_1), \dots, (u_{2r}, C_{2r}, C'_{2r})$ for $S$} \mcom
\end{flalign*}
as the term in the sum is $0$ unless the walk sequence is even.
\end{proof}

\begin{proof}[Proof of \cref{lem:sequencecounting}]
We shall upper bound the number of such sequences for each $S$ via an encoding argument.
For a set $S \in \cF_i$ and $u \in [p]$, we will say that $C,C' \in \cH_u$
\emph{extends} $S$ if $Z_u(S, S \oplus C^{(1)} \oplus {C'}^{(2)})$ is well-defined and non-zero. For $S \in \cF_j$, we make a similar definition, requiring that $Z_u(S \oplus C^{(1)} \oplus {C'}^{(2)}, S)$ is well-defined and non-zero.
The encoding is as follows:
\begin{enumerate}[(1)]
\item Choose $z \in [r]$, the number of \emph{distinct} $u$'s that appear in the sequence. Note that $z$ must be at most $r$ because the sequence is even; $u_h$ cannot appear once in $\{u_1, \dots, u_{2r}\}$, as then we must pair $(u_h, C_h)$ with $(u_h, C'_h)$, but we must have $C_h \ne C'_h$.
\item Choose $2z$ locations $L$ in $[2r]$. These will denote the first and last occurrence of each distinct $u_h$ for $h \in [z]$.
\item Choose a perfect matching $\pi$ for the $2z$ chosen locations. We will think of $\pi$ as a function $\pi \colon L \to [z]$, satisfying $t_1 < t_2 < \dots < t_z$, where $t_h$ is the first preimage of $h$ in $L$ (using the natural ordering on $L$ inherited from $[2r]$). We let $t'_h$ denote the second preimage of $h$ in $L$.
\item Proceed in order of steps $t = 1, \dots, 2r$. We thus know the set $S_t$ that we are currently ``at''. There are three cases.

\begin{enumerate}[(a)]
\item Suppose $t = t_h$ for some $h$. Then, \begin{inparaenum}[(1)] \item choose $u \in [p]$ (that has not yet been chosen); \item choose $C,C' \in \cH_u$ extending $S_t$; \item set the $t$-th element of the sequence to be $(u, C, C')$\end{inparaenum}.

\item Suppose that $t \ne t_h, t'_h$ for all $h \in [z]$. Then, pick a previously chosen $u$ (that has not yet reached its last occurrence according to the matching $\pi$), and pick $C,C' \in \cH_u$ that extends $S_t$. Set the $t$-th element of the sequence to be $(u,C,C')$.

\item Suppose that $t = t'_h$ for some $h$. Then, choose $u = u_h$ and let $C,C' \in \cH_u$ be the \emph{unique} pair that extends $S_t$ and keeps the sequence even. Set the $t$-th element of the sequence to be either $(u,C,C')$ or $(u,C',C)$.
\end{enumerate}
\end{enumerate}
We now count the number of choices. Let us first think of the first $3$ steps as fixed. There are $3$ cases. If we are choosing a new $u$, then there are $\sum_u \gamma_u(S_t) \leq 2^{\max(i,j)} d$ ways to pick $(u,C,C')$. If we are choosing an old $u$, then there are $z \Delta$ ways to pick $(u,C,C')$, as we have $z$ choices for $u$ and then $\gamma_u(S_t) \leq \Delta$ choices for the pair $C,C'$. Finally, if we are at $t = t'_h$ for some $h$, then we have $2$ choices. Hence, across all steps, we have $(2^{\max(i,j)} d)^z \cdot 2^z \cdot (z \Delta)^{2r - 2z}$ choices.

Next, we think of $z$ as fixed, and count the choices for Steps (2) and (3). These have ${2r \choose 2z}$ choices and $\frac{(2z)!}{2^z z!}$ choices, respectively. Combining, we thus have the bound
\begin{flalign*}
&\text{\# $(u_1, C_1, C'_1), \dots, (u_{2r}, C_{2r}, C'_{2r})$ even, well-formed for $S$} \leq \sum_{z = 1}^r {2r \choose 2z} \frac{(2z)!}{2^z z!} (2^{\max(i,j)} d)^z \cdot 2^z \cdot (z \Delta)^{2r - 2z} \mper
\end{flalign*}
We now observe that
\begin{flalign*}
&{2r \choose 2z} \frac{(2z)!}{z!} z^{2r - 2z} = \frac{(2r)!}{(2r - 2z)!  z!} \cdot z^{2r - 2z} \\
&= \frac{(2r)!}{r!r!} \cdot \frac{(r - z)!(r-z)!}{(2r - 2z)!} \cdot \frac{r!}{(r-z)!} \cdot \frac{r!}{z!(r-z)!} \cdot z^{2r - 2z} \\
&\leq 2^{2r} \cdot 1 \cdot r^z \cdot {r \choose z}  \cdot r^{2r - 2z} \\
&\leq (4r)^r {r \choose z} r^{r-z} \mper
\end{flalign*}

Thus, 
\begin{flalign*}
&\sum_{z = 1}^r {2r \choose 2z} \frac{(2z)!}{2^z z!} (2^{\max(i,j)} d)^z \cdot 2^z \cdot (z \Delta)^{2r - 2z} \leq (4 r)^r \sum_{z = 1}^r {r \choose z} (2^{\max(i,j)} d)^z \cdot (r \Delta^2)^{r - z} \\
&\leq  (4 r)^r  (2^{\max(i,j)} d + r \Delta^2)^{r} \mcom
\end{flalign*}
which finishes the proof.
\end{proof}

\subsection{Bounding the number of bad rows: proof of \cref{lem:numbadrows}}
\label{sec:numbadrows}

Let $\U_{\ell}$ be the uniform distribution on subsets of $[n] \times [2]$ of size exactly $\ell$. 
In order to bound the fraction of bad rows (i.e. the size of $|\cB|$), we will analyze the probability that a draw from $\U_{\ell}$ produces a set $S$ that indexes a bad row in the Kikuchi matrix $A$. 

We will do this by viewing $\gamma_u(S)$ as a polynomial of degree $k-1$ in the indicator vector of the set $S$:
\begin{lemma}[Polynomial View of $\gamma_u(S)$]
Let $P_u$ be the following polynomial in variables $\{z_{(i,b)}\}_{i \leq n, b \in \{1,2\}}$:
\[
P_u(z) = \sum_{(C, C') \in \cH_u \times \cH_u, C \ne C'}  \sum_{R \in \cR_{(C,C')}} z_R \mcom
\]
where $z_R := \prod_{(i,b) \in R} z_{i,b}$. 
Then, for every $S \subseteq [n] \times [2]$, we have: $\gamma_u(S) \leq P_u(\1_S)$, where $\1_S$ is the $0$-$1$ indicator of the set $S$ (i.e., $\1_S$ has a $1$ in the $(i,b)$-th coordinate if and only if $(i,b) \in S$).
\end{lemma}

\begin{proof}

By \cref{def:butterfly-degree}, we have:
\begin{equation*}
\gamma_u(S) = \sum_{(C,C') \in \cH_u \times \cH_u, C \ne C'} \sum_{R \in \cR_{(C,C')}} \1(S \cap (C^{(1)} \cup C'^{(2)}) = R) \leq \sum_{(C,C') \in \cH_u \times \cH_u, C \ne C'} \sum_{R \in \cR_{(C,C')}} \1(R \subseteq S) = P_u(\1_S) \mper \qedhere
\end{equation*}
\end{proof}

Thus, it is enough to upper bound the probability of the event $P_u(z) \geq \Delta$ under $\U_{\ell}$. Next, we will switch $\U_{\ell}$ with a more convenient-to-analyze product distribution $\U_{\ell}'$ on $z$. The following lemma argues why this suffices for our purpose:

\begin{lemma}[Switching to a Product Distribution]
Let $\U_{\ell}'$ be the distribution where each element $(i,b)$ in $[n] \times [2]$ is included in $S$ independently with probability $q = \frac{\ell}{2 n} (1+\beta)$ (equivalently, each $z_{i,b}$ is an independent Bernoulli$(q)$ random variable) where $\beta = \max \Paren{ \frac{4}{\ell} \ln( \frac{32 p N}{\eps^2 D}), \sqrt{\frac{4}{\ell} \ln(\frac{32 p N}{\eps^2 D})}}$. Then, for any $\lambda$,
\begin{equation*}
\Pr_{z \gets \U_{\ell}}[P_u(z) > \lambda] \leq \Pr_{z \gets \U'_{\ell}}[P_u(z) > \lambda] + \frac{\epsilon^2 D}{32pN} \mper
\end{equation*}
\label{lem:switch-to-product}
\end{lemma}
Note that under $\U_{\ell}'$, the set sampled does not always have size exactly $\ell$. 
\begin{proof}
To relate the two probabilities, we will couple $\U_{\ell}'$ with $\U_{\ell}$ as follows. First, sample $T \gets \U_{\ell}'$, and then choose $S$ to be a uniformly random subset of $T$ of size exactly $\ell$ (if $\abs{T} < \ell$, then abort). Let $\cJ$ be the joint distribution induced by this coupling. By Chernoff bound, we have for every $\delta \in [0,1]$,
\begin{flalign*}
\Pr_{T \sim \U_{\ell}'}[\abs{T} < (1 - \delta)(1 + \beta) \ell] \leq \exp\Paren{\frac{\delta^2 \ell (1 + \beta)}{2}} \mper
\end{flalign*}
Setting $\delta = 1 - \frac{1}{1 + \beta}$, we see that $\Pr_{T \sim \U_{\ell}'}[\abs{T} < \ell] \leq \frac{\epsilon^2D}{32p N}$, as $\frac{\beta^2}{1 + \beta} \geq \frac{2}{\ell} \ln(\frac{32 p N}{\eps^2 D})$, by choice of $\beta$.

We also observe that $P_u(T) \geq P_u(S)$ for \emph{any} $S \subseteq T$. In particular, if we first sample $T \gets \cD'$ and $P_u(T) \leq \lambda$, then it also holds that $P_u(S) \leq \lambda$, regardless of the choice of $S$. We thus have
\begin{equation*}
\Pr_{S \gets \U_{\ell}}[P_u(S) > \lambda] \leq \Pr_{(S,T)\sim \cJ}[P_u(T) > \lambda \mid \abs{T} \geq \ell] \leq \Pr_{T \gets \U_{\ell}'}[P_u(T) > \lambda] + \frac{\epsilon^2D}{32pN} \mper \qedhere
\end{equation*}
\end{proof}

We now finish the proof of \cref{lem:numbadrows} by analyzing $\Pr_{\U_{\ell}'} [ P_u(z) \geq \Delta]$:
\begin{proof}[Proof of \cref{lem:numbadrows}]

In order to bound the probability that $P_u(z) \geq \Delta$ under $\U_{\ell}'$, let's apply the polynomial concentration inequality (\cref{fact:schudy-sviridenko}). Let's first bound $\E_{\U_{\ell}'}[P_u(z)]$. Let $q = \frac{\ell}{2n} (1+\beta)$ as in \cref{lem:switch-to-product}. We have
\begin{equation*}
\E_{z \gets \U_{\ell}'}[P_u(z)] = \sum_{(C,C') \in \cH_u \times \cH_u, C \ne C'} q^{k-1} \cdot \abs{\cR_{(C,C')}} = w_kq^{k-1} \abs{\cH_u} (\abs{\cH_u} - 1) \mcom
\end{equation*}
where $w_k := {k-1 \choose \frac{k-1}{2}}^2$ if $k$ is odd and $w_k := 2 {k-1 \choose \frac{k}{2}} {k-1 \choose \frac{k-2}{2}}$ if $k$ is even.

Let $\eta = 4 \ln \left(\frac{32 p N}{\eps^2 D}\right)$. Notice that $\eta \geq 4 \ln 32 \geq 1$, as $N/D \geq 1$, $p \geq 1$, and $\eps < 1$ all hold. We also observe that $\frac{1 + \beta}{2} \leq \eta$.

Recall that by regularity of the polynomial $\psi$ (described by $\{\cH_u\}_{u \in [p]}$), we have that $\deg_u(Q) \leq \frac{1}{\eps^2}\Paren{\frac{n}{\ell}}^{\frac{k}{2} - 1 - \abs{Q}}$ for all $Q \subseteq [n]$, $\abs{Q} \leq \frac{k-2}{2}$. In particular, this means that $\abs{\cH_u} = \deg_u(\emptyset) \leq  \frac{1}{\eps^2}\Paren{\frac{n}{\ell}}^{\frac{k}{2} - 1}$, and thus
\begin{equation*}
\E_{z \gets \U_{\ell}'}[P_u(z)] \leq w_k \left(\frac{1+\beta}{2}\right)^{k-1} \frac{1}{\eps^4} \frac{\ell}{n} \leq w_k \eta^{k-1} \frac{1}{\eps^4} \frac{\ell}{n} \mper
\end{equation*} 

We now compute the parameters $\nu_r$ for $r = 0, \dots, k-1$ that appear in the statement of \cref{fact:schudy-sviridenko}. We have
\begin{equation*}
\nu_r = \max_{R \subseteq [n] \times [2], \abs{R} = r} \sum_{(C,C') \in \cH_u \times \cH_u, C \ne C'} \sum_{R' \in \cR_{(C,C')}} \1(R \subseteq R') \cdot q^{k-1 - \abs{R}} \mper
\end{equation*}
Letting $R_1$ and $R_2$ denote $R \cap {[n] \times \{1\}}$ and $R \cap {[n] \times \{2\}}$, we see that if $R \subseteq R'$ and $R' \in \cR_{(C,C')}$, then this implies that $R \subseteq C^{(1)} \cup C'^{(2)}$, and that (if $k$ is odd) $\abs{R_1}, \abs{R_2} \leq \frac{k-1}{2}$ or (if $k$ is even) $\abs{R_1}, \abs{R_2} \leq \frac{k}{2}$. For each $R$, the number of $C^{(1)} \cup C'^{(2)}$ such that $R \subseteq C^{(1)} \cup C'^{(2)}$ is at most $\deg_u(R_1) \deg_u(R_2)$, and the number of $R'$ with $R \subseteq R' \subseteq C^{(1)} \cup C'^{(2)}$ is at most $\abs{\cR_{(C,C')}} = w_k$. We thus have
\begin{flalign*}
&\nu_r \leq w_k q^{k-1 - r} \max_{R_1, R_2 \subseteq [n], \abs{R_1} + \abs{R_2} = r, \abs{R_1}, \abs{R_2} \leq \frac{k-1}{2}}  \deg_u(R_1) \deg_u(R_2) \ \ \text{(if $k$ is odd)} \\
&\nu_r \leq w_k q^{k-1 - r} \max_{R_1, R_2 \subseteq [n], \abs{R_1} + \abs{R_2} = r, \abs{R_1}, \abs{R_2} \leq \frac{k}{2}}\deg_u(R_1) \deg_u(R_2)  \ \ \text{(if $k$ is even)} \mper
\end{flalign*}

Fix $R_1, R_2$ that maximize the above expression. Because the $\cH_u$'s are $(\eps, \ell)$-regular, we have that $\deg_u(R_b) \leq \frac{1}{\eps^2}\Paren{\frac{n}{\ell}}^{\frac{k}{2} - 1 - \abs{R_b}}$ if $\abs{R_b} \leq \frac{k - 2}{2}$, and $\deg_u(R_b) \leq \frac{1}{\eps^2}$ if $\abs{R_b} = \frac{k-1}{2}$ (if $k$ odd) or $\frac{k}{2}$ (if $k$ even). So, if $\abs{R_b} \leq \frac{k-2}{2}$, then it holds that
\begin{equation*}
\deg_u(R_b) q^{\frac{k-1}{2} - \abs{R_b}} \leq \frac{1}{\eps^2} 
 \left(\frac{1+\beta}{2}\right)^{\frac{k-1}{2} - \abs{R_b}}  \cdot \Paren{\frac{n}{\ell}}^{\frac{k}{2} - 1 - \abs{R_b} - \frac{k-1}{2} + \abs{R_b}} = \frac{1}{\eps^2}  \left(\frac{1+\beta}{2}\right)^{\frac{k-1}{2} - \abs{R_b}}  \cdot \sqrt{\frac{\ell}{n}} \mper
\end{equation*}
If $\abs{R_b} = \frac{k-1}{2}$ (and thus $k$ is odd), we also have
\begin{equation*}
\deg_u(R_b) q^{\frac{k-1}{2} - \abs{R_b}} = \deg_u(R_b) \leq \frac{1}{\eps^2} \mcom
\end{equation*}
which implies that for $k$ odd, 
\begin{equation*}
\nu_r \leq w_k \frac{1}{\eps^4} \eta^{k-1} \mper
\end{equation*}
Let us now upper bound $\nu_r$ when $k$ is even. We either have $\abs{R_1}, \abs{R_2} \leq \frac{k - 2}{2}$, in which case $q^{k-1 -r} \deg_u(R_1) \deg_u(R_2) \leq \frac{1}{\eps^4} \eta^{k-1}$ trivially holds. Otherwise, suppose that one of $R_1$ or $R_2$ has size exactly $\frac{k}{2}$. Note that exactly one of $R_1, R_2$ can have size $\frac{k}{2}$, as $r \leq k-1$. Without loss of generality, let us suppose that $\abs{R_1} = \frac{k}{2}$, so that $\abs{R_2} = r - \frac{k}{2} \leq \frac{k-2}{2}$. We then have
\begin{flalign*}
q^{k-1-r} \deg_u(R_1) \deg_u(R_2) \leq \frac{1}{\eps^4}\left(\frac{1+\beta}{2}\right)^{k-1-r}\Paren{\frac{n}{\ell}}^{\frac{k}{2} - 1 - (r - \frac{k}{2}) - (k-1-r)} = \frac{1}{\eps^4}\left(\frac{1+\beta}{2}\right)^{k-1-r} \leq \frac{1}{\eps^4}\eta^{k-1} \mper
\end{flalign*}

Now, taking $\lambda =  w_k \frac{1}{\eps^4}\eta^{k-1} c^{k-1} \ln^{k-1}\left(\frac{32 p N}{\eps^2 D}\right)$ for some absolute constant $c$ and applying \cref{fact:schudy-sviridenko}, we get that
\begin{equation*}
\Pr_{z \gets \U_{\ell}'}\left[P_u(z) >  2w_k \frac{1}{\eps^4}\eta^{k-1} c^{k-1} \ln^{k-1}\left(\frac{32 p N}{\eps^2 D}\right)\right] \leq \frac{\epsilon^2D}{32pN} \mper
\end{equation*}
\cref{lem:numbadrows} now follows by a union bound on the $p$ different $u$'s and \cref{lem:switch-to-product}, and observing that
\begin{flalign*}
2w_k \frac{1}{\eps^4}\eta^{k-1} c^{k-1} \ln^{k-1}\left(\frac{32 p N}{\eps^2 D}\right) \leq {c'}^{k-1} \frac{1}{\eps^4}  \ln^{2(k-1)}\left(\frac{32 p N}{\eps^2 D}\right) = \Delta\mper
\end{flalign*}
where $c'$ is an absolute constant, as $\eta = 4 \ln \left(\frac{32 p N}{\eps^2 D}\right)$. 
\end{proof}


\section{Strong CSP Refutation: Smoothed via Semirandom} \label{sec:smoothed}
In this section, we show how the tight refutation of semirandom sparse polynomials in \cref{sec:poly-refute} can be used in a black-box way to derive nearly optimal algorithms for strongly refuting smoothed CSPs and, as a special case, semirandom CSPs.  

\parhead{Smoothed model.}
Let us first formally describe the model of smoothed Boolean CSPs.
\begin{definition}[Smoothed CSP Instances~\cite{Fei07}]
Let $k \in \N$. 
Let $\psi$ be an instance of a CSP with predicate $P:\on^k \rightarrow \zo$ specified by a collection of $k$-tuples $\cH$ and literal patterns $\Xi$.
Let $\vec{p} = \{p_{C,i}\}_{C \in \cH, i \in [k]}$ with each $p_{C,i} \in [0,1]$ be smoothing parameters, one for every $C \in \cH$ and $i \in [k]$. 
A $\vec{p}$-smoothing of $\psi$ is obtained as follows:
\begin{enumerate}
\item For every $C \in \cH$, let $S_C \subseteq [k]$ be obtained by adding $i$ to $S_C$ with probability $p_{C,i}$ independently for every $i \in C$. 
\item For every $i \in S_C$, reset $\Xi(C,i)$ to be a uniform and independent random bit in $\pm 1$.
\end{enumerate}

\end{definition}
\begin{remark}
\begin{enumerate}
	\item The notion of smoothing allows using a different probability of ``rerandomizing'' each of $mk$ literals in a $k$-CSP instance $\psi$ with $m$ constraints. 
	\item The two-step random process above is equivalent to flipping the negation pattern $\Xi(C,i)$ of the $i$-th literal in clause $C \in \cH$ independently of others with probability $p_{C,i}/2$. 
	\item Setting $p_{C,i} = 1$ for every $i,C$ yields the model where the literal patterns are uniformly random and independent in $\{\pm 1\}$. This is the semirandom model of CSPs.
\end{enumerate}
\end{remark}

We now proceed to state and prove our main results concerning refutation of smoothed instances, along the way noting also a better bound for the special semirandom case. We recall the notion of $t$-wise uniform distributions before presenting the main result. 
\begin{definition}[$t$-wise uniform distribution] \label{def:t-wise-uniform}
A probability distribution $\mu$ on $\on^k$ is said to be $t$-wise uniform if $\E_{z \sim \mu} \prod_{i \in S} z_i = 0$ for every $S \subseteq [k]$ of size $|S| \leq t$. 
\end{definition}

\begin{theorem}[Smoothed Boolean CSP Refutation]
\label{thm:smoothed-refutation}
Let $P:\on^k \rightarrow \zo$ be a $k$-ary Boolean predicate such that there is no $t$-wise uniform distribution supported on $P^{-1}(1)$. Let $\ell$ be an integer with $2(k-1) \leq \ell \leq n$.
There is an algorithm that takes as input an instance $\Theta$ of CSP($P$)  
and outputs a value $\algval(\Theta) \in [0,1]$ in time $n^{O(\ell)}$ satisfying the following:
\begin{enumerate}[(1)]
 \item $\val(\Theta) \leq \algval(\Theta) \leq 1$.
 \item Suppose the input instance $\Theta$ is a smoothing $\psi_s$ of an arbitrary CSP instance $\psi = (\cH, \Xi)$ with $n$ variables and $m$ constraints w.r.t.\ a vector of smoothing parameters $\vec p = \{p_{C,i}\}$ in $[0,1]$. Suppose that $m \geq \tfrac{2m_0}{q(\vec{p})}$, where 
\begin{equation*}
m_0=\frac{2^{O(k)} (\log_2 n)^{4t + 1}}{\eps^5} \cdot \ell \Paren{\frac{n}{\ell}}^{\frac{t}{2}}
\end{equation*}
and
\begin{equation}
    \label{eq:q(p)}
    q(\vec{p}) = \frac{1}{m} \sum_{C \in \cH} \prod_{i \in C} p_{C,i} \mper
\end{equation}
Then with probability at least $1-1/\poly(n)$ over the randomness of the smoothening process, it holds that $\algval(\Theta) \leq 1- \frac{q(\vec p)}{2} \cdot (\delta_t - \epsilon) + 2^{-n} $. Here, $\delta_t \geq 2^{-\tilde{O}(k^t)}$ depends only on the predicate $P$.

Furthermore, in the semirandom case (where all $p_{C,i} = 1$), we have $\algval(\Theta) \leq 1- \delta_t + \epsilon + 2^{-n}$ with probability $1-1/\poly(n)$.
\end{enumerate}
Moreover, the algorithm is captured by the canonical degree $2\ell$ sum-of-squares relaxation of the CSP maximization problem over the hypercube. 
\end{theorem}

The following result, proved in \cite{AOW15} using LP duality, plays a crucial role in our proof of the above theorem, by allowing us to bound the value of CSP with predicate $P$ that does not support a $t$-wise uniform distribution by a degree-$t$ polynomial as proxy.

\begin{fact}[Separating Polynomials, Lemma 3.16 and Theorem 4.10 in ~\cite{AOW15}]
Let $P:\on^k \rightarrow \zo$ be a predicate such that there is no $t$-wise uniform distribution supported on $P^{-1}(1)$. Then, there is a $\delta_t \geq 2^{-\tilde{O}(k^t)}$ such that for every $t$-wise uniform distribution $\zeta$, $\E_{\zeta}[P] \leq 1-\delta_t$. Furthermore, there is a degree-$t$ polynomial $Q:\on^k \rightarrow \R$ such that $Q(z) = \sum_{T \subseteq [k]} \hat{Q}(T) z_T$ and:
\begin{enumerate}
\item $P(z) \le 1 - \delta_t + Q(z)$ for every $z \in \on^k$
\item $\hat{Q}(\emptyset) = 0$, i.e. $Q$ has no constant coefficient, and, 
\item $\sum_{T \subseteq [k]} |\hat{Q}(T)| \leq 2^{2k}$. 
\end{enumerate}
\label{fact:AOW-dual-poly}
\end{fact}

We now turn to the task of proving \cref{thm:smoothed-refutation}.

\subsection{Proof of \cref{thm:smoothed-refutation}}
By \cref{fact:sosalg}, there is an algorithm that in $n^{O(\ell)}$-time outputs a value $\algval(\Theta) \in [0,1]$ such that $\beta \leq \algval(\Theta) \leq \beta + 2^{-n}$, where $\beta = \max \pE[\Theta]$, $\Theta(x) := \sum_{C \in \cH} P(\Xi(C,1)x_{C_1}, \dots, \Xi(C,k)x_{C_k})$ is a degree $\leq 2k$ polynomial, and the maximum is taken over degree-$2\ell$ pseudo-expectations $\pE$ over $\on^n$.  Note that $\Theta$ is indeed a degree $\leq 2k$ polynomial, as $P$ can always be expressed as a degree $\leq 2k$ polynomial. 

First, we observe that Item (1), i.e., completeness, is completely trivial: simply take $\pE$ to be the expectation $\E_{\mu}$ of a distribution $\mu$ supported only on optimal solutions to $\Theta$. Indeed, this implies that $\val(\Theta) \leq \beta \leq \algval(\Theta)$.
We thus focus on proving Item (2).

We will analyze the smoothing random process using the two steps that define it. 
Let us first consider the event that the first step chooses to re-randomize \emph{all} the literals in a given clause $C \in \cH$; the probability of this event is $\prod_{i = 1}^k p_{C,i}$. Let $\cG$ be the set of clauses for which this occurs. Observe that the $0$-$1$ indicator of ``all literals are chosen to be re-randomized in $C$'' is independent across clauses $C \in \cH$. The expected number of clauses in $\cG$ equals $m q(\vec p) = \sum_{C \in \cH} \prod_{i = 1}^k p_{C,i}$. Thus, by Chernoff bound, $\abs{\cG} \geq 0.5 m q(\vec p)$ with probability at least $1-e^{-m q(\vec{p})/8} \geq 1 - e^{-m_0/4} \geq 1 - 1/\poly(n)$, as $m q(\vec{p}) \geq 2 m_0$. Let us proceed assuming that $\abs{\cG} \geq 0.5 m q(\vec p)$.

Let $\Xi$ denote the literal patterns after re-randomizing. We see that for every $C \in \cG$ and $i \in [k]$, $\Xi(C,i)$ is drawn uniformly and independently from $\on$. We shall view $\Xi(C,i)$ as fixed for all $C \notin \cG, i \in [k]$, and think of the $\Xi(C,i)$'s for $C \in \cG, i \in [k]$ as being random. For $C \in \cG$, let $r_{C,i}$ denote the random variable $\Xi(C,i)$, which is uniformly random in $\on$.

Let 
\begin{flalign*}
&\psi_g = \frac{1}{|\cG|}\sum_{C \in \cG} P(r_{C_1} x_{C_1}, \ldots, r_{C_k} x_{C_k}) \mcom\\
&\psi_b = \frac{1}{|\cH| - \abs{\cG}}\sum_{C \not \in \cG} P(\Xi(C,1) x_{C_1}, \ldots, \Xi(C,k) x_{C_k}) \mcom
\end{flalign*}
so that $\abs{\cH}\psi_s = \abs{\cG}\psi_g + (\abs{\cH} - \abs{\cG})\psi_b$.
Thus, by linearity of pseudo-expectations, we must have that for any pseudo-expectation $\pE$, 
\begin{equation}\label{eq:good-bad-split}
\pE[\psi_s] \leq \frac{\abs{\cG}}{\abs{\cH}}|\pE[\psi_g]| + (1 - \frac{\abs{\cG}}{\abs{\cH}})|\pE[\psi_b]|\mper
\end{equation}
Note that $\psi_g$ and $\psi_b$ are not known to our algorithm; these quantities appear only in our analysis. 

Now, we know that for every $x$, $P(\Xi(C,1) x_{C_1}, \ldots, \Xi(C,k) x_{C_k}) \leq 1$. As $P$ is a degree $k$ polynomial on $k$ variables, by \cref{fact:sos-completeness}, for every pseudo-expectation $\pE$ of degree $2\ell \geq 2k$, $\pE[P(\Xi(C,1) x_{C_1}, \ldots, \Xi(C,k) x_{C_k})] \leq 1$. 
Using linearity of $\pE$ and adding up the inequalities above for $C \not \in \cG$ yields that:
\begin{equation} \label{eq:bound-on-bad}
\pE[\psi_b] \leq 1 \mper
\end{equation}

Let us now analyze $\pE[\psi_g]$. 
First, we invoke \cref{fact:AOW-dual-poly} to conclude that for every $x$, it holds that:
\[
P(r_{C,1}x_{C_1}, \ldots, r_{C,k} x_{C_k}) \leq 1-\delta_t + Q(r_{C,1}x_{C_1}, \ldots, r_{C,k} x_{C_k})\mper
\]
As $\deg(Q) = t \leq k$, by \cref{fact:sos-completeness} and summing up over $C \in \cG$, for every pseudo-expectation of degree~$2\ell \geq 2k$, we must have that:
\[
\pE[\psi_g] \leq 1-\delta_t + \frac{1}{|\cG|} \sum_{C \in \cG} \pE[Q(r_{C,1}x_{C_1}, \ldots, r_{C,k} x_{C_k})]\mper
\]
Next, let $T \subseteq [k]$ of size $\leq t$. 
For each $C$, let $x_{C \vert_T} = \prod_{i \in T} x_{C_{i}}$ and $b_{C \vert_T} = \Pi_{i \in T} r_{C,i}$. 
Observe that $Q(z) = \sum_{0<|T|\leq t} \hat{Q}(T) z_T$ from \cref{fact:AOW-dual-poly} and that further, $\sum_{0 <|T|\leq t} |\hat{Q}(T)| \leq 2^{2k}$. Thus, we have:
\[
\pE[\psi_g] \leq 1-\delta_t + \frac{1}{|\cG|} \sum_{C \in \cG} \sum_{T \subseteq [k],0< |T| \leq t} |\hat{Q}(T)| b_{C \vert_T} \pE\bigl[ x_{C \vert_T} \bigr] \mper
\]
Define $\phi_{T}$ to be the homogenous degree $|T|$ polynomial described by: 
\[
\phi_T (x) = \frac{1}{|\cG|} \sum_{C\in \cG} {b}_{C \vert_T} x_{C \vert_T}
\]
Then, notice that:
\begin{equation} \label{eq:split-dual-poly}
\pE[\psi_g] \leq 1-\delta_t + \sum_{T \subseteq [k],0< |T| \leq t} |\hat{Q}(T)|\pE[\phi_T] \mper
\end{equation}
We now observe that each $\phi_T$ is a polynomial with independent random coefficients in $\{-1,1\}$. Further, since $|\cG| \geq 0.5 q(\vec p) m \geq m_0$, by \cref{thm:mainpolyrefute}, with probability at least $1-1/\poly(n)$, we must have that for every pseudo-expectation $\pE$ of degree at least $2\ell$, 
\[
\pE [\phi_T] \leq \frac{\epsilon}{2^{2k}}\mper
\]

By a union bound over $\leq 2^k$ possible $T$, this bound holds for every $T$ with probability at least $1-1/\poly(n)$. Conditioning on this event, combining with \eqref{eq:split-dual-poly}, and using that $\sum_T |\hat{Q}(T)|\leq 2^{2k}$ gives:
\begin{equation} 
\pE[\psi_g] \leq 1-\delta_t  + \epsilon \mper
\end{equation}

Thus, plugging this bound into \eqref{eq:good-bad-split} and using \eqref{eq:bound-on-bad} yields:
\begin{equation} 
\pE[\psi_s] \leq \Paren{1-\frac{|\cG|}{|\cH|}} \cdot 1 + \frac{|\cG|}{|\cH|} \cdot (1-\delta_t +\epsilon) \leq 1 - \frac{\abs{\cG}}{\abs{\cH}} (\delta_t-\epsilon) \leq 1- (\delta_t-\epsilon) \cdot \frac{q(\vec p)}{2}  \mcom
\end{equation}
where we use that $\frac{\abs{\cG}}{\abs{\cH}} \geq q(\vec{p})/2$. Note that here we require $\delta_t \geq \eps$, although the conclusion is trivial if this does not hold. As $\algval(\psi_s) \leq \beta + 2^{-n} \leq 1- (\delta_t-\epsilon) \cdot \frac{q(\vec p)}{2}  + 2^{-n}$, this completes the proof for the smoothed case.

As the semirandom model is the special case of the smoothed model (where $p_{C,i} = 1$ for every $i$), the above argument directly yields an upper bound of $\pE[\psi] \leq 1-0.5 (\delta_t - \epsilon) + 2^{-n}$ for the case of semirandom instances.
However, we incurred the $0.5$ factor entirely due to the probabilistic bound on $|\cG|$, and in the semirandom setting, $|\cG| =|\cH|$ with probability $1$. Hence, for semirandom refutation, we do not lose this extra $0.5$ factor.
\section{Proof of Feige's Conjecture: Even Covers in Hypergraphs} \label{sec:fko}

In this section, we prove Feige's conjecture, that every $k$-uniform hypergraph with a certain number of hyperedges has a short even cover. In the next section, we will use it to establish (using Feige, Kim and Ofek's ideas) that there \emph{exist} polynomial size refutations for arbitrary \emph{semirandom} instances of 3-SAT at a density $m = \tilde{\Omega}(n^{1.4})$ which is $\tilde{O}(n^{0.1})$ factor smaller than the spectral threshold of $n^{1.5}$ for refuting random instances. An appropriate generalization of this result holds for $k$-SAT and more generally any CSP.

We begin by defining even covers.
\begin{definition}[Even (multi)covers]
Let $\cH$ be a $k$-uniform hypergraph on $[n]$. A set of \emph{distinct} hyperedges $C_1, C_2, \ldots, C_r  \in \cH$ is said to be an \emph{even cover of length $r$} in $\cH$ if every element $j \in [n]$ belongs to an even number of $C_i$'s; equivalently, $\oplus_{i = 1}^r C_i = \emptyset$.  An even \emph{multicover} in $\cH$ is exactly the same except $C_1, C_2, \ldots, C_r \in \cH$ need not be distinct. Even (multi)covers are defined similarly for bipartite hypergraphs, using the hyperedges $(u,C)$.
\end{definition}
 We note that if $\cH$ is not simple, i.e., $\cH$ is a multi-set, then $\cH$ trivially has an even cover of length~$2$. Indeed, $\cH$ must contain distinct elements $C_1$ and $C_2$ that are equal as sets, and so $C_1 \oplus C_2 = \emptyset$. 

 The main result of this section is a proof of Feige's conjecture (\cref{conj:feige}) up to $\poly \log n$ factor loss in the number of hyperedges $m$ in the hypergraph.
\begin{theorem}[Resolution of Feige's Conjecture] \label{thm:feige-conjecture}
Let $k \in \N$ and $\ell = \ell(n)$ with $2(k-1) \leq \ell \leq n$.
Let $\cH$ be a $k$-uniform hypergraph on $[n]$ with $m \geq  \Gamma^k \cdot n \Paren{\frac{n}{\ell}}^{\frac{k}{2}-1} \log^{4k+1} n$ hyperedges, where $\Gamma$ is an absolute constant.
Then, $\cH$ contains an even cover of size $O(\ell \log n)$. 
\end{theorem}

Our proof closely mimics the steps taken in \cref{sec:decomposition,sec:poly-refute,sec:partition-refute} on the way to obtaining an efficient refutation algorithm for semirandom sparse multilinear polynomials. In the first step, we observe that without loss of generality, we can assume that $\cH$ is a simple, $p$-bipartite, $(\epsilon, \ell)$-regular hypergraph for $\epsilon = 1/4$. 

\begin{lemma}[Reduction to Simple, $p$-bipartite, $(1/4,\ell)$-regular hypergraphs] \label{lem:fko-simple-bipartite}
Fix $k,\ell= \ell(n) \in \N$ with $2(k-1) \leq \ell \leq n$. 
Suppose that for every $p$-bipartite, $(1/4,\ell)$-regular, simple $k$-uniform hypergraph $\cH = \{\cH_u\}_{u \in [p]}$ with $m \geq   \max\{c^k \Paren{\frac{n}{\ell}}^{\frac{k-1}{2}} \sqrt{p \ell} \log^{2k+0.5} n, 16p\}$ hyperedges for some absolute constant $c$ and $\abs{\cH_u} = \frac{m}{p}$ for all $u$, there exists an even cover in $\cH$ of length at most $r$. Then, every $k$-uniform hypergraph $\cH$ with $m \geq \Gamma^k \cdot n \Paren{\frac{n}{\ell}}^{\frac{k}{2}-1} \log^{4k+1} n$ hyperedges has an even cover of length at most $r$. 
\end{lemma}

\begin{proof}
Let $\cH$ be an arbitrary $k$-uniform hypergraph. 
First, note that if $\cH$ is not simple, we are immediately done since any pair of parallel hyperedges yields an even cover of size $2$. We thus assume that $\cH$ is simple. Apply the decomposition algorithm from \cref{lem:decomposition} to $\cH$ to get bipartite hypergraphs $\cH^{(1)}, \dots, \cH^{(k)}$; these hypergraphs must be simple, as $\cH$ was. As $\sum_{t = 1}^k m^{(t)} = m$, there must exist some $t$ with $1 \leq t \leq k$ such that $m^{(t)} \geq m/k$. As $m^{(1)} \leq \eps m/k$ always holds, we must have $t \ne 1$.
The bound on $m^{(t)}/p^{(t)}$ in \cref{lem:decomposition} implies that $m^{(t)} \geq m/k \geq \max\{c^k \Paren{\frac{n}{\ell}}^{\frac{k-1}{2}} \sqrt{p^{(t)} \ell} \log^{2k+0.5} n, 16p^{(t)}\}$. Thus, the $p^{(t)}$-bipartite $(1/4,\ell)$-regular hypergraph $\cH^{(t)}$ must contain an even cover, say $(u_1, C_1), \ldots (u_{r'}, C_{r'})$ for some $r' \leq r$. From \cref{lem:decomposition}, for each $u_i$, there is a $Q_i$ such that each hyperedge $(u_i, C_i)$ in $\cH^{(t)}$ is a bipartite contraction of the unique hyperedge $(Q_i \cup C_i)$ in $\cH$. We then observe that $(Q_1 \cup C_1), \ldots, (Q_{r'} \cup C_{r'})$ is trivially an even cover of length $r' \leq r$ in $\cH$, which finishes the proof. 
\end{proof}
\ignore[using AHL]{
Next, we observe that we can further assume that the
$p$-bipartite, $(1/4,\ell)$-regular simple $k$-uniform hypergraph $\cH = \{\cH_u\}_{u \in [p]}$ has the property that any $C \subseteq [n]$ of size $k-1$ does not appear in too many $\cH_u$'s.

To do this, we will use  the following classical Moore bound from extremal graph theory.
\begin{fact}[See ~\cite{AHL02}]
Let $G$ be a graph on $n$ vertices with average degree $d>2$. Then, $G$ has a cycle of length $\leq 2\log_{d-1} n$. 
\end{fact}

\begin{lemma}
\label{lem:no-repeated-hyperedges}
Suppose $\cH = \{\cH_u\}_{u \in [p]}$ is a simple $k$-uniform bipartite hypergraph with $p \leq n^k$ partitions such that $\cH$ contains at least $t \geq 4p$ pairs of hyperedges $\{(u_{i_1}, C_{i}), (u_{i_2}, C_i)\}_{i \leq t}$. Then, $\cH$ has an even cover of length $\leq 4k\log_7 n$. 
\end{lemma}

\begin{proof}
Consider the ``repeated pairs'' provided in the statement of the lemma. Suppose first that there are $i,j$ such that the sets $\{u_{i_1}, u_{i_2}\} = \{u_{j_1}, u_{j_2}\}$. Then, observe that $(u_{i_1}, C_i), (u_{i_2}, C_i), (u_{j_1}, C_j), (u_{j_2}, C_j)$ is an even cover of length $4$. If not, then the graph where the edges are pairs $\{u_{i_1}, u_{i_2}\}$ as $i$ varies over $i = 1, \dots, 4p$ is a graph with $4p$ edges on $p$ vertices. Thus, it must have a cycle $(\{i^{(1)}_1, i^{(1)}_2\}, , \ldots, \{i^{(r)}_1, i^{(r)}_2\})$ of length $r \leq 2 \log_7 p$, as the average degree is $2 \cdot 4p/p = 8$. Then, observe that the hyperedges $\{(u_{i^{(j)}_b}, C_{i^{(j)}})\}_{j \leq r, b \in \{1,2\}}$ form an even cover of length at most $2r \leq 4 \log_7 p \leq 4k \log_7 n$ where we use that $p \leq n^k$. 
\end{proof}
}
This brings us to the crux of the argument presented in the following lemma.
\begin{lemma}[No even covers implies refutation for semirandom polynomials on regular bipartite hypergraphs] \label{lem:partitioned-poly-refute-odd-non-det}
Fix an odd $k \in \N$ and  $\ell = \ell(n)$ with $2(k-1) \leq \ell \leq n$.
Let $\cH = \{ \cH_u\}_{u \in [p]}$ be a $p$-bipartite $(1/4,\ell)$-regular simple $k$-uniform hypergraph with $m \geq m_0 = \max(c^k  \Paren{\frac{n}{\ell}}^{\frac{k-1}{2}} \sqrt{p \ell} \cdot \log^{2k+0.5} n, 16p)$ hyperedges, where $c$ is an absolute constant, and $\abs{\cH_u} = \frac{m}{p}$ for all $u$. Let $\psi$ be the polynomial $\frac{1}{m} \sum_{u \in [p]} \sum_{C \in \cH_u} b_{u,C} y_u x_C$ for arbitrary $b_{u,C} \in \{-1,1\}$. Suppose that $\cH$ has no even covers of length $\leq O(\ell \log n)$. Then, $\val(\psi) \leq 0.5$. 
\end{lemma}
Observe that this lemma has an absurd conclusion. Clearly, if one sets  $b_{u,C} = 1$ for all $u,C$, then $\val(\psi)$ is trivially $1$: simply set $x = 1^n$ and $y = 1^p$. Thus, this lemma immediately gives a contradiction, in that $\cH$ must admit an even cover of length $O(\ell \log n)$.

The reason we state the (somewhat absurd) lemma is because as we will see, our proof mimics our refutation argument from \cref{sec:partition-refute} and shows that we can essentially carry out all the steps for \emph{arbitrary} $b_{u,C}$'s as long as we can assume that $\cH$ has no even covers of length $O(\ell \log n)$. \cref{lem:partitioned-poly-refute-odd-non-det} effectively captures this argument and, in our opinion, is the most enjoyable way to present it.

It is easy to finish the proof of~\cref{thm:feige-conjecture} assuming the \cref{lem:partitioned-poly-refute-odd-non-det}. 

\begin{proof}[Proof of~\cref{thm:feige-conjecture}]
By \cref{lem:fko-simple-bipartite}, we can assume that $\cH := \cup_{u \in [p]} \cH_{u}$ is a $(1/4,\ell)$-regular, simple, $k$-uniform bipartite hypergraph with $p \leq n^k$ partitions and $m \geq m_0$ hyperedges.

Suppose for the sake of contradiction that the hypergraph $\cH$ has no even cover of length $O(\ell \log n)$. We set $b_{u,C} = 1$ for every $u,C$, and consider the polynomial $\psi = \frac{1}{\abs{\cH'}} \sum_{u \in [p]} \sum_{C \in \cH_u} b_{u,C} y_u x_{C}$ in $x,y$. Observe that by setting $x = 1^n, y = 1^p$, we obtain that $\val(\psi) = 1$. On the other hand, applying~\cref{lem:partitioned-poly-refute-odd-non-det} to $\psi$ yields that $\val(\psi) \leq 0.5$. This is a contradiction, and so $\cH$ must have an even cover of length $\leq O(\ell \log n)$.
\end{proof}

We now focus on the proof of \cref{lem:partitioned-poly-refute-odd-non-det}. 

\subsection{Proof of \cref{lem:partitioned-poly-refute-odd-non-det}}
\ignore[using AHL]{First, by \cref{lem:no-repeated-hyperedges} we can assume that there are $t\leq 4p$ pairs of hyperedges $\{(u_{i_1}, C_{i}), (u_{i_2}, C_i)\}_{i \leq t}$. We delete all hyperedges that appear in the above $t$ pairs, which removes at most $4p \leq 4m/16 \leq m/4$ hyperedges. The resulting hypergraph $\cH'$ now satisfies that for every $u \neq u'$, $\cH'_u \cap \cH'_{u'} = \emptyset$, and has $m' \geq 3m/4$ hyperedges. It thus suffices to show that the resulting instance has value $\leq 1/3$, as the deleted hyperedges contribute at most $\leq \frac{1}{4}$ to the value of the original instance, and then the remaining hyperedges contribute $\leq \frac{1}{3} \cdot \frac{3}{4} = \frac{1}{4}$. Note that deleting edges preserves regularity, and we have $\abs{\cH'_u} \leq \abs{\cH_u} = \frac{m}{p} \leq \frac{4m'}{3p} \leq \frac{2m'}{p}$. We thus redefine $\cH$ to be the hypergraph after deleting edges, and similarly for the polynomial $\psi$; it now remains to show that $\val(\psi) \leq \frac{1}{3}$.}

Our proof follows the exact same outline as in~\cref{sec:partition-refute} for finding an efficient refutation algorithm for the polynomial $\psi$. One important difference is that in this section, we will use the argument to argue an upper bound on $\val(\psi)$; we do not care about finding an efficient certificate for a bound on $\val(\psi)$ here. 

The key observation that we use in this proof is that there is exactly one step of the proof in~\cref{sec:partition-refute} that uses the randomness of the coefficients $b_{u,C}$'s -- namely, \cref{lem:Gspecbound}. Our proof in this section is exactly the same with the key innovation being an analog of \cref{lem:Gspecbound} that works for \emph{arbitrary} $b_{u,C}$'s as long as $\cH$ has no $O(\ell \log n)$-length even cover. Indeed, as the hypergraph $\cH$ satisfies the assumptions of \cref{thm:bipartite-poly-refute}, with this observation
we immediately see that in order to finish the proof, it suffices to show that the spectral norm bounds in \cref{lem:Gspecbound} still hold. In what follows, we use the exact same notation and conventions as in \cref{sec:partition-refute}.

Let $f$ be the polynomial obtained in~\cref{lem:cauchyschwarztrick} to the polynomial $\psi$. 
Let $A$ be the Kikuchi matrix (\cref{def:kikuchi-matrix}) corresponding to the polynomial $f$. 
Using~\cref{lem:kikuchi-quad-form}, we obtain that:
\[
\val(\psi)^2 \leq \frac{1}{12} + \val(f) \leq \frac{1}{12}+\frac{p}{m^2 D} \boolnorm{A}\mcom
\]
where we use that $12p \leq m$.\footnote{We note that this is the only other part where we deviate at all from the proof in \cref{sec:partition-refute}; here, we now have $12p \leq m$ instead of $16 p \leq m$ because we removed $4p$ edges; this is not important.} Recall also that $D := {k - 1 \choose \frac{k-1}{2}}^2 {2n - 2(k-1) \choose \ell - (k-1)}$ if $k$ is odd and $2 {k - 1 \choose \frac{k}{2}}{k - 1 \choose \frac{k-2}{2}} {2n - 2(k-1) \choose \ell - (k-1)}$ if $k$ is even. 

Next, let $\cB$ be the bad rows in $A$. Using \cref{lem:numbadrows}, we know that for $\Delta = {c'}^{k-1}\frac{1}{\eps^4}\ln^{2(k-1)}(\frac{32 p N}{\eps^2 D})$ (where $c'$ is an absolute constant and $\eps = 1/4$), $\frac{\abs{\cB}}{N} \leq \epsilon^2 D/16N$. Let $G$ be the matrix defined by zeroing out rows/columns in $\cB$ from $A$, as in the proof of \cref{lem:boolnormgood} in \cref{sec:boolnormgood}. Let $\cF_0 \cup \cF_1 \cup \ldots \cF_t$ for $t \leq 2 \log_2 m$ be the partition of non-bad rows of $A$ and let $G^{(i,j)}$ be the matrices obtained by zeroing out rows and columns not in $\cF_i$ and $\cF_j$ from $G$ respectively as in \cref{def:row-bucketing}. Let $G^{(i,j)}_u$ be defined similarly by zeroing out rows and columns not in $\cF_i$ and $\cF_j$ respectively from $G_u$. Then, following the steps in the proof of \cref{sec:boolnormgood}, all that remains to be shown is the conclusion of \cref{lem:Gspecbound} holds. In \cref{sec:boolnormgood}, we proved \cref{lem:Gspecbound} by crucially exploiting the randomness of $b_{u,C}$'s.
Here, the $b_{u,C}$'s are allowed to be \emph{arbitrary}. We nonetheless show that the same conclusion holds if we additionally assume that $\cH$ has no small even cover. Formally, we prove the following lemma.
\begin{lemma}[Spectral Norm of $G^{(i,j)}$'s when $\cH$ has no small even cover]
\label{lem:Gspecbound-no-even-cover}
Suppose that the $(1/4,\ell)$-regular $p$-bipartite simple $k$-uniform hypergraph $\cH$ associated to the polynomial $\psi$ has no even cover of length $\leq c_0 \ell \log_2 n$ for some large enough constant $c_0$. Then, for each $i,j \in \{0, \dots, t\}$, we have:
\begin{equation*}
\norm{G^{(i,j)}}_2 \leq O(1)  \cdot 2^{0.5 \max(i,j)} \sqrt{d \log N} + O(1) \Delta \log N\mper
\end{equation*}
\end{lemma}
\cref{lem:Gspecbound-no-even-cover} finishes the proof of \cref{lem:partitioned-poly-refute-odd-non-det}. Indeed, via the identical calculation in \cref{sec:partition-refute}, it implies that $\frac{p}{m^2 D} \boolnorm{A} \leq \eps^2 = \frac{1}{16}$, and thus $\val(\phi) \leq \frac{1}{12} + \frac{1}{16} \leq \frac{1}{3}$, so we are done.

It thus remains to prove \cref{lem:Gspecbound-no-even-cover}.
\begin{proof}[Proof of \cref{lem:Gspecbound-no-even-cover}]
We will follow the proof of \cref{lem:Gspecbound} that uses the trace method (\cref{sec:Gspecboundtrace}).
Fix a pair $(i,j$). For ease of notation, let us write $Z= G^{(i,j)}$ and $Z_u$ for $G_u^{(i,j)}$ in the following. We know that $\norm{Z}_2 \leq \tr( (Z Z^{\top})^r)^{1/{2r}}$ for every $r \in \N$. We prove \cref{lem:Gspecbound-no-even-cover} by upper bounding $\tr( (Z Z^{\top})^r)$ for some $r = O(\ell \log_2 n)$.

We remind the reader that the trace moment method is classically used in analyzing the spectral norms of \emph{random} matrices. In that setting, one bounds the \emph{expectation} of $\tr((ZZ^{\top})^r)$ which is analyzed by understanding the terms on the expansion on the right hand side above that contribute a non-zero expectation often by utilizing inherent independence in the random variables appearing as entries of the matrix $Z$. In contrast, there is \emph{no randomness} in the matrix $Z$, and so we are not bounding the expectation. Instead, we will analyze the ``contributing'' terms on the right hand side by appealing to a crucial (and hitherto unobserved) property of the contributing walks in the Kikuchi matrix. We stress that the analysis appearing below does (as in fact any such analysis must!) strongly rely on the combinatorial structure of the support of the non-zero entries in our Kikuchi matrix $A$ and cannot work for arbitrary matrices. 

In fact, our key observation is to show that if $\cH$ has no short even covers, then our upper bound on the \emph{expectation} of $\tr((ZZ^{\top})^r)$ in the semirandom setting (\cref{prop:exptracebound}) still holds for $\tr((ZZ^{\top})^r)$, i.e., when the $b_{u,C}$'s are \emph{arbitrary}. Formally, we show the following.
\begin{proposition}
\label{prop:tracebound}
Suppose that the $(1/4,\ell)$-regular $p$-bipartite simple $k$-uniform hypergraph $\cH$ associated to the polynomial $\psi$ has no even cover of length $\leq 4 c_0 \ell \log_2 n$ for some large enough constant $c_0$. Then, for $r \leq c_0 \ell \log_2n$, it holds that $\tr((ZZ^{\top})^r) \leq \sum_{S \in \cF_i} \text{\# even walk sequences $(u_1, C_1, C'_1), \dots, (u_{2r}, C_{2r}, C'_{2r})$ for $S$}$.
\end{proposition}
We note (at the cost of repetition) that \cref{prop:tracebound} holds \emph{regardless} of the $b_{u,C}$'s and is a consequence of the combinatorial structure of the support of Kikuchi matrices. 

We now finish the proof of \cref{lem:Gspecbound-no-even-cover} assuming \cref{prop:tracebound}. This is immediate given the calculations in \cref{sec:Gspecboundtrace}. 
By \cref{lem:sequencecounting}, we know that for each $S \in \cF_i$, the  number of such sequences is at most $(4 r)^r  (2^{\max(i,j)} d + r \Delta^2)^{r}$. Hence, 
 \begin{equation*}
\norm{Z}_2^{2r} \leq \tr((ZZ^{\top})^r) \leq N (4 r)^r  (2^{\max(i,j)} d + r \Delta^2)^{r} \mper
\end{equation*}
Setting $r = c_0 \ell \log_2 n$ for $c_0$ a sufficiently large constant, the above implies that
\begin{equation*}
\norm{Z}_2 \leq O(1) 2^{0.5\max(i,j)} \sqrt{d \log_2 N} + O(1) \Delta \log_2 N \mcom
\end{equation*}
assuming that $\cH$ has no even cover of length $\leq 4r = 4 c_0 \ell \log_2 n$. This finishes the proof, up to \cref{prop:tracebound}.
\end{proof}

\begin{proof}[Proof of \cref{prop:tracebound}]
We compute:
\begin{equation} \label{eq:trace-expansion-odd}
\tr((ZZ^{\top})^r) = \sum_{u_1, S_1, u_2, S_2, \ldots, u_{2r}, S_{2r}} \prod_{h=1}^r Z_{u_{2h-1}} (S_{2h-1},S_{2h}) Z_{u_{2h}} (S_{2h+1}, S_{2h})\mper
\end{equation}
where we let $u_{2r+1} := u_1$ and $S_{2r+1} := S_1$.

Observe that each term in \eqref{eq:trace-expansion-odd} can contribute a value at most $1$ since all $b_{u,C}$'s are $\{\pm 1\}$ and $\cH$ is simple. Thus, the RHS of \eqref{eq:trace-expansion-odd} is upper-bounded by the number of non-zero ``walk'' terms, i.e., the number of terms in the sum in \eqref{eq:trace-expansion-odd}. 

The central observation is the following lemma that observes a combinatorial property of non-zero terms on the RHS in \eqref{eq:trace-expansion-odd}.

\begin{claim}[Non-zero terms are even multicovers] \label{claim:non-zero-multicovers}
If the walk term corresponding to $(u_1, S_1, u_2, S_2, \ldots, u_{2r}, S_{2r})$ is non-zero, then for every $h \in [2r]$, there exist $C_h \ne C_h' \in \cH_{u_h}$ such that $S_{h+1} = S_h \oplus C_h^{(1)} \oplus {C'_h}^{(2)}$. Moreover, $\bigoplus_{h \leq 2r} (u_h,C_h) \oplus (u_h,C'_h) =  \emptyset$, i.e., $\{(u_h,C_h), (u_h,C'_h)\}_{h \leq 2r}$ is an even multicover in $\cH$.
\end{claim}
\begin{proof}
By definition of the \kikuchi matrix, the walk term equals
\begin{multline}
\prod_{h \leq r} Z_{u_{2h-1}} (S_{2h-1}, S_{2h}) Z_{u_{2h}} (S_{2h+1}, S_{2h}) \\= \prod_{h \leq r} b_{u_{2h-1}, C_{2h-1}} b_{u_{2h-1}, C_{2h-1}'} b_{u_{2h}, C_{2h}} b_{u_{2h}, C_{2h}'} \1( S_{2h-1} \overset{C_{2h-1}^{(1)}, {C'_{2h-1}}^{(2)}}{\longleftrightarrow} S_{2h}) \1( S_{2h} \overset{C_{2h}^{(1)}, {C'_{2h}}^{(2)}}{\longleftrightarrow} S_{2h+1})  \mcom
\end{multline}
where for each $h$, $C_{2h-1}, C_{2h-1}' \in \cH_{u_{2h-1}}$ and $C_{2h}, C_{2h}' \in \cH_{u_{2h}}$. 

Clearly, if the term corresponding to $(u_1, S_1, u_2, S_2, \ldots, u_{2r}, S_{2r})$ is non-zero then $\1( S_{2h-1} \overset{C_{2h-1}^{(1)}, {C'_{2h-1}}^{(2)}}{\leftrightarrow} S_{2h}) =1$ for every $h \leq r$. 
Expanding the definition, this implies that $S_{2h} = S_{2h - 1} \oplus C_{2h-1}^{(1)} \oplus {C'_{2h-1}}^{(2)}$. Similarly, we also have that $S_{2h + 1} = S_{2h}  \oplus C_{2h}^{(1)} \oplus {C'_{2h}}^{(2)}$.

To show the ``moreover'', we observe that by adding up all the aforementioned two equations, we obtain:
\begin{equation*}
\bigoplus_{h = 2}^{2r + 1} S_h = \bigoplus_{h = 1}^{2r} S_h \oplus \bigoplus_{h = 1}^{2r} C_h^{(1)} \oplus {C'_h}^{(2)} \mper
\end{equation*}
As $S_{2r+1} := S_1$, canceling the $S_h$'s on both sides yields $\bigoplus_{h \leq 2r} C^{(1)}_h \oplus {C'_h}^{(2)} = \emptyset$. This then trivially implies that $\bigoplus_{h \leq 2r} C_h = \bigoplus_{h \leq 2r} {C'_h} = \emptyset$, and hence $\bigoplus_{h \leq 2r} (u_h,C_h) \oplus (u_h,C'_h) = \emptyset$, as $(u_h,C_h) \oplus (u_h,C'_h) = C_h \oplus C'_h$.
\end{proof}

Observe that the even multicover $\{(u_h,C_h), (u_h,C'_h)\}_{h \leq 2r}$ in \cref{claim:non-zero-multicovers} need not be an even cover as the $(u_h,C_h)$'s need not be distinct. Indeed, the main punch of what follows is that when there are no small even covers in $\cH$, then the $(u_h, C_h)$'s must occur in pairs, i.e., each $(u_h, C_h)$ appears an even number of times in the two multicovers obtained in \cref{claim:non-zero-multicovers}. 

\begin{claim}[No short even cover implies short multicovers are unions of  pairs] \label{claim:multi-covers-in-pairs}
Suppose $\cH = \{\cH_u\}_{u \in [p]}$ has no even cover of length $\leq 4r$. Then, if the walk term in \eqref{eq:trace-expansion-odd} corresponding to $\{u_h, S_h, C_h, C_h'\}_{h \leq 2r}$ is non-zero, then each $(u,C) \in \cup_{u \in [p]} \cH_u$ occurs an even number of times in the multiset $\{(u_h, C_h), (u_h, C'_h)\}_{h \leq 2r}$. In particular, $\{(u_h, C_h, C'_h)\}_{h \leq 2r}$ is an \emph{even walk sequence for $S_1$}, as defined in \cref{def:walksequence}.
\end{claim}
\begin{proof}
From \cref{claim:non-zero-multicovers}, $\bigoplus_{h = 1}^{2r} (u_h, C_h) \oplus (u_h, C'_h) = \emptyset$.
Start from the multiset $\{(u_h, C_h), (u_h, C'_h)\}_{h \leq 2r}$, and remove pairs greedily until this is no longer possible. Observe that the symmetric difference of the resulting set must also be empty since we removed sets in equal pairs. If at the end of this process, we are left with a non-zero number of hyperedges, i.e., we assume that the conclusion does not hold, then we have at most $4r$ distinct hyperedges whose symmetric difference is empty. Thus, the remaining set must be an even cover of length $\leq 4r$ in $\cH$, which is a contradiction.
\end{proof}

Combining \cref{claim:non-zero-multicovers,claim:multi-covers-in-pairs}, we thus see that the RHS of \eqref{eq:trace-expansion-odd} is upper bounded by $\sum_{S \in \cF_i} \text{\# even walk sequences $(u_1, C_1, C'_1), \dots, (u_{2r}, C_{2r}, C'_{2r})$ for $S$}$, which finishes the proof of \cref{prop:tracebound}.
\end{proof}

\section{Polynomial Size Refutation Witnesses Below the Spectral Threshold} 
\label{sec:fko2}
In this section, we use our smoothed refutation algorithm along with our proof of Feige's conjecture to show the existence of polynomial size refutation witnesses below the spectral threshold for smoothed instances of Boolean CSPs. Modulo the use of our key new ingredients -- \cref{thm:mainpolyrefute,thm:feige-conjecture} -- the rest of the proof plan largely follows the influential work of Feige, Kim and Ofek~\cite{FKO06} who proved that \emph{fully random} instances of 3-SAT admit polynomial size refutation witnesses whenever they have at least $\tilde{O}(n^{1.4})$ constraints. Our new ingredients allow us to \begin{inparaenum}[(1)] \item show a similar result for not just fully random instances, but also semirandom and smoothed ones, and \item provide an arguably simpler refutation witness even for the fully random instances of $3$-SAT studied by~\cite{FKO06}\end{inparaenum}.

Let us first formalize the idea of a \emph{refutation witness}, or equivalently, a nondeterministic refutation algorithm.
\begin{definition}[Nondeterministic refutation]
Fix $k \in \N$, and let $P:\on^k \rightarrow \zo$ be a predicate. We say that a nondeterministic algorithm $V$ is an \emph{nondeterministic efficient weak refutation algorithm} if $V$ takes as input a CSP instance $\psi$ with predicate $P$ in $n$ variables and $m$ clauses and in $\poly(n,m)$-nondeterministic time outputs either ``unsatisfiable'' or ``don't know'', such that for every $\psi$, if $V(\psi)$ outputs ``unsatisfiable'' then $\psi$ is unsatisfiable. If $V(\psi)$ outputs ``unsatisfiable'', then we say that $V$ weakly refutes $\psi$. The string $\pi \in \zo^{\poly(n,m)}$ of nondeterministic guesses of $V$ is called the weak refutation witness.
\end{definition}

We will sketch a proof of the following theorem. We only provide a proof sketch, as the proof merely combines the ideas of \cite{FKO06} with our theorems, \cref{thm:mainpolyrefute,thm:feige-conjecture}.
\begin{theorem}
\label{thm:ourfko}
Let $k \geq 3$, and let $P \colon \on^k \to \zo$ be a non-trivial predicate. Then there is a nondeterministic efficient weak refutation algorithm $V$ with the following properties. Let $\psi$ be an instance of a CSP with predicate $P$ with $n$ variables and $m$ clauses, specified by a collection of $m$ $k$-tuples $\cH$ and literal patterns $\xi$. Then:
\begin{enumerate}[(1)]
\item If $\psi$ is a uniformly random instance with $m \geq \tilde{O}(1) \cdot n^{\frac{k}{2} - \frac{k-2}{2(k+2)}}$ clauses, then $V$ weakly refutes $\psi$ with probability at least $1 - 1/\poly(n)$.
\item If $\psi$ is a semirandom instance with $m \geq \tilde{O}(1) \cdot n^{\frac{k}{2} - \frac{k-2}{2(k+8)}}$ clauses, then $V$ weakly refutes $\psi$ with probability at least $1 - 1/\poly(n)$.
\item If $\psi$ is a smoothed instance obtained using smoothing parameters $\vec{p} = \{p_{C,i}\}_{C \in \cH, i \in [k]}$ with $m \geq \tilde{O}(1) \cdot n^{\frac{k}{2} - \frac{k-2}{2(k+8)}}/q(\vec{p})$ clauses, where $q(\vec{p}) := \frac{1}{m} \sum_{C \in \cH} \prod_{i \in C} p_{C,i}$, then $V$ weakly refutes $\psi$ with probability at least $1 - 1/\poly(n)$. 
\end{enumerate}
Finally, if $k = 3$, the threshold of $m$ for the semirandom/smoothed case can be improved to $\tilde{O}(n^{1.4})$ and $\tilde{O}(n^{1.4})/q(\vec{p})$, respectively, matching the random case.
\end{theorem}
We will first begin by focusing on the case of $k$-XOR. As in the case of \cref{sec:smoothed}, refuting arbitrary predicates $P$ will reduce to refuting XOR.

In~\cite{FKO06}, FKO observed that the following type of refutation witnesses, which we shall call \emph{ideal FKO witnesses}, allow for a non-trivial\footnote{Note that by running Gaussian elimination, one can  decide if a $k$-XOR instance is unsatisfiable in polynomial time. This is a \emph{trivial} weak refutation.} weak refutation of instances of $k$-XOR whenever the $b_C$'s are chosen uniformly and independently at random. Informally speaking, ideal FKO witnesses are simply a disjoint collection of even covers in $\cH$. 

\begin{definition}[Ideal FKO witnesses]
Let $\cH$ be $k$-uniform hypergraph on $[n]$. We say that a collection of even covers $E_1, E_2, \ldots, E_r \subseteq \cH$ is an \emph{ideal FKO witness of length} $h$ if each $E_i \cap E_j = \emptyset$ for every $i \neq j$ and $|E_i| \leq h$ for every $i$, where $\abs{E_i}$ denotes the length of the even cover $E_i$. The size of the witness is $s = \sum_{i = 1}^r \abs{E_i} \leq h r$.
\end{definition}
Ideal FKO witnesses yield non-trivial weak refutation witnesses for semi-random instances of $k$-XOR. 
\begin{lemma}[Ideal FKO witnesses yield refutation witnesses for XOR]
Let $\psi = (\cH, b)$ be an instance of $k$-XOR on $n$ variables. Suppose $E_1, E_2, \ldots, E_r \subseteq \cH$ is an ideal FKO witness in $\cH$. Suppose further that each $b_C$ is a uniformly random and independent bit in $\pm 1$. Then, with probability at least $1-\exp(\Omega(r))$ over the draw of $b= \{b_C\}_{C \in \cH}$, $\val(\psi) \leq 1-\frac{r}{3m}$. \label{lem:ideal-fko-witness-implies-refutation}
\end{lemma}

\begin{proof}

For each $i$, consider $Z_i =\prod_{C \in E_i} b_C$. Then, notice that $Z_1, Z_2, \ldots, Z_r$ are independent random variables, each uniformly drawn from $\on$. Thus, by a Chernoff bound, with probability at least $1-\exp(\Omega(r))$ there must exist at least $r/3$ $E_i$'s such that $Z_i = -1$. Consider any such $E_i$ where this holds.

Suppose some $x \in \on^n$ satisfies all the constraints in $\psi$ corresponding to $k$-tuples $C \in E_i$. Then, $\prod_{C \in E_i} b_C = \prod_{C \in E_i} \prod_{j \leq k} x_{C_j}$. Since $E_i$ is an even cover, every variable occurs an even number of times in the $C$'s in $E_i$. Since even powers of any $x_j$ evaluate to $1$, the RHS above must evaluate to $1$. Since we know that $\prod_{C\in E_i} b_C = -1$, this implies that such an $x$ cannot exist: every $x$ must violate at least one constraint in each $E_i$ if $\prod_{C \in E_i} b_C = -1$. Since $E_i$'s are disjoint, this implies that every $x$ violates at least $r/3$ constraints in $\psi$. The bound on $\val(\psi)$ now follows.
\end{proof}

The key question is whether Ideal FKO witnesses exist in the $k$-uniform hypergraph specifying the $k$-XOR instance. In~\cite{FKO06}, the authors study the question of finding such refutation witnesses in \emph{random} sufficiently dense hypergraphs. They comment that, while they expect Ideal FKO witnesses to exist in the regime they are working in, proving that they exist appears hard. They instead show that a related form of witnesses (these are ``almost disjoint'' even covers instead of perfectly disjoint) exist by means of a sophisticated second moment method argument. 

Here, we show that Ideal FKO witnesses do indeed exist -- not only in random dense hypergraphs but in \emph{arbitrary} hypergraphs with the same density. Indeed, this follows almost immediately from~\cref{thm:feige-conjecture}.
\begin{lemma} \label{lem:one-even-cover-is-enough}
Fix $k \in \N$ and $\ell = \ell(n)$. Let $\cH$ be any $k$-uniform hypergraph with $m \geq 2m_0$ hyperedges, where $m_0 = \Gamma^k \cdot n \Paren{\frac{n}{\ell}}^{\frac{k}{2}-1} \log^{4k+1} n$ is the threshold appearing in~\cref{thm:feige-conjecture}. Then, $\cH$ contains a collection of $m_0/h(n)$ hyperedge-disjoint even covers each of length at most $h(n) = O(\ell \log n)$. 
\end{lemma}

\begin{proof}
The idea is simple. Let $m_0$ be the number of constraints required in \cref{thm:feige-conjecture}. Choose $m = 2m_0$. Then, by an application of \cref{thm:feige-conjecture}, there is an even cover in $\cH$, say, $E_1$ of size $|E_1| \leq h(n) = O(\ell \log n)$. Let $\cH_0 = \cH$. We now repeat the following process for $i = 1,2,\ldots,r$: apply \cref{thm:feige-conjecture} to $\cH_i := \cH_{i-1} \setminus E_i$ to find an even cover $E_{i+1} \subseteq \cH_i$ of size $\leq h(n) = O(\ell \log n)$. Notice that the conditions of \cref{thm:feige-conjecture} are met so long as $|\cH_i| \geq m-h(n)r \geq m/2$, i.e., if $r \leq 0.5 m/h(n)$. Further, each of the even covers $E_1, E_2, \ldots, E_r$ are pairwise disjoint by construction. This completes the proof. 
\end{proof}

By combining the above observation with semirandom refutation algorithms, one can show that Ideal FKO witnesses yield weak refutation witnesses for all $k$-CSPs at densities polynomially below $n^{k/2}$. This is one of the key insights of FKO~\cite{FKO06} -- to use the non-trivial weak refutation offered by (their variant of) ideal FKO witnesses in order to show the existence of polynomial size weak-refutation witnesses for random $3$-SAT with $m = \tilde{\Omega}(n^{1.4})$ constraints: namely, in a regime of $m$ where known spectral algorithms, and more generally those based on the polynomial-time canonical sum-of-squares relaxation, provably fail. 
\cref{thm:feige-conjecture} (and its consequence \cref{lem:one-even-cover-is-enough}) implies that the same result holds for \emph{arbitrary} constraint hypergraphs, up to additional $\polylog(n)$ factors in the number of constraints.

\begin{lemma}[Ideal FKO witnesses yield weak refutation witnesses for 3-SAT] \label{lem:weak-refuting-3-sat-fko}
Let $\psi = (\cH,\Xi)$ be an instance of $3$-SAT described by a $3$-uniform hypergraph $\cH$ on $[n]$ with $m \geq \tilde{O}(n^{1.4})$ arbitrary constraints and uniformly randomly generated literal patterns. Then, with probability at least $1 - 1/\poly(n)$ over the draw of the literal patterns in the instance, there is a polynomial-size refutation witness that certifies $\val(\psi) < 1$.
\end{lemma}

\begin{proof}[Proof Sketch]
Let $P:\on^3 \rightarrow \zo$ be the $3$-SAT predicate. 
Then, $P(z) = \frac{7}{8} + \frac{1}{8} (z_1 +z_2 +z_3) -\frac{1}{8} \Paren{z_1 z_2 + z_2 z_3 + z_1 z_3- z_1 z_2 z_3}$. 
We write 
\begin{align*}
\psi(x) &= \frac{1}{|\cH|}\sum_{C \in \cH} P(x_{C_1} \Xi_{C,1}, x_{C_2} \Xi_{C,2}, x_{C_3} \Xi_{C,3})\\
&= \frac{7}{8} + \frac{1}{8|\cH|} \sum_{C \in \cH} (\Xi_{C,1} x_{C_1}+ \Xi_{C,2} x_{C_2} + \Xi_{C,3} x_{C_3}- \Xi_{C,1} x_{C_1}\Xi_{C,2} x_{C_2} - \Xi_{C,2} x_{C_2}\Xi_{C,3} x_{C_3}\\
&-\Xi_{C,1} x_{C_1}\Xi_{C,3} x_{C_3} + \Xi_{C,1} \Xi_{C,2} \Xi_{C,3} x_{C_1} x_{C_2} x_{C_3})\mper
\end{align*}  where the $\Xi_{C,i}$'s are the literal negation patterns in $\on$. Note that $\psi(x)$ computes the fraction of constraints satisfied by the assignment $x \in \on^n$. We refute each of the $7$ different XOR instances produced by taking each of the $7$ non-constant terms in the expansion of $P$ as a multilinear polynomial above separately. 

Our refutation witness helps us efficiently refute each of the instances corresponding to the $7$ terms in the expansion above. Specifically, by collecting coefficients together, each the first three terms each produce a linear polynomial of the form $\sum_i B_i x_i$. The next three terms each produce a homogenous quadratic polynomial of the form $\frac{1}{|\cH|} \sum_{C \in \cH} x_{C_i}x_{C_j}$, and finally the last term is a cubic polynomial of the form $\frac{1}{|\cH|} \sum_{C \in \cH} x_{C_1} x_{C_2} x_{C_3}$. Our refutation witness for each linear polynomial is simply $\Norm{B}_1$, where $B = (B_1, \dots, B_n)$, noting that this is exactly the maximum of the first kind of terms as $x$ varies over the hypercube. For the quadratic case, our refutation witness is the value of SDP relaxation for the $\infty \to 1$ norm that gives a $<2$ factor approximation to maximum of bilinear forms over the hypercube. For the homogeneous degree $3$ term, our witness is an ideal FKO witness guaranteed by \cref{lem:one-even-cover-is-enough}.

By Chernoff and union bound argument (applied to every assignment in $\on^n$), $\Norm{B}_1$ for any linear term above is at most $O(\sqrt{n/m})$.

By Chernoff and union bound argument, the $\infty \rightarrow 1$-norm of the matrix defining the 2-XOR constraints is at most $O(\sqrt{n/m})$. By Grothendieck's inequality (\cref{fact:grothendieck}), we can certify this value efficiently (with an additional loss of at most a factor of $<2$) using an SDP.

Thus, we can certify an upper bound of $O(\sqrt{n/m})$ on all but homogeneous degree $3$ polynomial produced in the Fourier expansion above. When $m \geq \tilde{\Omega}(n) n^{0.5(1-\delta)}$, i.e., $\ell = n^{\delta}$, by \cref{lem:one-even-cover-is-enough}, $\cH$ has a collection of $\frac{m}{\tilde{O}(n^{\delta})}$ pairwise disjoint even covers of length at most $\tilde{O}(n^{\delta})$. By Chernoff bounds, at least $\frac{1}{3}$ of these even covers must violated and thus, we have obtained a certificate for an upper-bound of $1-\frac{1}{\tilde{O}(n^\delta)}$ on the value of the final term. 

Putting these upper bounds together gives an upper bound of $\frac{7}{8} + \frac{1}{8} O(\sqrt{\frac{n}{m}}) + \frac{1}{8}(1- \frac{1}{\tilde{O}(n^{\delta})})$ on the value of the 3-SAT instance. For $\delta = 0.2$, we observe that $\sqrt{\frac{n}{m}} = \tilde{O}(-n^{0.25+\delta/4}) \ll \frac{1}{\tilde{O}(n^{\delta})}$. Thus, for $m \geq \tilde{O}(n^{1.4})$, with probability at least $1- 1/\poly(n)$, we obtain a refutation for the input 3-SAT instance.
\end{proof}

\cref{lem:weak-refuting-3-sat-fko} generalizes to all $k$-CSPs with predicate $P$, provided that $P$ is non-trivial, i.e., $P$ is not identically $1$. 
We only need the following basic fact (and the rest of the proof remains the same as above), as well as known results for spectral refutation of \emph{random} $k-1$ and smaller-arity XOR instances.
\begin{lemma}[Highest Fourier Coefficient of Boolean Functions] \label{lem:last-fourier-coeff-is-small}
Let $P:\on^k \rightarrow \zo$. Let $\sum_{S \subseteq [k]} \hat{P}(S) x_S$ be the Fourier polynomial representation of $P$. Then, $\hat{P}(\emptyset) + |\hat{P}([k])| \leq 1$. 
\end{lemma}
\begin{proof}
For each $b \in \on$, consider the distribution that is uniform on all $x$ such that $\prod_i x_i = b$. Then, the expectation of $P$ on this distribution is exactly $\hat{P}(\emptyset) + b\hat{P}([k])$. On the other hand, since $P$ takes values in $\zo$, this expectation cannot exceed $1$. Thus, $1 \geq \hat{P}(\emptyset) + b\hat{P}([k])$ for both values of $b$ and in particular, $1 \geq \hat{P}(\emptyset) + |\hat{P}([k])|$ as desired. 
\end{proof}

We now sketch a proof of the generalization of \cref{lem:weak-refuting-3-sat-fko} to all \emph{fully random} CSPs. This is captured by Item (1) in \cref{thm:ourfko}. We will assume that the Fourier coefficient $\hat{P}([k])$ is nonzero, as otherwise by \cref{thm:smoothed-refutation}, we have enough constraints to give a polynomial time \emph{deterministic} refutation.\footnote{This is because there cannot be a $(k-1)$-uniform distribution $\mu$ supported on $P^{-1}(1)$, as otherwise we would have $1 = \E_{x \sim \mu}[P(x)] = \hat{P}(\emptyset) < 1$, where we have $\hat{P}(\emptyset) < 1$ as $P$ is nontrivial. And then we observe that the CSP instance has at least  $\tilde{O}(n^{\frac{k}{2} - \frac{k-2}{2(k+2)}})$ constraints, which is at least $\tilde{O}(n^{\frac{k-1}{2}})$.}
\begin{lemma}[Polynomial Size Refutation Witnesses for all \emph{random} $k$-CSPs] \label{lem:full-random-poly-size-witness}
Let $P:\on^k \rightarrow \zo$ be an arbitrary $k$-ary Boolean predicate for $k \geq 3$. Let $\psi$ be a CSP instance with predicate $P$ specified by $\cH$-- a collection of uniformly at random and independently generated $m \geq m_0= \tilde{O}(1) \cdot n^{\frac{k}{2} - \frac{k-2}{2(k+2)}}$ $k$-tuples and uniformly random and independently generated literal patterns $\{\Xi(C,i)\}_{C \in \cH, i \in [k]}$. Then, with probability at least $1 - 1/\poly(n)$ over the draw of $\cH$ and $\Xi(C,i)$'s, there exists a polynomial size refutation witness for $\psi$. 
\end{lemma}
\begin{proof}
Observe that the instance $\psi$ has $m = \tilde{O}(1) \cdot \left(\frac{n}{\ell}\right)^{k/2} \ell$ constraints for $\ell \leq \tilde{O}( n^{\frac{1}{k+2}})$.
We now use Fourier analysis to decompose $\psi(x) := \frac{1}{|\cH|}\sum_{C \in \cH} P(x_{C_1} \Xi_{C,1}, \dots, x_{C_k} \Xi_{C,k})$ into $2^{k}$ polynomials, each of degree $t \leq k$. 
We use the same certificate as in \cref{lem:weak-refuting-3-sat-fko} for the linear polynomials appearing in this decomposition. For quadratic and higher degree $(\leq k-1)$ terms, we now use spectral refutation from prior results on refuting fully random CSPs, such as Theorem 1 in~\cite{AOW15}. Each degree $t$ polynomial (with $t \leq k-1$) that appears requires at least $\tilde{O}(n^{t/2}/\eps^2)$ constraints to certify an upper bound of $\eps$ on its value; we can thus certify an upper bound of $\eps = \sqrt{\frac{n^{(k-1)/2}}{m}}$ on each polynomial. Note that by choice of $m$, we have $\eps \leq 1$. 

Finally, to refute the final and highest degree polynomial obtained by taking the $[k]$-indexed Fourier coefficient of $P$, we use the the Ideal FKO witness from \cref{lem:ideal-fko-witness-implies-refutation}. Then, as in the argument for $3$-SAT above, we arrive at a certificate that (with probability at least $1 - 1/\poly(n)$) certifies an upper-bound of $\hat{P}(\emptyset) + \tilde{O}( \sqrt{\frac{n^{(k-1)/2}}{m}}) + |\hat{P}([k])| \cdot (1- \frac{\tilde{O}(1)}{\ell \log n})$ on the value of $\psi$, using \cref{lem:one-even-cover-is-enough}. The size of the witness  is $s(n) \leq m_0 = \poly(n)$, as the degree $< k$ terms used deterministic refutations. Using \cref{lem:last-fourier-coeff-is-small}, we thus certify an upper bound of $1 + \tilde{O}(\sqrt{\frac{n^{(k-1)/2}}{m}}) - \frac{\tilde{O}(1)}{\ell \log n} = 1 - o(1)$ on $\psi(x)$, which finishes the proof. Note that this is indeed $1 - o(1)$ as $ \tilde{O}(1) \sqrt{\frac{n^{(k-1)/2}}{m}} = \tilde{O}(1) \cdot \ell^{\frac{k}{4} - \frac{1}{2}}/n^{\frac{1}{4}} \ll \tilde{O}(1/\ell)$, since $\ell \leq \tilde{O}(1)  n^{\frac{1}{k+2}}$.
\end{proof}

By switching the CSP refutation algorithms in~\cite{AOW15} with the semirandom refutation algorithm from \cref{thm:mainpolyrefute} in this work, we arrive at Item (2) of \cref{thm:ourfko}, a version of the above result that shows the existence of polynomial size refutation witnesses below the $n^{k/2}$-threshold for \emph{semirandom} instances. As the proof is very similar, we omit the details of the proof; the final bound is stated in Item (2).
 Note that the precise value of $m$ at which this refutation succeeds is strictly larger (though still polynomially smaller than $n^{k/2}$) than the one in \cref{lem:full-random-poly-size-witness}, i.e., Item (1). The difference comes from the fact that the dependence on $\epsilon$ (the strength of the refutation) in our semirandom refutation algorithms grows as $1/\epsilon^5$ instead of the $1/\epsilon^2$ dependence of algorithms for fully random instances; we thus have to take $\eps = \left( n^{(k-1)/2}/m\right)^{1/5}$ instead of $\left( n^{(k-1)/2}/m\right)^{1/2}$, which in turn makes $\ell = n^{1/(k+8)}$ and then $m \geq \tilde{O}(1) n^{\frac{k}{2} - \frac{k-2}{2(k+8)}}$. Our belief is that the $1/\eps^5$ dependence is sub-optimal in the semirandom setting but inherent to our current proof techniques.

We note that for large $k$, the density required for the polynomial size refutation witnesses to exist in both Item (1) and Item (2) is $\sim n^{\frac{k}{2} - 0.5+o_k(1)}$, effectively giving a $\sqrt{n}$ factor ``win'' over the threshold at which spectral (and sum-of-squares based methods more generally) succeed.

In the specific case of $k = 3$, we can improve the bound in the semirandom case to match the $\tilde{O}(n^{1.4})$ achieved in the random case. This is because the instances appearing in the decomposition are all semirandom $2$-XOR instances, and we can refute these instances with the correct $1/\eps^2$ dependence: see Proposition 5.2.2 and Theorem 5.2.3 in \cite{Witmer17}, combined with the fact that the value of a semirandom $2$-XOR instance is at most $\frac{1}{2} + \eps$ when $m \gg n/\eps^2$.

Finally, to handle Item (3), we observe that by Chernoff bound, if $m \geq O(1) m_0/q(\vec{p})$, where $m_0 = \tilde{O}(1) \cdot n^{\frac{k}{2} - \frac{k-2}{2(k+8)}}$, then with high probability there are at least $m_0$ clauses in $\psi$ where all literals in the clause are re-randomized by the smoothing process. Call this subinstance $\psi'$. As $\psi'$ is semirandom, by Item (2) there is a weak refutation for $\psi'$. As we can nondeterministically guess $\psi'$, it follows that the smoothed instance $\psi$ also has a weak refutation.

We note that technically speaking, the smoothed nondeterministic refutation algorithm $V$ is different than the $V$ for the random/semirandom settings, as it has the additional step of guessing $\psi'$. However, we can use the $V$ for the smoothed case also in the random/semirandom settings, by simply guessing $\psi' = \psi$. 


\bibliographystyle{alpha}
\bibliography{bib/custom,bib/dblp,bib/mathreview,bib/scholar,bib/references,bib/witmer}

\appendix
\section{Analyzing the \cite{WeinAM19} Approach for Random $3$-XOR} \label{sec:wam-sug-does-not-work}
In this section, we will prove that the approach suggested by~\cite{WeinAM19} (in their Appendix F.1, F.2) for strongly refuting random $k$-XOR with $k$ odd does not yield the right trade-off for $m$ as a function of $n,\ell$. Our proof reduces to showing that a certain matrix defined in~\cite{WeinAM19} does not have small spectral norm.
For simplicity, we present the argument for $k = 3$.

First, we give a brief overview of their approach. Let $\phi$ be a random $3$-XOR instance in $n$ variables and $m$ clauses, with hypergraph $\cH$ and coefficients $\{b_C\}_{C \in \cH}$. We will assume that each pair $C_1 \ne C_2 \in \cH$ has $\abs{C_1 \cap C_2} \leq 1$; this ``morally'' holds with high probability provided that $m \ll n^2$ (and recall that we are working in the regime of $m \sim n^{1.5}$ or smaller, as for $m \gg n^{1.5}$ there is a polynomial-time refutation~\cite{AbascalGK21}). More formally, when $m \ll n^2$, then with high probability over $\cH$, one can remove $o(m)$ constraints from $\cH$ so that the remaining hypergraph satisfies this condition.

The construction of~\cite{WeinAM19} is as follows. First, partition the hyperedges $\cH$ arbitrarily into $\cH_1, \dots, \cH_n$, such that if $C \in \cH_u$ then $u \in C$. From now on, we shall think of $\cH$ as $\cup_{u = 1}^n \cH_u$. We note that our lower bound will hold regardless of the choice of the partition here.

Next, let $\phi$ be the polynomial $\phi(x) := \frac{1}{m} \sum_{C \in \cH} b_C x_C$, where $x_C := \prod_{i \in C} x_i$. Applying the Cauchy-Schwarz inequality, we have that
\begin{flalign*}
\phi(x)^2 \leq \frac{1}{m}\sum_{u = 1}^n x_u^2 + \frac{n}{m^2} \sum_{u = 1}^n \sum_{C \ne C' \in \cH_u} b_{C} b_{C'} x_{C \setminus \{u\}} x_{C' \setminus \{u\}} = \frac{n}{m} + f(x)  \mcom
\end{flalign*}
where $f(x) :=  \frac{n}{m^2}\sum_{u = 1}^n \sum_{C \ne C' \in \cH_u} b_{C} b_{C'} x_{C \setminus \{u\}} x_{C' \setminus \{u\}}$.

We now recall the following definition from~\cite{WeinAM19}.
\begin{definition}
Let $\ell \in \N$, and let $\cH = \cup_{u = 1}^n \cH_u$ be a $3$-uniform hypergraph.
For $\vec{S},\vec{T} \in [n]^{\ell}$ and $C_1 = \{u,v_1,w_1\}, C_2 = \{u,v_2,w_2\} \in \cH_u$ with $\{v_1,w_1\} \cap \{v_2, w_2\} = \emptyset$, we write $\vec{S} \overset{C_1,C_2}{\leftrightarrow} \vec{T}$ if there exist $i \ne j \in [\ell]$ such that \begin{inparaenum}[(1)] \item $\vec{S}_t = \vec{T}_t$ for all $t \ne i,j$, and \item $\{\vec{S}_{i}, \vec{S}_j\}$ contains exactly one element from each of $\{v_1,w_1\}$ and $\{v_2,w_2\}$, and $\{\vec{T}_i, \vec{T}_j\}$ contains the other two remaining elements\end{inparaenum}. Here, $\vec{S}_i$ denotes the $i$-th element in the tuple $\vec{S} \in [n]^{\ell}$. We note that if $\vec{S} \overset{C_1,C_2}{\leftrightarrow} \vec{T}$ for some $C_1, C_2$, then we cannot have $\vec{S} \overset{C'_1,C'_2}{\leftrightarrow} \vec{T}$ for any other pair $C'_1, C'_2$.

Let $A_u \in \R^{n^{\ell} \times n^{\ell}}$ be the matrix where $A_u(\vec{S},\vec{T}) = b_{C_1} b_{C_2}$ if $\vec{S} \overset{C_1,C_2}{\leftrightarrow} \vec{T}$ for some $C_1 \ne C_2 \in \cH_u$, and $0$ otherwise, and let $A := \sum_{u = 1}^n A_u$.
\end{definition}
It is simple to observe that $\max_{x \in \on^n} f(x) \leq \frac{n}{m^2} \cdot O(\frac{n^2}{\ell^2}) \norm{A}_2$, as $\frac{m^2}{n} f(x) = \frac{1}{4{\ell \choose 2}(n-4)^{\ell-2}} (x^{\otimes \ell})^{\top} A x^{\otimes \ell}$ for all $x \in \on^n$ because each pair $C_1 \ne C_2 \in \cH_u$ ``appears'' exactly $4 {\ell \choose 2} (n-4)^{\ell - 2}$ times in the matrix $A$. Thus, in order to get the correct $m = n^{1.5}/\sqrt{\ell}$ trade-off, we need to show that $\norm{A}_2 \leq O(\ell)$, with high probability over $\cH$ and the $b_C$'s.

We prove that $\norm{A}_2$ is in fact \emph{large} with high probability, and so the above approach of~\cite{WeinAM19} fails. Formally, we prove that with high probability, the matrix $A$ has a spectral norm $\Omega(\min(\ell^2, \frac{m^2}{n^2}))$, which has the following implications. If the minimum is $\frac{m^2}{n^2}$, then the upper bound certified on $f$ is $\Omega(n/\ell^2)$, and thus the upper bound certified on $\phi$ is $\Omega(\sqrt{n}/\ell)$. This is not very useful, as it is greater than $1$ when $\ell \ll \sqrt{n}$. If the minimum is $\ell^2$, then we certify a good upper bound on $f$ (and therefore also $\phi$) only if $m \geq n^{1.5}$, which is higher than the desired threshold of $n^{1.5}/\sqrt{\ell}$.
\begin{proposition}
\label{prop:wamconjfalse}
Let $\phi$ be a $3$-XOR instance with $n$ variables and $m$ constraints, with constraint hypergraph $\cH = \cup_{u = 1}^n \cH_u$ and coefficients $\{b_C\}_{C \in \cH}$. Suppose that $2 n \leq m$, and that for every pair of constraints $C_1 \ne C_2 \in \cH$, it holds that $\abs{C_1 \cap C_2} \leq 1$.
Let $\ell \leq n$. Then, $\norm{A}_2 \geq {\ell' \choose 2}$, where $\ell' := \min(\ceil{\frac{m}{2n}}, \ell)$.
\end{proposition}
We note that the \cref{prop:wamconjfalse} holds regardless of the choice of the partitioning of $\cH$ into the $\cH_u$'s, and also for any choice of the $b_{C}$'s (and so, in particular, for random $b_C$'s). We also note that \cref{prop:wamconjfalse} essentially holds for a random $\cH$, provided that $m \ll n^2$, for the same reason mentioned earlier: when $m \ll n^2$, with high probability over $\cH$, after removing $o(m)$ constraints from $\cH$, the resulting hypergraph $\cH'$ satisfies $\abs{C_1 \cap C_2} \leq 1$ for all $C_1 \ne C_2 \in \cH'$.
\begin{proof}
As $m \geq 2n$, there must exist some variable $u \in [n]$ that appears in at least $\frac{m}{n}$ constraints. Hence, there must exist at least $\ceil{\frac{m}{2n}}$ constraints that include $u$ and all have the same sign $b \in \on$.

Let $\ell' := \min(\ceil{\frac{m}{2n}}, \ell)$. By the above, we have $\ell'$ constraints $\{C_i\}_{i \in [\ell']} = \{\{u, v_i, w_i\}\}_{i \in [\ell']}$ such that $b_{C_i} = b$ for all $i$. Furthermore, by assumption on $\cH$, we have $\abs{C_i \cap C_j} \leq 1$ for all $i \ne j \in [\ell']$. As $u \in C_i \cap C_j$, it thus follows that $\{v_i, w_i\} \cap \{v_j, w_j\} = \emptyset$. 
Let $z \in [n]$ be arbitrary. Let $\cR$ denote the set of tuples $(r_1, \dots, r_{\ell'}, z, \dots, z) \in [n]^{\ell}$ such that $r_i \in \{v_i, w_i\}$ for all $i \in [\ell']$. We note that the element $z$ merely pads each tuple in $\cR$ to have length exactly $\ell$ when $\ell' < \ell$.

Let $M$ be the submatrix of $A$ indexed by the tuples in $\cR$. Note that $M$ is a $2^{\ell'} \times 2^{\ell'}$ matrix, as $\abs{R} = 2^{\ell'}$. Let $\vec{S} = (r_1, \dots, r_{\ell'}, z, \dots, z)$ be a row in $M$. We will show that each row of $M$ has exactly ${\ell' \choose 2}$ nonzero entries, each of which is $1$.

First, let us consider the contribution to $M$ from $A_u$. Fix a row $\vec{S} \in \cR$. For each pair of indices $i \ne j \in [\ell']$, we can replace the $i$-th and $j$-th elements of $\vec{S}$ with the elements of $\{v_i, w_i\}$ and $\{v_j, w_j\}$ not used in $\vec{S}$, and this will yield some $\vec{T} \in \cR$ with $\vec{S} \overset{\{u, v_i, w_i\}, \{u, v_j, w_j\}}{\leftrightarrow} \vec{T}$. Hence, $A_u(\vec{S}, \vec{T}) = b^2 = 1$. Any other $\vec{T} \in \cR$ will differ from $\vec{S}$ by at least $2$ elements, and thus we must have $A_u(\vec{S}, \vec{T}) = 0$ for such $\vec{T}$.

Next, let us consider the contribution to $M$ from $A_{u'}$ for $u' \ne u$. Fix a row $\vec{S} \in \cR$. It suffices to only consider $\vec{T}$ obtained by swapping the $i$-th and $j$-th entries of $\vec{S}$, for some $i \ne j \in [\ell']$, as above. If $A_{u'}(\vec{S}, \vec{T})$ is nonzero, then we must have $\vec{S} \overset{\{u', v_i, w_i\}, \{u', v_j, w_j\}}{\leftrightarrow} \vec{T}$, and thus that $\{u', v_i, w_i\}, \{u', v_j, w_j\} \in \cH_{u'}$. However, this implies that $\abs{\{u, v_i, w_i\}, \{u', v_i, w_i\}} = 2 > 1$, which contradicts our assumption on $\cH$.

We have thus shown that the matrix $M$ is $2^{\ell'} \times 2^{\ell'}$, with each row having exactly ${\ell' \choose 2}$ nonzero entries, all of which are $1$. It thus follows that $\norm{A}_2 \geq \norm{M}_2 \geq (1^{2^{\ell'}})^{\top} M 1^{2^{\ell'}} / 2^{\ell'} = {\ell' \choose 2}$, which finishes the proof.
\end{proof}

\end{document}